\documentclass[reqno]{amsart}
\usepackage[dvipsnames]{xcolor}
\usepackage{amsmath}
\usepackage{graphicx, color}
\usepackage{chngcntr}
\usepackage{tikz}
\usetikzlibrary{backgrounds}
\usepackage{amsmath,amsfonts, amssymb}
\usepackage[normalem]{ulem}
\usepackage{wrapfig}
\usepackage{chngcntr}
\usepackage{paralist}
\usepackage{graphics} 
  \usepackage{epsfig} 
 \usepackage[colorlinks=true]{hyperref}
 \usepackage{graphicx}  \usepackage{epstopdf}
 \usepackage{xcolor}
 
 \usepackage{tikz-3dplot}
 \usepackage{tikz}
 \usepackage{pgfplots}
 \usepackage{bm}
 \usepgfplotslibrary{fillbetween}
 \usetikzlibrary{backgrounds}
 \usetikzlibrary{patterns}
 \usetikzlibrary{scopes}
 \usetikzlibrary{shapes,arrows}
 \usetikzlibrary{plotmarks,intersections}
 \usetikzlibrary{calc}
 \usetikzlibrary{positioning}
 
 \usepackage{float}
 
 \usepackage[toc,page]{appendix}

 \usepackage{multirow}
 \usepackage{array}
 
\usetikzlibrary{calc}
\usepackage{pgfplots}
 
 \usepackage{bbm}
\hypersetup{urlcolor=blue, citecolor=red}

\newtheorem{theorem}{Theorem}[section]
\newtheorem{corollary}{Corollary}

\newtheorem{lemma}[theorem]{Lemma}
\newtheorem{proposition}{Proposition}

\newtheorem{definition}[theorem]{Definition}
\newtheorem{remark}{Remark}

\newtheorem{assum}[theorem]{Assumption}

\numberwithin{equation}{section}

\newcommand\drawline[1][black]{%
	\begin{tikzpicture}                                                           
	\draw[#1] (0pt,0pt) -- (15pt,0pt);                                       
	\end{tikzpicture}%
}      

\title[]
      {On the Cauchy problem for   Boltzmann equation modeling a polyatomic gas }

\author[Irene M. Gamba and Milana Pavi\'c-\v Coli\'{c}]{}

\newcommand{\nocontentsline}[3]{}
\newcommand{\tocless}[2]{\bgroup\let\addcontentsline=\nocontentsline#1{#2}\egroup}

\newcommand{\co}{\color{BurntOrange}}

\newcommand{\pmk}{\mathfrak{m}_k}

\newcommand{\dgu}{d_{\gamma}^{{ub}}(r)}
\newcommand{\dgl}{d_{\gamma}^{{lb}}(r)}
\newcommand{\egu}{e_{\gamma}^{{ub}}(R)}
\newcommand{\egl}{e_{\gamma}^{{lb}}(R)}

\newcommand{\dgun}{d_{\gamma}^{{ub}}}

\newcommand{\cgaur}{c_{\gamma,\alpha}^{{ub}}}

\newcommand{\CgauR}{C_{\gamma,\alpha}^{{ub}}}
\newcommand{\cgalr}{c_{\gamma,\alpha}^{{lb}}}
\newcommand{\CgalR}{C_{\gamma,\alpha}^{{lb}}}
\newcommand{\clb}{c_{lb}}

\newcommand{\tf}{ \left| v - v_* \right|^\gamma + \left( \frac{I+I_*}{m} \right)^{\gamma/2} }
\newcommand{\tfu}{ \left| u \right|^\gamma + \left( \frac{I+I_*}{m} \right)^{\gamma/2}  }
\newcommand{\tB}{  \tilde{\mathcal{B}}}

\newcommand{\Cg}{C_{k_*}}

\newcommand{\bks}{ \bar{k}_*}
\newcommand{\ks}{ k_*}

\newcommand{\tk}{ 2k }

\newcommand{\hk}{  k}

\newcommand{\hg}{ \gamma/2}

\newcommand{\ho}{ 1}

\newcommand{\K}{K}

\newcommand{\pl}{\delta}

\newcommand{\Ml}{  \mathcal{M}_l} 
\newcommand{\Mu}{  \mathcal{M}_u}
\newcommand{\El}{ \mathcal{E}_l}
\newcommand{\Eu}{\mathcal{E}_u}
\newcommand{\D}{\Delta}

\newcommand{\Aks}{  { \bar{A}}_{k_*}} 

\newcommand{\B}{\bar{B}}

\begin{document}
\maketitle


\centerline{\scshape Irene M. Gamba}
\medskip
{\footnotesize
	\centerline{Department of Mathematics and Oden Institute of Computational Engineering and Sciences}
	\centerline{University of Texas at Austin}
	\centerline{2515 Speedway Stop C1200
		Austin, Texas, 78712-1202, USA}

}

\medskip

\centerline{\scshape Milana Pavi\'c-\v Coli\'{c}}
\medskip
{\footnotesize
  \centerline{Department of Mathematics and Informatics}
  \centerline{Faculty of Sciences, University of Novi Sad}
 \centerline{Trg Dositeja Obradovi\'ca 4, 21000 Novi Sad, Serbia}
}

\bigskip


\bigskip


\begin{abstract}
In the present manuscript we consider the Boltzmann equation that models a polyatomic gas by introducing one additional continuous variable, referred to as microscopic internal energy. We establish existence and uniqueness theory in the space homogeneous setting for the full non-linear case, under an extended Grad assumption on transition probability rate, that comprises hard potentials  for both the relative speed and internal energy with the rate in the interval $(0,2]$, which  is multiplied by  an integrable angular part and integrable partition functions.   The Cauchy problem is resolved by means of an abstract ODE theory in Banach spaces, for an initial data with finite and strictly positive gas  mass and energy, finite momentum, and additionally finite $k_*$ polynomial moment, with $k_*$ depending on the rate of the transition probability and the structure of a polyatomic molecule or its internal degrees of freedom. Moreover, we prove that polynomially and exponentially weighted Banach space  norms associated to the solution are both generated and propagated  uniformly in time.
\end{abstract}

\tableofcontents

\section{Introduction}

In this manuscript we consider a single polyatomic gas. The more complex structure of a molecule that may have more than one atom causes new phenomena at the level of molecular collisions. In particular, besides the translational motion in the physical space as in the classical case of monatomic elastic collisions, there appear possibilities of molecular rotations and/or vibrations, referred to as internal degrees of freedom. Collisional kinetic theory captures this feature by introducing  the so-called microscopic internal energy of a molecule. Then, an elastic collision of polyatomic gases means conservation of the total -- kinetic plus internal -- energy of the two colliding molecules if  binary interactions are taken into account. \\

Within the kinetic theory there is no unique way how to approach  the microscopic internal energy,  which can be understood  as a measure of deviation from the classical case of a single monatomic gas. For example, in the semi-classical approach \cite{Cha-Cow, Gio, W-C, Della}, internal energy is assumed to take discrete values. The idea is to prescribe  one distribution function  to each energy level, resulting in a system of equations describing a gas. This model uses experimental data and is adequate for computational tasks. On the other hand,  there exist  continuous kinetic models that   take the  another path -- the idea is to introduce  one additional variable, a continuous microscopic internal energy, and to parametrize both molecular velocity and internal energy, using the so-called Borgnakke-Larsen procedure, which leads to  one single Boltzmann equation \cite{LD-Bourgat94, DesMonSalv, LD-Toulouse}. Among continuous models, there are subtle differences causing by the choice of a functional space as an environment where physical intuition is provided, which are reflected on the structure of the cross-section, as pointed out in \cite{MPC-Dj-S}. \\

Both semi-classical  and continuous   models abound with many formal results. For instance, the Champan-Enskog method was developed in  \cite{Alexeev} for the semi-classical models and recently in \cite{LD-Bisi} for the continuous model from \cite{DesMonSalv}. Many macroscopic models of extended thermodynamics are derived starting from the continuous kinetic model \cite{DesMonSalv}. In fact, the additional variable of internal energy fitted naturally into the two hierarchies of moment equations for a polyatomic gas, as first observed in \cite{MPC-Rugg-Sim, MPC-Sim}. The maximum entropy principle was the main tool to close system of  equations corresponding to six and fourteen moments \cite{Rugger-MEP, Rugg-Bisi-6, MPC-Madj-Sim, MPC-Madj-Sim-Parma}, and to numerically test those models \cite{Aoki}. An interested reader can consult \cite{Rugg-Poly} and references therein. In addition, the formal results spread to the mixture of polyatomic gases \cite{MPC, Bisi-multi-1,Bisi-multi-2}, that may be reactive as well \cite{Soares-Bisi, Soares-Salvarani}.\\ 

Despite the great interest of research community, so far there are  few rigorous  mathematical  results addressing  analytical properties  of polyatomic gas  kinetic model. A  basic question is what are the  differential cross section and transition probability rates for which a solution is admissible in a suitable functional space depending on the nature of the polyatomic Boltzmann flow operator defined by their  particle molecular velocity and exchange of internal energy laws, or whether the   initial value problem in such spaces  is solvable, if solutions are  well defined  globally in time, and what is their high energy tail behavior.\\

 In this manuscript we  first establish the existence and uniqueness theory for the solution of the space homogeneous Boltzmann equation  for a polyatomic gas continuous  kinetic model  introduced in \cite{LD-Bourgat94}. The underlying assumption on the transition probability rate is of extended Grad type, meaning that besides the positive power  of relative speed, we need the very same contribution of the internal energy. Moreover, the relative speed and internal energy are combined additively, and not multiplicatively as it was  used in literature so far. Surprisingly enough, such a model of transition function perfectly embeds into physical interpretation, providing the total agreement with the models of extended thermodynamics  and experimental data, as shown in \cite{MPC-Dj-S}.\\

The existence and uniqueness result is obtained by an application of the theory of ODEs in Banach spaces \cite{Martin}, that has been successfully used in many frameworks, such as mixture setting \cite{IG-P-C}, polymers kinetic problems \cite{AlonsoLods18},  quantum Boltzmann equation for bosons in very low temperature \cite{IG-Alonso-Tran-pre} and the weak wave turbulence models for stratified flows \cite{IG-Smith-Tran}. 
\\

   In the present manuscript   we propose to study the scalar  Boltzmann flow  for interacting  polyatomic gases interchanging pairs of pre and post molecular velocities  $v \in \mathbb{R}^3$ and internal energies $I\in [0,\infty)$ as one additional continuous variable. This collisional model   describes the statistical time evolution of  probability distribution density $f(t,v,I)$ in the Banach space  $L^1_k(\mathbb R^{3^+})$, of integrable functions in the upper half space $\mathbb R^{3+}:=\mathbb{R}^{3}\times [0,\infty)$  with the  Lebesgue weight type function $\langle v,I\rangle^k:=(1+ \frac{|v|^2}{2} + \frac{I}{m})^{k/2}$, 
	for  molecular velocities $v\in \mathbb R^3$ and internal energies $I \in  [0,\infty)$.  Consequently,  $\|f\|_{L^1_k(\mathbb R^{3^+})}(t)$ are also referred as the $k-$Lebesgue moments associated to the solution of the Boltzmann flow.\\

   A keypoint  in this analysis consists in   showing that that the dissipation  built in the collision operator for polyatomic gasses is manifested   by the  decay of   $k$-polynomial moment  of the collisional form for a sufficiently large $k$, depending on  data with finite initial mass, energy and  a moment of order bigger or equal that $k_*$, with $\ks>1$, since, in our notation $k=1$ is  macroscopic mass plus the  total (kinetic + internal) energy.  For this collisional gas model, such property is warranted     by the control, from above and below, of transition rates associated to the velocity and internal energy interactions laws.  This proposed approach requires detailed  averaging over parameters  in a compact manifold that distributes scattering mechanisms as functions of the scattering angle over the sphere of influence of the interaction, as much as the parameters, or partition functions, that distribute total  energy to molecular variables. 

 In the classical case of a single monatomic gas, this result is known as sharp Povzner Lemma by angular averaging over the sphere, introduced for the first time in \cite{Bob97} for hard spheres in three dimensions and extended in \cite{GambaPanfVil09}   to variable hard potentials in velocity dimension bigger or equal than  three, and used  in many results for kinetic collisional binary transport in inelastic interactions theories, such as granular flow \cite{GambaBobPanf04}, and gas mixtures \cite{IG-P-C}.   For this classical single monatomic gas describing binary  elastic interactions of molecular velocities $v\in\mathbb R^{d}$, dissipation has   an immediate effect:   the evolution of the positive contribution  of the Boltzmann operator $k$-Lebesgue moments,  or equivalently classical $k$-moments, occurs  for any $k>1$,  where the  macroscopic gas mass plus energy corresponds  to $k=1$ that  is a conserved quantity, since the local mass and energies  associated  the gain and loss operator balance to zero. 
 	Another example of this behavior was recently shown by the authors in the gas mixture  setting \cite{IG-P-C} corresponding to a system of Boltzmann equations with disparate masses,  with their  corresponding  energy identity showing that the positive contribution of $k$-moments coming from  each  pairs of gain operators  associated to the system  yields an analog dissipative effect, albeit for $k>k_*$, with $k_*$ depending on the  mass species ratios when taken by pairwise interactions, and shown  $k_*$ to  grow with the disparateness of molecular masses.

This averaging on  compact manifolds ambient spaces, where  transition rate functions are defined, only involve  moments of the positive contribution to the Boltzmann flow dynamics, that translates into a property for the moments associated to the gain operator. Namely, such compact manifold averaging  property produces a dissipative mechanism for a large enough moment $k$, depending on  the potential rate of  transition probabilities and  on molecules internal modes related to the complexity of molecular  structure.

Once the dissipation  of the  collision operator gain part is shown, it remains to treat its loss term.   In fact, it is a crucial aspect of our analysis  to find a bound from below for the negative contribution of the polyatomic Boltzmann flow, which supports the coerciveness estimate in the natural Banach space $L^1_k(\mathbb R^{3^+})$ associated to polyatomic gas model.  This result may be understood as  the analogue to coerciveness property in the classical theory for diffusion type equations in continuum mechanics where the control of the energies or suitable Banach norms is done in Sobolev spaces, while in the framework of statistical mechanics the suitable norms are non-reflexive Banach spaces, as in this case $L^1_k(\mathbb R^{3^+})$.   The {\em coercive   factor }  ${A}_{\ks}$  which is the strictly positive constant  associated to be scalar Boltzmann flow of the  order $k_*$  at which the positive contribution of Lebesgue moments becomes submissive is identified by  the smallest positive constant   ${A}_{\ks}$ to be characterized   in Section~\ref{Sec: poly mom}, which depends on  the averaging manifold Lemma  associated to the  potential rate of the transition probabilities,  on internal modes of molecules and  on the initial data. One can  view ${A}_{\ks}$ as the analog to the role of coercive form associated to elliptic  and parabolic flows in continuum mechanics  modelling  where coerciveness is crucial for the existence and uniqueness theories in Sobolev spaces. The coerciveness estimate is based on a   functional inequality that controls  from below  the convolution of any function $f \in L_{1^+}^1$ with a potential function of rate $\gamma$  by the corresponding Lebesgue bracket of order $\gamma$, for any $\gamma\in [0,2]$. The proof is inspired  by the work \cite{Alonso-IG-BAMS} for the classical Boltzmann equation. In addition, the coercive estimates we present here are fundamental to obtain in $W^{m,2}$ Sobolev (Hilbert) spaces as  shown in \cite{IG-A-M}.\\  

The decay of the $k-$Lebesgue moment of the Boltzmann operator positive contribution with respect to $k>k_*$ and the coerciveness estimate for the negative contribution allows for  
a priori  $k^{th}$-moment estimates  that are sufficient to generate infinitely many Ordinary Differential Inequalities (ODIs) with a negative superlinear term  proportional to ${A}_{\ks}$.  That is a sufficient condition for the Boltzmann flow under consideration in suitable invariant region $\Omega$ of Banach space 
$ {L^1_k(\mathbb R^{3^+})}$   to be solvable globally in time.  
   
  We emphasize that conditions on initial data  excludes singular measures and the need  of entropy bounds. Yet the resulting construction of a unique solution in the space on polynomial moments yields entropy boundedness at any time, if initially so. \\
  
  Our analytical approach focus on studying  norms in the Banach space $L^1(\mathbb{R}^{3+})$ with $ \mathbb{R}^{3+}:= \mathbb{R}^{3}\times \mathbb{R}_+$ associated to the solution of the   time evolution problem for Boltzmann equation,   in the space probability density functions defined over  the classical metric space $\mathbb{R}^{3+}$, with both  polynomial and exponential weights, by following  
  the usual analytical path established for     the classical Boltzmann equation for the single elastic monatomic gas model.  This research   started with \cite{Des93, wennberg97}    in the case of  polynomial moments and with \cite{Bob97}   where the concept of exponential moments is presented, as much as in the techniques for moments summability,  leading to the understanding   high energy tail  behavior of with inverse Gaussian in velocity space, associated to the solution of the Boltzmann equation  for hard spheres (i.e. power exponent  $\gamma=1$) and constant angular part. These results were developed  in \cite{GambaBobPanf04} for inelastic interactions  and non-Gaussian weight moments, and later, in 	\cite{GambaPanfVil09} to collision kernels for  hard potentials (i.e. $\gamma\in (0,1]$) for any angular section	with $L^{1+}$-integrability. Further, generation of exponential moments of order $\gamma/2$ with bounded angular section were shown in \cite{Mouhot06}. New approach was taken in \cite{Gamba13} based on partial sums summability techniques, that extended the results to collision kernels for hard potentials  with $\gamma\in(0,2]$, for any angular section,	with just $L^{1}$-integrability. In particular, exponential moments of order $\gamma$  are shown to generate in a finite time, while Gaussian moments  propagate  if initially are finite, which holds	independently of $\gamma$. Moreover, all these result were generalized when the angular part is not integrable (in the angular non-cutoff regime)  \cite{GambaTask18, LuMouhot12, BobGamba17, PT, Alonso}.\\

The manuscript is organized as follows. First we introduce a kinetic model describing a polyatomic gas in Section  \ref{Sec: kin mod poly}, together with the notation and main definitions. Then in Section \ref{Sec: Suff}  we make precise sufficient properties for establishing existence and uniqueness theory, which comprises  assumption on the form of transition function as the additive form of relative speed and microscopic internal energy with a potential $\gamma \in (0,2]$ multiplied by some factors, together with its upper and lower bounds.  Then in Section  \ref{Coerciveness} we prove the Coerciveness Estimate for the loss operator, and in Section 
\ref{Sec: fund lemmas} we state and prove the  two fundamental lemmas, namely the Energy Identity Decomposition and the Polyatomic Compact Manifold Averaging Lemma. These estimates identify the ${\ks}$-moment  that will yield the coercive  constant ${A}_{\ks}$.    Section \ref{Sec: poly mom}  deals with the statements and proofs to a priori estimates for $k^{th}$-moments of any order $k\ge {\ks}$ and defines the explicit form of ${A}_{\ks}$.  
These results enable us to identify an invariant region in which the collision operator will satisfy all the properties needed for existence and uniqueness result proved  theory in Section \ref{Section Ex Uni proof} by means of solving an time evolution  ODE in a suitable invariant region $\Omega$ in the Banach  space  $L^1(\mathbb R^{3+})$. Then, in Section \ref{Section gen prop exp mom} we show the solution of the Boltzmann equation for polyatomic gases has a property of summability of moments that is expressed, both,  in  the generation and propagation of  exponential moments. In the propagation case, the unique  solution of the Boltzmann flow for polyatomic gases  initial data having an exponential moment of prescribed order  $2s$,  $s\in (0,1]$ and rate $\beta_0$, there exists a smaller rate $\beta$ such that the exponential moment of order $\beta$ and same rate $s$ is bounded uniformly in time (i.e. propagation holds), with such bound being three times the initial exponential moment. The rate $\beta$ satisfies a minimum of three conditions, one of them degenerates as the coercive constant $A_{\ks}$ and with a  the maximum of an estimate of a few finite number of $k-$Lebesgue bounds for the moments propagation estimates. 
The case of generation of exponential,  or summability of,  moments  holds for initial data  with  just $k_*-$Lebesgue moments. In this case, the order is $2s$, with $s\in (0, \gamma/2]$ and the rate  $\beta=\beta(t)$, for positive time,   also depends  on the  coercive factor $A_{\ks}$ and with a  the maximum of an estimate of a few finite number of $k-$Lebesgue bounds for the moments generation estimates as much as on the $k_*-$Lebesgue moments of the initial data.  The generation of the exponential moment bound is achieved in short time, and remains uniformly bounded in time by the initial's data $k_*-$Lebesgue moment, exhibiting an invariant region estimate. 
  Finally, the Appendix contains some technical results needed for the theory.

\section{Kinetic model for a polyatomic gas}\label{Sec: kin mod poly}

In this Section we will describe the Boltzmann equation for a polyatomic gas. We adapt the \emph{continuous} approach, which  introduces a \emph{single continuous} variable $I$ that we call \emph{microscopic internal energy}, supposed to capture all the phenomena related to a more complex structure of a polyatomic molecule. The main feature is the presence of  internal degrees of freedom that a  molecule undertakes on an interaction, or collision. Besides the usual translational motion, a polyatomic molecule may experience  rotations and/or vibrations, referred to as internal modes. 

At the microscopic level of collisions, such motions cause appearance of the microscopic internal energy, apart from the usual kinetic energy in the conservation of energy law, under the assumption of elastic collisions. On the other side, at the macroscopic level, internal modes reflect on the energy law as well. As in this manuscript, we restrict to polytropic gases (meaning that macroscopic internal energy is linear with respect to the temperature), the caloric equation of state reads
\begin{equation}\label{caloric}
e = D \,  \frac{k \, T}{2\, m},
\end{equation}
where $e$ is the internal energy, $k$ the Boltzmann constant, $m$ mass and $T$ temperature of the gas. The constant $D$ is related to the degrees of freedom.
 In the classical case for elastic interactions, only translation is taken into account, corresponding to $D$ taking the value of the space dimension and the kinetic collisional model of the classical Boltzmann equation.
In general  $D$ is determined by the sum of the total degrees of freedom,   translational as much as  rotational and vibrational  motion associated to the collision.  That means, in space dimension three, this constant takes at least  the value $D=3$, which is the classical case for monatomic gases modeled by the scalar Boltzmann equation, but for the polyatomic model  the constant $D$ must be larger than the dimension of the space of motion, $D>3$.

\subsection{Collision modelling}\label{Sec coll model}
 The starting point is to model a collision process between two interacting molecules. We suppose that a colliding pair of molecules have velocities and microscopic internal energies $(v', I')$, $(v'_*, I'_*)$ $\in  \mathbb{R}^{3+}:=\mathbb{R}^3\times [0,\infty)$ before the interaction, that became  $(v, I)$ and $(v_*, I_*)$, respectively, after such interaction. Under the assumption of elastic interactions,  these quantities are linked through the conservation laws of local momentum and total (kinetic + microscopic internal) molecular energy, namely,
\begin{equation}\label{micro CL}
\begin{split}
 v +  v_* &=  v' + v'_*,\\
\frac{m}{2} \left|v\right|^2 + I + \frac{m}{2} \left|v_*\right|^2 + I_* &= \frac{m}{2} \left|v'\right|^2 + I' + \frac{m}{2} \left|v'_*\right|^2  + I'_*.
\end{split}
\end{equation}
It is often more convenient to work in  the center of mass reference frame by the introduction of  center of mass $V$ and relative velocity $u$ as follows
\begin{equation}\label{cm-rv}
V:= \frac{v+v_*}{2}, \quad u:=v-v_*.
\end{equation}
 Then, the  total molecular energy law from \eqref{micro CL}  can be simply written by  
\begin{equation}\label{micro CL energy mass-rel vel}
\frac{m}{4} \left|u\right|^2 + I + I_* = \frac{m}{4} \left|u'\right|^2 + I' + I'_*=:E.
\end{equation}
since  clearly conservation of local momentum  implies conservation of center of mass velocity 
\begin{equation}\label{CL V}
V=V'.
\end{equation}	

 Collisional laws   express pre-collisional quantities $v', I', v'_*, I'_*$ in terms of the post-collisional ones. This is achieved via a parametrization of  local conservation equations \eqref{micro CL}, according to the Borgnakke-Larsen procedure. To this end, we focus on energy \eqref{micro CL energy mass-rel vel} and first  introduce  a parameter $R\in[0,1]$ that distributes local energy proportion of the  total energy $E$ into a pure kinetic part $RE$ and   a pure internal part of energy, proportional to $(1-R)E$,  according to 
\begin{equation*}
\frac{m}{4} \left|u'\right|^2 = R E, \quad  I' + I'_* =  (1-R) E.
\end{equation*}

In addition, we set a parameter $r\in [0,1]$ to distribute  the proportion  of total internal energy $(1-R) E$ to each interacting states corresponding to the incoming molecular internal energy pair   $I', I'_*$  as follows
 \begin{equation}\label{micro int energies}
I'= r (1-R) E, \qquad I'_* = (1-r) (1-R) E.
 \end{equation}
Finally, we introduce    the classical scattering direction associated to classical collisional elastic theory,   $ \sigma \in S^2$, in order to parametrize pre-collisional relative molecular velocity $u'$,
\begin{equation}\label{sigma-scatter}
u' = \left| u' \right| \sigma = 2 \sqrt{\frac{R E}{m}} \sigma.
\end{equation}
 We note that this relation holds for a classical monatomic single species model in the absence of internal energy modes for which $\left| u' \right| = \left| u \right|$.

This representation introduces  the fundamental  set of coordinates in center of mass  and the {\em pure} kinetic energy.
The last equation together with moment conservation law from \eqref{micro CL} yields expressions for velocities,
\begin{equation}\label{velocities}
v' = V + \sqrt{\frac{R E}{m}} \sigma, \quad v'_* = V - \sqrt{\frac{R E}{m}}\sigma.
\end{equation}

\begin{figure}
	\begin{center}
	\tdplotsetmaincoords{70}{110}
\begin{tikzpicture}[tdplot_main_coords,rotate around z=-40,scale=1.6]
		\begin{scope}[on background layer]
		\fill [fill=gray!10]  (4.5,0,0)  -- (0,5.5,0) -- (-4.5,0,0) --(0,-4.5,0) ;
		\end{scope}
			\begin{scope}
		\draw [->,gray!10] (0,0,0) -- (4.5,0,0) node [pos=0.75, above=-2,rotate=52] {\color{black}{velocity space}};
		\draw [-,gray!10] (0,0,0) -- (-4.5,0,0) node [right] {};
		\draw [->,gray!10] (0,0,0) -- (0,4.5,0) node [left] {};
		\draw [-,gray!10] (0,0,0) -- (0,-4.5,0) node [left] {};
		\draw [->,gray] (4.5,0,0) -- (4.5,0,4) node [pos=1, right,rotate=0] {\color{black}{internal energy space}};
		\node (vI) at (1.73205,1,2) [circle,scale=1.2, shade,ball color=white!70!gray] {};
			\node (vI1) at (1.6,0.9,2) [circle,scale=0.8, shade,ball color=white!70!gray] {};
				\node (vI2) at (1.8,0.9,2) [circle,scale=0.8, shade,ball color=white!70!gray] {};
		\node (vsIs) at (-1.73205,1,2) [circle,scale=1.2, shade,ball color=white!70!gray] {};
			\node (vsIs1) at (-1.7,0.9,2) [circle,scale=0.8, shade,ball color=white!70!gray] {};
				\node (vsIs2) at (-1.8,1.1,2.05) [circle,scale=0.8, shade,ball color=white!70!gray] {};
		\node (vpIp) at (1.67332, -0.67332, 0.84)  [circle,scale=1.2, shade,ball color=Purple] {};
				\node (vpIp11) at (1.6, -0.7, 0.7)  [circle,scale=0.8, shade,ball color=Purple] {};
					\node (vpIp12) at (1.8, -0.69, 0.8)  [circle,scale=0.8, shade,ball color=Purple] {};
		\node (vpsIps) at (-1.67332, 2.67332, 0.56) [circle,scale=1.2, shade,ball color=Purple] {};
				\node (vpsIps1) at (-1.7, 2.8, 0.65) [circle,scale=0.8, shade,ball color=Purple] {};
						\node (vpsIps2) at (-1.65, 2.82, 0.5) [circle,scale=0.8, shade,ball color=Purple] {};
		\node (vpIp2) at (0.591608, 0.408392, 3.78)  [circle,scale=1.2, shade,ball color=RoyalBlue] {};
			\node (vpIp21) at (0.6, 0.4, 3.92)  [circle,scale=0.8, shade,ball color=RoyalBlue] {};
				\node (vpIp22) at (0.7, 0.45, 3.8)  [circle,scale=0.8, shade,ball color=RoyalBlue] {};
		\node (vpsIps2) at (-0.591608, 1.59161, 2.52) [circle,scale=1.2, shade,ball color=RoyalBlue] {};
			\node (vpsIps21) at (-0.6, 1.6, 2.65) [circle,scale=0.8, shade,ball color=RoyalBlue] {};
				\node (vpsIps22) at (-0.5, 1.9, 2.59) [circle,scale=0.8, shade,ball color=RoyalBlue] {};
		\draw [->,line width=0.15ex] (1.73205,1,0) -- (-1.73205,1,0) node [pos=1,below] {$u$};
		\node (start) at (0,0,0)  [circle,scale=0.1, gray] {};
			\end{scope}	
		\begin{scope}[>=latex]
		\node[right=-0.3 of vsIs] {$(v_*,I_*)$};
		\node[left=-0.2 of vI] {$(v,I)$};
		\node[ right=-0.1 of vpsIps]  {$(v'_*,I'_*)$};	
		\node[ left=-0.1 of vpIp] {$(v',I')$};
		\node[ left=-0.1 of vpIp2] {$(v',I')$};
		\node[ above right=0.2 of vpsIps2] {$(v'_*,I'_*)$};	
		\end{scope}		
		\begin{scope}[>=latex]
		\draw [-,dotted,line width=0.15ex] (1.73205,1,0) -- (vI) node [left] {};
		\draw [-,dotted,line width=0.15ex] (-1.73205,1,0) -- (vsIs) node [left] {};
		\draw [-,dashdotted,line width=0.15ex,Purple,line width=0.1ex] (1.67332, -0.67332, 0) -- (vpIp) node [left] {};
		\draw [-,dashdotted,line width=0.15ex,Purple] (-1.67332, 2.67332, 0) -- (vpsIps) node [left] {};
		\draw [->,color=Purple, line width=0.2ex]  (1.67332, -0.67332, 0) -- (-1.67332, 2.67332,0) node [pos=1,below ] {$u'$};
			\draw [->,color=RoyalBlue, line width=0.15ex] (0.591608, 0.408392,0) --(-0.591608, 1.59161,0) node [pos=0.9,below right] {\color{RoyalBlue}$u'$};
		\draw [-,dashed,line width=0.15ex,RoyalBlue] (0.591608, 0.408392,0) -- (vpIp2) node [left] {};
		\draw [-,dashed,line width=0.15ex,RoyalBlue] (-0.591608, 1.59161,0) -- (vpsIps2) node [left] {};
		\draw [-,dashed,line width=0.15ex,RoyalBlue] (vpIp2) -- (vpsIps2) node [left] {};
		\draw [-,dashdotted,line width=0.15ex,Purple] (vpIp) -- (vpsIps) node [left] {};
		\draw [-,dotted,line width=0.15ex] (vI) -- (vsIs) node [left] {};
		\draw [-,line width=0.15ex,white!70!gray] (0,1,0) --  (0., 1., 3.15)  node [left] {};
		\end{scope}
		\begin{scope}
	\node at (0., 1., 3.15)  [circle,scale=0.3, shade,ball color=RoyalBlue] {};
\node  at (0., 1., 0.7)  [circle,scale=0.3, shade,ball color=Purple] {};
\node  at (0., 1., 2)  [circle,scale=0.3, shade,ball color=white!70!gray] {};
\node (V) at (0., 1., 0)  [circle,scale=0.3, shade,ball color=white!70!gray] {};
	\node[ below=0.1 of V]  {$V$};	
		\end{scope}		
		\tdplotsetrotatedcoords{0}{0}{50}
			\begin{scope}[tdplot_rotated_coords] 
	\draw[dashdotted,color=Purple,line width=0.1ex] (1,0,0) coordinate (c) circle (2.36643);
	\draw[dotted] (1,0,0) coordinate (c) circle (1.73205);
	\draw[dashed,color=RoyalBlue] (1,0,0) coordinate (c) circle (0.83666);
	\end{scope}	
	\begin{scope}
		\node (vp old) at (1.22474, -0.224745,0)   [circle,scale=1, shade,ball color=green!60!black!50!white] {};
	\node (vps old) at (-1.22474, 2.22474,0)   [circle,scale=1, shade,ball color=green!60!black!50!white] {};
	\node[below left=-0.25cm of vp old] {$v'_m$};
	\node[below right=-0.25cm of vps old] {$v'_{*m}$};
	\end{scope}
		\end{tikzpicture}
	\end{center}
	\caption{Given the two coliding  polyatomic molecules with velocities and internal energies $(v,I)$ and $(v_*, I_*)$ rendered in  gray tone, we prescribe a scattering direction $\sigma$ and partition functions $r$ and $R$. Then states $(v',I')$ and $(v'_*,I'_*)$ can be calculated using the collision transformation \eqref{micro int energies} and \eqref{velocities}. The  violet dash-dotted state  (\drawline[Purple,ultra thick,dashdotted])  corresponds to the choice $R=0.8$ and the  blue dashed one  (\drawline[RoyalBlue,ultra thick,dashed]) to $R=0.1$. The  total microscopic energy $E$ weights a bigger proportion to the internal energy for smaller values of $R \in [0,1]$, while larger values of  $R \in [0,1]$ would give a bigger proportion to the  kinetic energy of molecules. One  can view that the limit $R=1$ shrinks the total energy $E$ to be just the kinetic energy, while the limit $R=0$ renders the total energy as the pure  internal energy. The green state corresponds to the collisions without internal energy, or classical monatomic elastic collision.}
	\label{picture}
\end{figure}
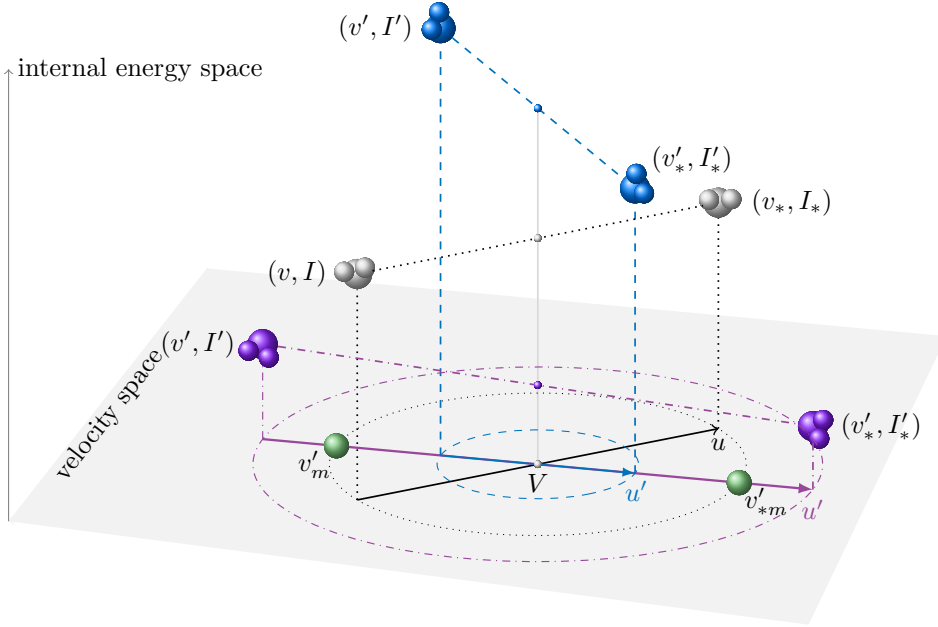

\begin{figure}
	\begin{center}
\begin{tikzpicture}[scale=2]
\node (v) at (1.73205,1) [circle,scale=1.2, shade,ball color=white!50!gray] {};
\node (v1) at (1.7,1.1) [circle,scale=0.8, shade,ball color=white!70!gray] {};
\node (v2) at (1.8,0.9) [circle,scale=0.8, shade,ball color=white!70!gray] {};
\node (vs) at (-1.73205,1) [circle,scale=1.2, shade,ball color=white!70!gray] {};
\node (vs1) at (-1.7,0.9) [circle,scale=0.8, shade,ball color=white!70!gray] {};
\node (vs2) at (-1.75,1.1) [circle,scale=0.8, shade,ball color=white!70!gray] {};
\node (V) at (0,1)  [circle,scale=0.5, shade,ball color=black] {};
\node (vp) at (0.591608, 0.408392)  [circle,scale=1.2, shade,ball color=RoyalBlue] {};
\node (vp1) at (0.56, 0.3)  [circle,scale=0.8, shade,ball color=RoyalBlue] {};
\node (vp2) at (0.65, 0.5)  [circle,scale=0.8, shade,ball color=RoyalBlue] {};
\node (vps) at (-0.591608, 1.59161) [circle,scale=1.2, shade,ball color=RoyalBlue] {};
\node (vps1) at (-0.63, 1.5) [circle,scale=0.8, shade,ball color=RoyalBlue] {};
\node (vps2) at (-0.46, 1.62) [circle,scale=0.8, shade,ball color=RoyalBlue] {};
\node (vpp) at (1.67332, -0.67332)  [circle,scale=1.2, shade,ball color=violet] {};
\node (vpp1) at (1.75, -0.75)  [circle,scale=0.8, shade,ball color=violet] {};
\node (vpp2) at (1.57, -0.7)  [circle,scale=0.8, shade,ball color=violet] {};
\node (vpps) at (-1.67332, 2.67332) [circle,scale=1.2, shade,ball color=violet] {};
\node (vpps1) at (-1.63, 2.57) [circle,scale=0.8, shade,ball color=violet] {};
\node (vpps2) at (-1.55, 2.7) [circle,scale=0.8, shade,ball color=violet] {};
\node (vp old) at (1.22474, -0.224745)   [circle,scale=1, shade,ball color=green!60!black!50!white] {};
\node (vps old) at (-1.22474, 2.22474)   [circle,scale=1, shade,ball color=green!60!black!50!white] {};
\begin{scope}[>=latex, on background layer]
\draw[dotted,line width=0.1ex] (V) circle (1.73205cm);
\draw[dashed,color=RoyalBlue,line width=0.1ex] (V) circle (0.83666cm);
\draw[dashdotted,color=violet,line width=0.1ex] (V) circle (2.36643cm);
\node[right=-0.05cm of v] {$(v,0)$};
\node[left=-0.05cm of vs] {$(v_*,0)$};
\node[left=0.05cm of vps] {\color{RoyalBlue}$(v'_*,0)$};
\node[right=0.05cm of vp] {\color{RoyalBlue}$(v',0)$};
\node[left=0.05cm of vpps] {\color{violet}$(v'_*,0)$};
\node[right=0.05cm of vpp] {\color{violet}$(v',0)$};
\node[left=-0.1cm of vp old] {$v'_m$};
\node[right=0.0cm of vps old] {$v'_{*m}$};
\node[below left=-0.1cm of V] {$V$};
\draw[->,violet,line width=0.1ex]  (vpps)--(vpp) node[violet,pos=0.95, above right=-0.1cm] {$u'$};	
\draw[->,line width=0.1ex,gray] (vs) -- (v) node[pos=0.9,sloped, black,above] {$u$};	
\draw[->,RoyalBlue,line width=0.1ex]  (vps)--(vp) node[RoyalBlue,pos=0.8, above right=-0.1cm] {$u'$};		
\end{scope}		
\end{tikzpicture}
	\end{center}
\caption{ A projection of the polyatomic graph of Figure \ref{picture} into  the plane corresponding to the coplanar direction of $\mathbb{R}^3$ containting the vector $u$  and $ \sigma = u'/\left|u'\right|$. Note that the scattering direction is an invariant of the polyatomic molecule collision law regardless of partition functions $r$, $R$, with  $|u'|$ is proportional to $R$.}
\label{picture 2}
\end{figure}
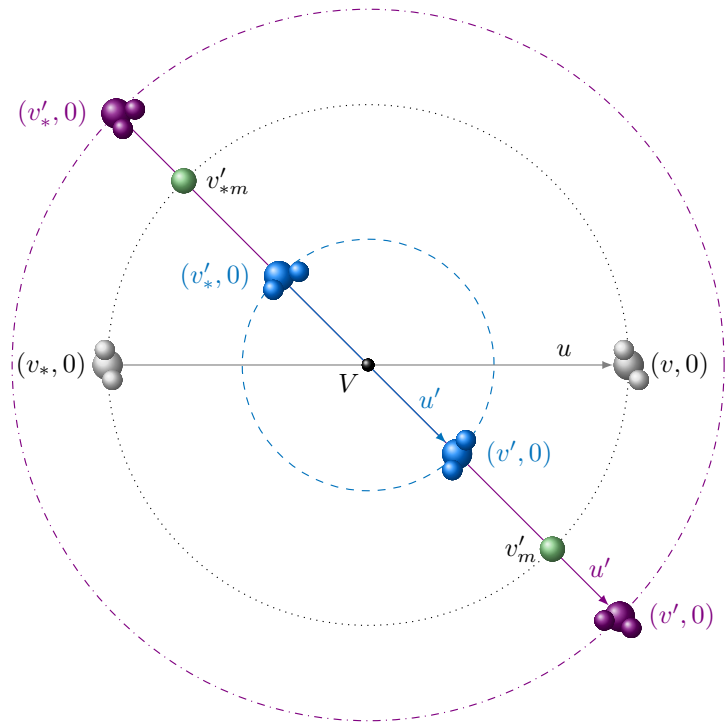

\subsection{The collision transformation}

The first step in modelling the collision operator is to  study transformation from post- to pre-collisional quantities. In particular, we need to compute Jacobian of this transformation, in order to ensure invariance of the measure appearing in the weak form of collision operator.

\begin{lemma}\label{Lemma coll mapping}
The Jacobian of transformation
\begin{equation}\label{coll mapping}
T: (v, v_*, I, I_*, r, R, \sigma) \mapsto (v', v'_*, I', I'_*, r', R', \sigma'), 
\end{equation}
where velocities $v'$ and $v'_*$ are defined in \eqref{velocities}, energies $I'$ and $I'_*$ in \eqref{micro int energies}, and 
\begin{equation}\label{coll mapping-2}
r'=\frac{I}{I+I_*} = \frac{I}{E-\frac{m}{4}\left| u \right|^2}, \quad R'=\frac{m \left|u\right|^2}{4 E}, \quad \sigma'=\frac{u}{\left|u\right|},
\end{equation}
is given by
\begin{equation}\label{coll mapping-3}
J_T = \frac{(1-R) R^{1/2}}{(1-R')R'^{1/2}} = \frac{(1-R) \left|u'\right|}{(1-R')\left|u\right|}.
\end{equation}
\end{lemma}
The proof of this Lemma can be found in Appendix \ref{app Jacobian}.\\

For later purposes we also prove the following Lemma, which finds a function invariant with  respect to the collision process that contains the factor $I^\alpha I_*^\alpha$, crucial for polyatomic modelling. As we will see, $\alpha$ will be related to the degrees of freedom $D$ from macroscopic caloric equation of state \eqref{caloric}.	

We first introduce the following functions, referred by {\em partition functions} for the kinetic-internal energy split, and internal molecular energy split, respectively, given by 
	\begin{equation}\label{fun r R}
\varphi_\alpha(r) := (r(1-r))^{\alpha}, \qquad \psi_\alpha(R) :=  (1-R)^{2\alpha},
\end{equation}
which ensure the expected  invariance property for the conservative polyatomic gas model.

\begin{lemma}\label{Lemma fun r R}
Let functions $\varphi_\alpha(r)$ and $\psi_\alpha(R) $ be from \eqref{fun r R}. The following invariance holds
	\begin{equation*}
	I^\alpha I_*^\alpha \varphi_\alpha(r) \psi_\alpha(R)  = I'^\alpha I'^\alpha_* \varphi_\alpha(r'), \psi_\alpha(R'),
	\end{equation*}
for any power $\alpha\in \mathbb{R}$,	where the involved quantities are linked via the mapping \eqref{coll mapping}.
\end{lemma}
\begin{proof}
	We first write
	\begin{equation*}
	r(1-R)=\frac{I'}{E}, \ I=r'(1-R')E, \ (1-r)(1-R) =\frac{I'_*}{E}, \ I_*=(1-r')(1-R')E,
	\end{equation*}	
	so that
	\begin{equation*}
	r(1-R) \, I \, (1-r)(1-R) \, I_*= I' \, r'(1-R') \,I'_* \, (1-r')(1-R').
	\end{equation*}
	To conclude the proof, it remains to raise this equation to the power $\alpha$.
\end{proof}

\subsection{The Boltzmann collision operator for binary polyatomic gases} 

In this manuscript, we follow the  definition of  collision operator   from \cite{LD-Bourgat94}. Then the natural working functional framework for the evolutions of probability densities is the   Banach space $ L^1(\mathbb R^{3+})$ in the variables $v$ and $I$. 

 This Boltzmann type collision operator, written in strong bilinear form,  is modeled by the non-local operator acting on a pair of probability density measures $(f,g)(v,I)$ defined as follows
\begin{multline}\label{collision operator}
Q(f,g)(v,I) = \int_{{\mathbb{R}^{3+}} \times \K} \left(f' \, g'_* \left(\frac{I \, I_* }{I' \, I'_*} \right)^{\alpha} - f g_* \right) 
\\ \times
\mathcal{B}  \, \varphi_\alpha(r) \, \psi_\alpha(R)  \, (1-R) \, R^{1/2} \, \mathrm{d}R  \, \mathrm{d}r\, \mathrm{d}\sigma \mathrm{d}I_*\, \mathrm{d}v_*,
\end{multline}
$\alpha>-1$,
with  functions $ \varphi_\alpha(r)$,  $\psi_\alpha(R)$ from \eqref{fun r R}. The region of integration is ${\mathbb{R}^{3+}} \times \K$, where  ${\mathbb{R}^{3+}}$ denotes the upper half $3$-dimensional space of unbounded regions of definition of molecular velocity $v$ and internal energy $I$, and $\K$ a compact manifold embedded in the four dimensional space, that is, 
\begin{equation}\label{collision operator2}
{\mathbb{R}^{3+}} := \mathbb{R}^3 \times [0,\infty),  \qquad \text{and }\qquad  \K:= [0,1]^2 \times S^2\, .
\end{equation}
We have used standard abbreviations $f':=f(v',I')$, $g'_*:=g(v'_*,I'_*)$, $f:=f(v,I)$, $g_*:=g(v_*,I_*)$. 

The transition probability rates are, in part,  quantified by  probability measures denoted by
\begin{equation}\label{cross-section}
\mathcal{B}:= \mathcal{B}(v,v_*,I,I_*,R,r,\sigma) \geq 0,
\end{equation}
 that are assumed to be invariant with respect to the following two changes of variables 
\begin{align}
(v,v_*,I,I_*,R,r,\sigma) & \leftrightarrow (v',v'_*,I',I'_*,R',r',\sigma'), \label{microreversibility changes prime}\\
(v,v_*,I,I_*,R,r,\sigma)&\leftrightarrow (v_*,v,I_*,I,R,1-r,-\sigma), \label{microreversibility changes star}
\end{align}
which secures microreversibility.   That means the transition function $\mathcal{B}$  is invariant for such exchange of state satisfying 
\begin{equation}\label{microreversibility changes}
\mathcal{B}(v,v_*,I,I_*,R,r,\sigma) = \mathcal{B}(v',v'_*,I',I'_*,R',r',\sigma')=  \mathcal{B}(v_*,v,I_*,I,R,1-r,-\sigma).
\end{equation}

Besides these usual assumptions on the transition function $\mathcal{B}$, we will have some additional ones, as stated  in the Section \ref{Sec: Ass B} below.\\

In addition,  it is worthwhile to rewrite strong form \eqref{collision operator} in  the symmetric form,
\begin{multline}\label{collision operator-pull-out}
Q(f,g)(v,I) = \int_{{\mathbb{R}^{3+}} \times \K} \left( \frac{f' \, g'_*}{\left( I' \, I'_* \right)^\alpha} - \frac{f g _*}{ \left( I\, I_*\right)^\alpha} \right) 
\\ \times
\mathcal{B}  \, \varphi_\alpha(r) \, \psi_\alpha(R)  \, (1-R) \, R^{1/2} \,  I^\alpha I_*^\alpha \mathrm{d}R  \, \mathrm{d}r\, \mathrm{d}\sigma \mathrm{d}I_*\, \mathrm{d}v_*,
\end{multline}
obtained by just pulling out the factor $\left( I\, I_*\right)^\alpha$ from the gain term.\\

We explain a role of each term involved in the definition of collision operator  \eqref{collision operator} or equivalently \eqref{collision operator-pull-out}. First, renormalization of a  distribution function $f$ by the factor $I^\alpha$ will allow to obtain a correct macroscopic energy law from \eqref{caloric}, and we shall see below that there is a link between  $\alpha$ and degrees of freedom $D$ introduced in \eqref{caloric}. Because of the additional factor $\left( I\, I_*\right)^\alpha$ we need to incorporate functions $ \varphi_\alpha(r)$ and $\psi_\alpha(R)$ in order to have invariance property by virtue of the Lemma \ref{Lemma fun r R}. Finally, term $ (1-R) \, R^{1/2}$ is coming from the Jacobian of collision transformation computed in the Lemma \ref{Lemma coll mapping}, which  ensures good definition of the weak form.\\

In particular, when the transition probability measure  
$$\mathcal{B}(v,v_*,I,I_*,R,r,\sigma)\, \varphi_\alpha(r) \, \psi_\alpha(R)  \, (1-R) \, R^{1/2}  \, \mathrm{d}R  \, \mathrm{d}r\, \mathrm{d}\sigma$$ 
is a finite function of  the pre and post collisional states $(v,v_*,I,I_*)$,  the weak form \eqref{collision operator} or, equivalently, its symmetrized form   \eqref{collision operator-pull-out} can be split into the difference of two positive parts, referred as to the gain $Q^+$ and loss $Q^-$	collisional forms,  namely,
\begin{equation*}
Q(f,g)(v,I) =  Q^+(f,g)(v,I)  \ -  \ Q^-(f,g)(v,I),
\end{equation*}	
with
	\begin{align}
	Q^+(f,g)(v,I) &=  \int_{{\mathbb{R}^{3+}} \times \K}  \frac{f' \, g'_*}{\left( I' \, I'_* \right)^\alpha} \mathcal{B} \, \varphi_\alpha(r)  \psi_\alpha(R) \, (1-R) R^{1/2} \,  I^\alpha I_*^\alpha \mathrm{d}R  \, \mathrm{d}r \, \mathrm{d}\sigma \, \mathrm{d}I_* \, \mathrm{d}v_*, \nonumber \\
	Q^-(f,g)(v,I) &=  {f}(v,I)\,  \nu[g](v,I), \label{loss op}
 \end{align}
where the loss term  is local in $f(v,I)$,  proportional to the {\em collision frequency}, $\nu[g](v,I)$, defined by \\
\begin{equation}\label{collision frequency}
  \nu[g](v,I) := \int_{{\mathbb{R}^{3+}} \times \K} {g _*} \, \mathcal{B}\, \varphi_\alpha(r)  \psi_\alpha(R) \, (1-R) R^{1/2}  \, \mathrm{d}R \,  \mathrm{d}r \, \mathrm{d}\sigma \, \mathrm{d}I_* \, \mathrm{d}v_*.
\end{equation} 
	
\begin{remark}\label{  } It should be noted that the  transition probability rate form   \eqref{cross-section}     is a more general form than a  {\em  differential cross section},  which is the usual expression for classical elastic collisional theory given by just   $|u|\, b(\hat u\cdot\sigma)$ in three dimensions.
In this work, the form of $\mathcal{B}$ may  not only include  such  {\em  differential cross section} factor,  but also needs to include other factors in order to obtain an invariant measure that describes the transition states \eqref{microreversibility changes prime}  and  \eqref{microreversibility changes star} involving internal energies that characterized the  modelling of polyatomic gases. Because of this fact, we refer to $\mathcal{B}  \varphi_\alpha(r)  \psi_\alpha(R)  (1-R) R^{1/2}   I^{\alpha} I_*^{\alpha} $  as the  transition probability rate form. 

In addition, the roll of this factor in the Boltzmann collisional theory is crucial for the theory of existence and uniqueness, as much as decay rates to equilibrium.
\end{remark}

\subsection{Weak form of collision operator} 

We first describe an invariant measure which ensures well defined weak form of the collision operator, by means of the following Lemma.

\begin{lemma}\label{Lemma measure inv}
For any $\alpha > -1$, the following measure is invariant with respect to the changes \eqref{microreversibility changes prime}-\eqref{microreversibility changes star}
	\begin{align}\label{inv-measure}
	\mathrm{d}A=\mathcal{B}(v,v_*,I,I_*,R,r,\sigma) \, \varphi_\alpha(r)  \, \psi_\alpha(R)\,  (1-R) R^{1/2} I^{\alpha}  I_*^{\alpha} \,  \mathrm{d}R \, \mathrm{d}r \, \mathrm{d}\sigma \, \mathrm{d}I_* \,  \mathrm{d}v_* \, \mathrm{d}I  \, \mathrm{d}v.
	\end{align}
\end{lemma}

\begin{proof}
	The proof for \eqref{microreversibility changes prime} easily follows using invariant properties of the transition probability  rate associated to $\mathcal{B}$, Lemma \ref{Lemma fun r R} and Jacobian of the collision transformation from Lemma \ref{Lemma coll mapping}. 
	On the other hand, invariance \eqref{microreversibility changes star} clearly holds.
\end{proof}

\begin{lemma} For  any test function $\chi(v,I)$ that makes the following left hand side finite, the collision operator \eqref{collision operator} takes the following weak form 
\begin{equation}\label{weak form}
\begin{split}
\int_{\mathbb{R}^{3+}} &Q(f, g)(v,I)  \, \chi(v,I) \, \mathrm{d}I \, \mathrm{d}v 
\\
&= \frac{1}{2} \int_{({\mathbb{R}^{3+}})^2 \times K }   f g_* \left( \chi(v',I') + \chi(v'_*,I'_*) - \chi(v,I) - \chi(v_*,I_*)\right)
\\
& \qquad \qquad \qquad \qquad\times  \mathcal{B} \, \varphi_\alpha(r) \,\psi_\alpha(R)\, (1-R) R^{1/2} \, \mathrm{d}R \,  \mathrm{d}r \, \mathrm{d}\sigma \,  \mathrm{d}I_*\, \mathrm{d}v_* \,  \mathrm{d}I \, \mathrm{d}v\\
&= \frac{1}{2} \int_{({\mathbb{R}^{3+}})^2 \times K }  \frac{ f g_*}{(I I_*)^\alpha} \left( \chi(v',I') + \chi(v'_*,I'_*) - \chi(v,I) - \chi(v_*,I_*)\right) \mathrm{d}A,
\end{split}
\end{equation}
with the measure $\mathrm{d}A$ from \eqref{inv-measure}.
\end{lemma}
\begin{proof} We integrate the collision operator \eqref{collision operator} against a suitable test function  $\chi(v,I)$ with respect to $v$ and $I$ variables and then perform changes of variables \eqref{microreversibility changes prime} and \eqref{microreversibility changes star}. Using invariance properties of the measure $\mathrm{d}A$  \eqref{inv-measure} stated  in Lemma \ref{Lemma measure inv}, we obtain
	\begin{multline*}
	\int_{\mathbb{R}^{3+}} 
	Q(f,g)(v,I)  
	\chi(v,I)
\mathrm{d}I	\, \mathrm{d} v 
\\	= 	\int_{({\mathbb{R}^{3+}})^2 \times K} \frac{ f g_* }{I^\alpha \, I_*^\alpha} \left( \chi(v',I') -  \chi(v,I)\right)  \mathrm{d}A
\\	= 	\int_{({\mathbb{R}^{3+}})^2 \times K} \frac{ f g_* }{I^\alpha \, I_*^\alpha} \left( \chi(v'_*,I'_*) -  \chi(v_*,I_*)\right)  \mathrm{d}A,
	\end{multline*}
which yields desired estimate \eqref{weak form}.
\end{proof}

\subsection{The Boltzmann equation} 

In order to describe a polyatomic gas, list of arguments of a distribution function is extended by \emph{microscopic internal energy } $I$, i.e. we take
\begin{equation*}
f:= f(t, v, I).
\end{equation*}
The evolution of $f$ is governed by the Boltzmann equation
\begin{equation}\label{BE}
\partial_t f = Q(f,f)(v,I), 
\end{equation}
where  the  collision operator is written in \eqref{collision operator} or equivalently \eqref{collision operator-pull-out} and  \eqref{collision frequency},  namely,
\begin{equation*}
Q(f,f)(v,I), = Q^+(f,f)(v,I) - Q^-(f,f)(v,I) = Q^+(f,f)(v,I) - f \,\nu[f](v,I).
\end{equation*}

\subsection{${\mathcal H}$-theorem}
 A natural dissipative quantity that is minimized at the statistical equilibrium is usually given by the concept of entropy in the associated evolution of the probability density function, solutions to the associated Boltzmann equation for binary polyatomic gases. In this case, as in the classical case of  elastic monatomic gases, such quantity is  is given by the  entropy  functional, written in the space space homogeneous setting,  
 \begin{equation}\label{entropy}
{\mathcal H}(f)(t) :=   \int_{\mathbb{R}^{3+}} f(t,v,I) \log(f(t,v,I) I^{-\alpha})\, \mathrm{d}I \mathrm{d}v\, .
\end{equation}

Then, by means of the weak formulation associated to the equation \eqref{BE} defined above, the  evolution of the entropy \eqref{entropy} is obtained when multiplying both sides of equation \eqref{BE} by  $\log(f(t,v,I) I^{-\alpha})$  and integrating with respect to  the pair $v$ and $I$.   This results in the    entropy
production functional  ${\mathcal D}(f)(t)$  associated to the Boltzmann collision operator for binary polyatomic gases, that is
\begin{equation}
{\mathcal D}(f)(t) := \int_{\mathbb{R}^{3+}} Q(f,f)(t,v,I) \log(f(t,v,I,) I^{-\alpha}) \, \mathrm{d}I \mathrm{d}v\, .
\end{equation}
The following theorem focuses on the properties of this entropy dissipation functional.
\ \\
\begin{theorem}[The ${\mathcal H}$-theorem] 
	Let the transition function $\mathcal{B}$ be positive function almost everywhere, and let $f\geq 0$ such that the collision operator $Q(f,f)$ and entropy production ${\mathcal D}(f)$ are well defined. Then the following properties hold
	\begin{itemize}
		\item[i.] Entropy production is non-positive, that is
		\begin{equation*}
		{\mathcal D}(f) \leq 0.
		\end{equation*}
		\item[ii.] The three following properties are equivalent
		\begin{itemize}
			\item[(1)]${\mathcal D}(f) =0$,\\
			\item[(2)] $Q(f,f) =0$  for all  $(v,I) \in \mathbb{R}^{3+}$,\\
			\item[(3)] There exists $n\geq 0, U \in \mathbb{R}^3, \ \text{and} \ T>0$, such that the unit mass renormalized Maxwellian equilibrium for polyatomic gases is 
			\begin{equation}\label{Maxwellian}
			M_{eq}(v, I) = \frac{n}{Z(T)} \left( \frac{m}{2 \pi k_B T} \right)^{3/2} I^\alpha \ e^{- \frac{1}{k T} \left( \frac{m}{2} \left| v - U \right|^2 + I \right)} ,
			\end{equation}
			where $Z(T)$ is a partition (normalization) function 
			\begin{equation*}
			Z(T) = \int_{[0,\infty)} I^\alpha e^{-\frac{I}{k T} } \mathrm{d}I = (k T)^{\alpha +1} \Gamma(\alpha+1),
			\end{equation*}
			with	$\Gamma$ as gamma function.
		\end{itemize}		
		
	\end{itemize}
\end{theorem}  
The proof can be found in \cite{LD-Bourgat94}.

\subsection{Functional space}

The choice of a functional space is crucial for the framework to construct  solutions with a good physical and mathematical meaning.
Given the description of the dynamics of the polyatomic gas Boltzmann model {   for} the evolution of a non-negative  probability density  measure describing the binary  interaction of particles exchanging the states of their   molecular velocities and internal energies, a natural functional normed Banach space is the one of integrable functions with polynomial weights, whose norms describe the observables or moments associated to  such measures. 
   
 The natural polynomial weight to generate a   weighted  Banach space containing the same spaces with a   lower polynomial degree weight, is   given  by 
\begin{equation}\label{brackets}
\left\langle  v, I \right\rangle = \sqrt{ 1+  \frac{1}{2} \left| v  \right|^2 + \frac{I}{m}},
\end{equation} 
all  associated to the velocity $v\in \mathbb{R}^3$ and microscopic internal energy $I \in  [0,\infty)$,  which is independent of mass units. 
We refer to this weight function as the Lebesgue bracket and  
  will  look for a  unique solution to the Cauchy problem associated to the Boltzmann equation \eqref{BE} in the Banach space  weighted by powers of this Lebesgue weight form. 
More precisely, we define
\begin{equation}
\begin{split}\label{space L_k^1}
L_k^1 &= \left\{ f \ \text{measurable}: \int_{\mathbb{R}^{3+}} \left| f(v, I)  \right|\left\langle v, I \right\rangle^{\tk}  \mathrm{d} I \mathrm{d}v < \infty, \ k\geq 0 \right\},
\end{split}\end{equation}  
with the range of integration ${\mathbb{R}^{3+}} = \mathbb{R}^3 \times [0, \infty)$ from \eqref{collision operator2}. Its associated norm is
\begin{equation}\label{norm}
\left\| f \right\|_{L_k^1} =  \int_{\mathbb{R}^{3+}} \left| f(v, I)  \right|\left\langle v, I\right\rangle^{\tk} \mathrm{d} I \mathrm{d}v.
\end{equation}
We recall the monotonicity property for norms weighted with \eqref{brackets},
\begin{equation}\label{monotonicity of norm}
\left\| f \right\|_{L_{k_1}^1} \leq \left\| f \right\|_{L_{k_2}^1} \ \text{whenever} \ 0\leq \ k_1 \leq k_2.
\end{equation}

We also define the associated semi-norm, with the classical   notation, 
\begin{equation}
\label{semi-norm_k}
\dot L_k^1 = \left\{ f \ \text{measurable}: \int_{\mathbb{R}^{3+}} | f(v, I)|\,  \left( \frac{1}{2} \left| v  \right|^2 + \frac{I}{m} \right)^{k}  \mathrm{d} I \mathrm{d}v < \infty, \ k\geq 0 \right\}.
\end{equation}

Yet, when we refer to the distribution function, without loss of generalization,  the norm \eqref{norm} is called the polynomial moment, as the following Definition \ref{def poly moment} precises.
\begin{definition}\label{def poly moment} 
	The Lebesgue polynomial moment of order $k\geq 0$ associated to any integrable function  $g(t,v,I)$ is defined with
\begin{equation*}
\pmk[g](t) = 	\int_{\mathbb{R}^{3+}}  g(t,v,I) \left\langle  v, I \right\rangle^{\tk} \mathrm{d}I	\mathrm{d} v \end{equation*}
with the Lebesgue bracket weight from \eqref{brackets}.
\end{definition}

Note that the Lebesgue polynomial moment of any order of any non-negative probability density function $f(v,I)$ coincides with $ \left\| f \right\|_{L_{k}^1}$, and when $k=0$  is the  macroscopic number density associated to the  probability density $f(t,v,I)$, hence $ \mathfrak{m}_0[f](t) =  \| f \|_{L_0^1} (t) = \| f \|_{\dot L_0^1} (t)$. However, the semi-norm   $\| f \|_{\dot L_k^1} (t) <  \pmk[f](t),$ for any $k>0$. 

The particular case $k=1$ is associated to  the  classical total kinetic and internal energy of a non-negative $f(t,v,I)$  is smaller that the one generated by the Lebesgue polynomial moment, that is   $ \| f \|_{\dot L_1^1} (t) <   \mathfrak{m}_1[f](t)$.  It is important to notice that   if the classical energy of a non-negative probability density measure is positive, then such $f(t,v,I)$ can not be a an absolutely singular measure.

In this manuscript, we also study exponentially weighted $L^1$-norms, by means of  exponential moments defined as follows.
\begin{definition}
Exponential moment for an integrable  function $g$,   of  rate $\beta>0$ and order $ 2s$, $0<s\leq \ho$,   is defined by 
\begin{equation}\label{exp moment}  
\mathbb{\mathcal{E}}_{ s}[g](\beta,t) :=  \int_{\mathbb{R}^{3+}} g(t,v,I) \, e^{\beta \left\langle v, I \right\rangle^{2s}} \mathrm{d} I \, \mathrm{d}v. 
\end{equation} 
\end{definition}
Note that, { for } any non-negative probability density function $f(v,I)$,   their associated exponential moments coincide with their   exponential weighted $L^1-$norm.

In the  upcoming Section, we provide a physical interpretation to some polynomial moments.
 
\subsection{Macroscopic observables} We first note that for certain test functions weak form \eqref{weak form} annihilates. This is encoded in the collision conservation laws \eqref{micro CL}. Namely, the  following Lemma holds.

\begin{lemma}
	The collision invariants for the collision operator \eqref{collision operator}, i.e. functions $\chi(v,I)$ for which the weak form \eqref{weak form} annihilates
	\begin{equation*}
	\int_{\mathbb{R}^{3+}} Q(f, g)(v,I)  \, \chi(v,I) \, \mathrm{d}I \, \mathrm{d}v 
=0,
	\end{equation*}
	 are linear combination of the following functions
	\begin{equation}\label{coll inva}
	\chi_1(v, I) =m, \qquad \chi_{k}(v, I) = m\, v_k, \ k=1,2,3, \qquad \chi_5(v, I) = \frac{m}{2} \left|v\right|^2 + I.
	\end{equation}
\end{lemma}
Macroscopic observables are defined as moments of the distribution function $f$ against   functions of molecular variables, velocity $v\in\mathbb{R}^3$ and internal energy $I\in [0,\infty)$. For example, when test functions are  the  collision invariants  \eqref{coll inva} then we define mass density $\rho$, momentum density $\rho\,U$ and total energy density $\frac{\rho}{2} \left| U \right|^2 + \rho\, e $ of a polyatomic gas as  the following moments
\begin{equation*}
\rho = \int_{ {\mathbb{R}^{3+}}} m f \, \mathrm{d}I \, \mathrm{d}v, \quad
\rho \, U = \int_{ {\mathbb{R}^{3+}}} m v f \, \mathrm{d}I \, \mathrm{d}v, \quad
\frac{\rho}{2} \left| U \right|^2 + \rho\, e = \int_{ {\mathbb{R}^{3+}}} \left( \frac{m}{2} \left|v\right|^2 + I \right) f \, \mathrm{d}I \, \mathrm{d}v.
\end{equation*}
We highlight  the relation between collision invariants and the Lebesgue weight \eqref{brackets}, 
\begin{equation*}
\langle v, I \rangle^2 = \frac{1}{m} \left( \chi_1 + \chi_5 \right).
\end{equation*} 
Therefore, polynomial moment  of the order $k=0$  multiplied by mass,  $ m\, \mathfrak{m}_{0}(t)$ is interpreted as gas mass density, while for $k=1$,  the moment $ m\, \mathfrak{m}_{1}(t)$  is   the sum of  mass density plus total energy density of the gas,
\begin{equation*}
m\, \mathfrak{m}_0 = \rho, \qquad m\, \mathfrak{m}_2=\rho + \frac{\rho}{2} \left| U \right|^2 + \rho\, e .
\end{equation*}


We can finally find connection with the caloric  equation of state for a polytropic gas \eqref{caloric}.
By introducing the peculiar velocity $c=v-U$, we can define the  internal energy density,
\begin{equation*}
 \rho\, e = \int_{ {\mathbb{R}^{3+}}} \left( \frac{m}{2} \left|c\right|^2 + I \right) f \, \mathrm{d}I \, \mathrm{d}v.
\end{equation*}
In  the local  equilibrium state, when distribution function has a shape of local Maxwellian \eqref{Maxwellian}, internal energy takes the form
 \begin{equation}\label{caloric alpha}
 \rho\, e =\left( \alpha + \frac{5}{2} \right) n\,k\,T.
 \end{equation}
 Now relation to caloric equation of state for a polytropic gas \eqref{caloric} becomes evident. Then we can connect $D$ from \eqref{caloric} and $\alpha$ from \eqref{caloric alpha}, that also appears in the definition of collision operator \eqref{collision operator},
 \begin{equation*}
 \alpha = \frac{D-5}{2}.
 \end{equation*}
 Since for polyatomic gases $D> 3$, we obtain the overall condition  $\alpha> -1$. We recall that $\alpha=-1$ corresponds to a monatomic gas ($D=3$), when we obtain the classical relation $\rho\, e= (3/2) n k T$.
 
 The following Table \ref{Tab} shows a relation between different modes of a polyatomic gas molecule with  its number of degrees of freedom $D$ as much as the corresponding value of the parameter $\alpha$. 
 
 \begin{center}
 	\begin{table}[h]
 		\caption{Degrees of freedom $D$ and the value of $\alpha$ that correspond to different modes (combinations of translation/rotation/vibration), where $N \geq 2$  stands for the number of atoms in a polyatomic gas molecule.}
 		\begin{tabular}{| m{2cm}||m{1.6cm}|| m{1.6cm}||m{2cm}| @{}m{0cm}@{} } \hline
 			& \multicolumn{2}{c||}{Translation and rotation} & \multirow{2}{2cm}{Translation, rotation and vibration} & \\[5pt] \cline{2-3}
 			& Linear molecule &  Non-linear molecule & & \\[5pt] \hline
 			\hfil  D   & \hfil 5  & \hfil6 & \hfil$3N$ & \\[5pt] \hline
 			\hfil 	  $\alpha$  & \hfil0 &\hfil $\frac{1}{2}$ &\hfil $\frac{1}{2}(3 N-5) $ &   \\[5pt]  \hline
 		\end{tabular}
 		\label{Tab}
 	\end{table}
 \end{center}

\section{Sufficient properties for existence and uniqueness theory}\label{Sec: Suff}

In this Section we describe sufficient tools needed to build the Existence and Uniqueness theory  to be presented in Section \ref{Section Ex Uni proof}.   Namely, at the first place we choose an appropriate transition function $\mathcal{B}$ that corresponds to an extended Grad assumption.  Namely, we assume the form of hard potentials with a rate $\gamma\in(0,2]$ for both relative speed and molecular internal energies combined additively, 
$$\tf, \qquad \gamma\in(0,2],$$	
 with a control, from above and below, with respect to the parameters $R, r, \sigma$ belonging to the compact manifold $\K$.  Moreover, we find relevant physical examples, namely the three models for $\mathcal{B}$, that satisfy the imposed  assumptions on $\mathcal{B}$. With this form of the transition function, we prove essential estimates, its upper and lower bounds in molecular velocities and internal energies, that allow to conclude {   the coerciveness estimate for the loss part of the collision operator in Section  \ref{Coerciveness} and the decay of $k-$Lebesgue moment of its gain part  with respect to $k$ in Section \ref{Sec: fund lemmas}.}
 Then in Section \ref{Sec: poly mom}   we are going to have a priori estimates for any solution of the Boltzmann equation for polyatomic gases in $L^1_{k}$ for  $ k> k_*$ with $k_*$    determined by the bounds for  transition function $\mathcal{B}$ and the conditions of  fundamental Lemmas of Sections  \ref{Coerciveness} and \ref{Sec: fund lemmas}.  Finally, we define  
  an \emph{invariant region} $\Omega \subset L^1_{\ho}$ for the Boltzmann equation,  in which the collision operator $Q: \Omega \rightarrow L_{\ho}^1$  satisfies (i) H\"{o}lder continuity, (ii) Sub-tangent and (iii) one-sided Lipschitz conditions. These are sufficient conditions to obtain existence and uniqueness of a global in time solution with a regularity to be described in Section \ref{Section Ex Uni proof}.

\subsection{Transition function $\mathcal{B}$}\label{Sec: tf}
One of the essential ingredients for building existence and uniqueness theory is an assumption on transition function $\mathcal{B}$, that quantifies the collision frequency through scattering mechanisms and partition functions as a function of the total molecular energy \eqref{micro CL energy mass-rel vel}. 
In this manuscript, we want to keep the transition function $\mathcal{B}$ as general as possible in order to allow our kinetic model to cover a broad class of physical interpretations.   To that end, apart from its  positivity and micro-reversibility requirements stated in \eqref{microreversibility changes}, we assume the following minimal mathematical requirements  to ensure existence and  uniqueness properties associated to the initial value problem for the  Boltzmann equation for binary interaction of polyatomic gases as defined in \eqref{BE}.

\begin{assum}[The form of the transition function $\mathcal{B}$] \label{Sec: Ass B} 
Let $\tB=\tB(|u|,I,I_*)$ be defined as 
\begin{equation}\label{tf}
\tB(|u|,I,I_*) := \tfu, \quad u:= v-v_*, \quad \gamma \in (0,2].
\end{equation}
 We assume that 	the transition function $\mathcal{B}:= \mathcal{B}(v,v_*,I,I_*,r,R,\sigma)$ satisfies the following extended Grad assumption for collision kernels,
\begin{equation}\label{trans prob rate ass}
\dgl \, \egl \, b(\hat{u}\cdot\sigma)\, \tB(|u|,I,I_*)
\leq 
\mathcal{B} 
\leq  \dgu \, \egu \,b(\hat{u}\cdot\sigma)\, \tB(|u|,I,I_*),
\end{equation}
for every $v, v_* \in \mathbb{R}^3$, $I,I_*  \in [0,\infty)$, $r, R\in[0,1]$, $\sigma\in S^2$,  with $\hat{u}=u/\left|u\right|$,
where  functions $b(\hat{u} \cdot \sigma)$,   $\dgu$, $\dgl$, $\egu$,  and $\egl$ satisfy the following integrability  conditions,
\begin{enumerate}
	\item the angular function  $b(\hat{u} \cdot \sigma)$ is integrable  with respect to the measure $\mathrm{d}\sigma$,
	\begin{equation}\label{ass b}
	b(\hat{u} \cdot \sigma) \in L^1(S^2; \mathrm{d}\sigma),
	\end{equation}
	\item functions $\dgu$ and $\dgl$  are integrable with respect to the measure $\varphi_\alpha(r) \mathrm{d}r$, with $\varphi_\alpha(r)$ from \eqref{fun r R}, more precisely
	\begin{equation}\label{ass d}
	\dgu \varphi_\alpha(r), \quad  \dgl \varphi_\alpha(r)  \in L^1([0,1]; \mathrm{d}r),
	\end{equation}
	and additionally
	\begin{equation*}
	\dgu = d^\text{ub}_\gamma(1-r), \quad 	\dgl = d^\text{lb}_\gamma(1-r),
	\end{equation*}
	which ensures the second microreversibility property \eqref{microreversibility changes},\\
	
	\item functions $\egu$ and $\egl$ are integrable with respect to the measure $\psi_\alpha(R) (1-R) R^{1/2} \mathrm{d} R$, where $\psi_\alpha$ is introduced in \eqref{fun r R}, namely 
	\begin{equation}\label{ass e}
	\egu \psi_\alpha(R) (1-R) R^{1/2}, \quad  \egl \psi_\alpha(R) (1-R) R^{1/2}  \in L^1([0,1]; \mathrm{d}R).
	\end{equation}
	\end{enumerate}
\end{assum}
\begin{remark}
	We observe that   conditions \eqref{ass d} and \eqref{ass e}  involve the weighted averages of the  factors $\dgl$ and $\dgu$ product to  the partition function for the molecular energy $ \varphi_\alpha(r)$; as well as $ \egl$ and $ \egu$ product to  the partition function for the split of kinetic and internal energy $ \psi_\alpha(R)$; respectively. We introduce  the short hand notation to these  averages by defining the following constants,
	\begin{equation}\label{meas values}
	\begin{alignedat}{2}
	\cgalr &:= \int_0^1 	\dgl \varphi_\alpha(r) \mathrm{d}r,  \quad \CgalR &:= \int_0^1 \egl \psi_\alpha(R) (1-R) R^{1/2} \mathrm{d}R, \\
	\cgaur &:= \int_0^1 	\dgu \varphi_\alpha(r) \mathrm{d}r, \quad \CgauR &:= \int_0^1 \egu \psi_\alpha(R) (1-R) R^{1/2} \mathrm{d}R.
	\end{alignedat}
	\end{equation}
Moreover,
\begin{equation}\label{kappas}
\kappa^{lb} = \left\| b \right\|_{L^1(\mathrm{d}\sigma)}  	\cgalr \CgalR, \quad \quad \kappa^{ub} = \left\| b \right\|_{L^1(\mathrm{d}\sigma)}  	\cgaur \CgauR.
\end{equation}
\end{remark}

 Note that assumption~\ref{Sec: Ass B}  stresses the transition probability associated to the differential cross section is an extended Grad assumption for hard potentials with a dependance to the internal energy exchange which characterizes polyatomic collisional gas models. As such, the assumption of   transition function  $\mathcal B$ will be shown to be achievable for at least three choices of models  of $\mathcal{B}(v,v_*,I,I_*,r,R,\sigma)$ satisfying conditions (\ref{tf}-\ref{ass e}),  that are sufficient  to rigorously  prove  the existence of finite  upper and strictly positive  lower bounds 
sufficient  for solving the Cauchy problem associated to  a natural Banach space defined by \eqref{space L_k^1} with a natural norm characterized by the Lebesgue weight function \eqref{brackets} and norm  by \eqref{norm}.  
Such upper and lower estimates must allow  to control infinity system of ordinary Differential Inequalities (ODI) associated to solutions to the initial value problem   \eqref{BE}  uniformly in time.  
This strategy is rather elaborated and shall be presented in several steps.

\subsection{Models for transition function $\mathcal{B}$}\label{Sec: tf models}

Next  we propose three different choices of the transition function  $\mathcal{B}$   and  prove   that   each choice  satisfies all conditions on Assumption~\ref{Sec: Ass B}. 
We also focus our attention  to the multiplicative factors depending on $r$ and $R$, to  the term  $  b(\hat{u}\cdot\sigma)\, \tB$, in the upper and lower bounds of  $\mathcal B$ from \eqref{trans prob rate ass}. 

For each choice of the extended Grad decomposition that follows we specify functions $\dgu$, $\dgl$, $\egu$ and $\egl$ that, not only  fulfill the integrability conditions, but also   provide explicit expressions for controlling constants, from above and below \eqref{meas values}.   These integrals are used in the Section \ref{Sec: k*}. Here we  calculate the constants \eqref{meas values} only in the first example.

\subsubsection{Model 1 (The total energy)}\label{Sec Model 1}

We first consider the total energy in the relative velocity-center of mass velocity framework, that is
\begin{equation}\label{model 1}
{\mathcal{B}}(v,v_*,I,I_*,r,R,\sigma)  = b(\hat{u} \cdot \sigma)\left( \frac{m}{4} \left|v-v_*\right|^2 + I + I_* \right)^{\gamma/2}, \qquad\gamma \in (0,2].
\end{equation}
Then $\mathcal{B}$ is of the form \eqref{trans prob rate ass} with 
\begin{equation*}
\dgl = \dgu= 1, \quad \egl=  m^{\gamma/2} 2^{-(\gamma/2+1)}, \quad   \egu=m^{\gamma/2},
\end{equation*}
as proven is in the Appendix \ref{Sec Model 1 upper}, \ref{Sec Model 1 lower}.  Therefore, using \eqref{fun r R},  performing the integration to calculate  the constants \eqref{meas values}  by means of the  Gamma function,   for this Model 1 it follows
\begin{multline}\label{model 1 const}
\cgalr = \cgaur =  \frac{\Gamma(\alpha+1)^2}{\Gamma(2\alpha+2)}, \quad \CgalR=m^{\gamma/2} \ 2^{-(\gamma/2+1)} \frac{\sqrt{\pi}  \Gamma(2\alpha+2)}{2 \Gamma\left(2\alpha + \frac{7}{2}\right)}, \\
\CgauR =  m^{\gamma/2} \frac{\sqrt{\pi}  \Gamma(2\alpha+2)}{2 \Gamma\left(2\alpha +\frac{7}{2}\right)},
\end{multline}
for $\alpha>-1$, where $\Gamma$ represents the Gamma function.

\subsubsection{Model 2 (kinetic and microscopic internal energy detached)}\label{Sec Model 2} 
In this model, we split kinetic and microscopic internal energy for the colliding pair of particles, by using parameter $R \in [0,1]$, 
\begin{equation}\label{model 2}
{\mathcal{B}}(v,v_*,I,I_*,r,R, \sigma) = b(\hat{u} \cdot \sigma) \left( R^{\gamma/2} |v-v_*|^\gamma + (1-R)^{\gamma/2} \left(\frac{I+I_*}{m}\right)^{\gamma/2} \right),
\end{equation}
for $\gamma \in (0,2]$. 
As proven in the appendix sections \ref{Sec Model 2 upper} and \ref{Sec Model 2 lower}, this model satisfies the form 
\eqref{trans prob rate ass}, with
\begin{equation}\label{model 2 plot}
\dgl = \dgu= 1, \quad \egl=\min\{R, (1-R) \}^{\gamma/2}, \quad \egu = \max\{R, (1-R) \}^{\gamma/2}.
\end{equation}
Another possible choice is
\begin{equation*}
\dgl = \dgu= 1, \quad \egl=R^{\gamma/2} (1-R)^{\gamma/2},  \quad \egu=1.
\end{equation*}

\subsubsection{Model 3 (kinetic and particle's microscopic internal energies detached)}\label{Sec Model 3}  In this model we separate kinetic and microscopic internal energy with the parameter $R \in [0,1]$. Furthermore, internal energy is divided among colliding particles with the help of parameter $r \in [0,1]$. More precisely, we consider, 
\begin{multline}\label{model 3}
{\mathcal{B}}(v,v_*,I,I_*,r,R) =  b(\hat{u} \cdot \sigma)\left(  R^{\gamma/2} |v-v_*|^\gamma  \right. \\ \left. + \left( r (1-R)\frac{I}{m} \right)^{\gamma/2} + \left( (1-r) (1-R)\frac{I_*}{m} \right)^{\gamma/2} \right),
\end{multline}
for $\gamma \in (0,2]$.
Then the form \eqref{trans prob rate ass} is satisfied with 
\begin{multline}\label{model 3 plot}
\dgl =\min\{ r, (1-r)  \}^{\gamma/2},  \quad \dgu=1, \\
\egu =  2^{1-\gamma/2}   \max\left\{ R, (1-R) \right\}^{\gamma/2}, \quad
\egl = \min\{ R, (1-R) \}^{\gamma/2}, 
\end{multline}
or with
\begin{equation*}
\dgl= r^{\gamma/2} (1-r)^{\gamma/2}, \quad \dgu=1,\quad
\egu =  2^{1-\gamma/2}, \quad
\egl = R^{\gamma/2} (1-R)^{\gamma},
\end{equation*}
as shown in \ref{Sec Model 3 upper} and \ref{Sec Model 3 lower}.

This Model 3 for the transition function $\mathcal{B}$ is of particular importance for establishing macroscopic moments models starting from the Boltzmann equation. Namely, in \cite{MPC-Dj-S} it is shown that the six moments model with   this transition function provides a source term which satisfies the macroscopic residual inequality on the whole range of six moments model validity, that provides the total agreement with the extended thermodynamics theory of six moments,  as one of the rare  systems  admitting a non-linear or  \emph{ far from equilibrium}  closure of the governing equations using the entropy principle. Moreover, for this model 3   of the transition function $\mathcal{B}$, the macroscopic fourteen moments model achieves   the matching with experimental data as well. More precisely, \cite{MPC-Dj-S}  shows that the model 3 yields transport coefficients (shear and bulk viscosities and heat conductivity)   for which both experimentally measured dependence of the shear viscosity upon temperature  is recovered and  the Prantdtl number  coincides with its  theoretical value at a satisfactory level. This remarkable success  relies on the new \emph{additive}  form of the transition function that we propose in \eqref{tf}, instead of the multiplicative one used so far in the literature.

\section{Coerciveness estimates }\label{Coerciveness}

In this and next  Sections we prove   fundamental lemmas, that should be used sequentially,  as they are presented. \\

All of them are motivated by the search of a proof showing that $k$-th polynomial moment of the solution will satisfy an Ordinary Differential Inequality (ODI) in the  Banach Space $L^1_k(\mathbb{R}^{3+})$  with a negative super-linear term, that is for any $\gamma\in (0,2]$, 
\begin{equation}\label{ODI Banach}
\begin{split}
\frac{\mathrm{d}}{\mathrm{d}t} \left\|f\right\|_{L^1_{k}}\!\!(t) &= 	\int_{ \mathbb{R}^{3+}} Q(f,f)(v,I) \, \langle v, I \rangle^{\tk} \mathrm{d} I \mathrm{d}v \\
&= \mathfrak{m}_k\left[  Q(f,f) \right]  \leq - A_{\ks}   \left\|f\right\|_{L^1_{k}}^{1+ \frac{\gamma/2}{k}}\!\!(t)   +  B_k \left\|f\right\|_{L^1_{k}}\!\!(t),
\end{split}
\end{equation}
 where the sequence of strictly positive constants $\{A_k\}_k$ is strictly decreasing and   is such that 
 $0<A_k < A_{\ks}$ for all $k>\ks$,  for large enough $\ks$, to be determined  later. This $A_{\ks}$   depends on the transition probability and partition functions, on $\gamma$ and the initial mass and  is independent of $k>\ks$. 
Our aim is to gather sufficient a priori estimates that will show the existence of global in time solutions to an initial value problem  associated to  to \eqref{ODI Banach}. 

This Section~\ref{Coerciveness} focuses on the fundamental question of obtaining coerciveness  in the natural Banach space norm for the collisional form associated to polyatomic gas models.    This result can be viewed as the  analogue to a coercive constant in the classical theory for diffusion type equations in continuum mechanics where the control of the energies or suitable Banach norms is done in Sobolev spaces, while in the framework of statistical mechanics the suitable norms are non-reflexive Banach spaces, as in this case $L^1_k$. \\

Coercive estimate is fundamental for the existence of global solutions for the Boltzmann flow.  This property provides the following functional inequality that controls the convolution of any function $f \in L_{1^+}^1$ with a potential function of the form  \eqref{tf}  from below by the corresponding Lebesgue bracket \eqref{brackets}. We are inspired by the work \cite{Alonso-IG-BAMS} for the classical Boltzmann equation, as  such  inequality is sufficient  to control from below the collision frequency associated  \eqref{collision frequency}, and the  loss collision  operator  associated to the polyatomic gas model, under assumptions  on the transition probability states as described in Section~3. In addition, the coercive estimates we present here are fundamental to obtain in $W^{m,2}$ Sobolev (Hilbert) spaces as  shown in \cite{IG-A-M}.\\

\begin{lemma}[Lower bound]\label{lemma lower bound}
	Let $\gamma \in [0,2]$. For some function  $0\leq \left\{ f(t) \right\}_{t\geq 0} \subset L_{\ho}^1$ we assume that it satisfies
	\begin{equation}\label{mass energy}
	\Ml \leq  \int_{\mathbb{R}^{3+}}  f(t, v, I)  \mathrm{d}I \mathrm{d}v \leq \Mu, \quad \El \leq  \int_{\mathbb{R}^{3+}} f(t, v, I) \left( \frac{1}{2}  \left|v\right|^2 + \frac{I}{m} \right) \mathrm{d}I \mathrm{d}v \leq \Eu,
	\end{equation}
	for some positive constants $\Ml$, $\Mu$, $\El$ and $\Eu$, and
	\begin{equation}\label{momentum}
	\int_{\mathbb{R}^{3+}} f(t, v, I)  \, v \,  \mathrm{d}I\, \mathrm{d}v=0.
	\end{equation}
	Assume also boundedness of the moment
	\begin{equation}\label{moment epsilon}
	\int_{\mathbb{R}^{3+}} f(t, v, I) \,   \left( \frac{1}{2} \left| v \right|^2 + \frac{I}{m} \right)^\frac{2+\pl}{2}   \mathrm{d}I\, \mathrm{d}v \leq \D, \quad \pl>0. 
	\end{equation}
	Then there exists a constant $\clb>0$ which depends on constants  $c$, $C$, $B$, $\pl$ and $\gamma$ from the assumptions \eqref{mass energy}-\eqref{moment epsilon} above such that the convolutional form of any $f\in  L^1_{1}$ satisfying conditions \eqref{mass energy} and \eqref{moment epsilon}  with a power law function of order  $\gamma$  from \eqref{tf},   is controlled from below by the local Lebesgue weight of order $\gamma$, namely
	\begin{multline}\label{lower bound lemma}
	f*\left(|v|^\gamma \!+ \! \left(\frac Im\right)^{\!\frac{\gamma}2}\!\right) \!:= \! \int_{\mathbb{R}^{3+}} \! \!\! f(t, w, J) \left( \!\left|v\!-w\right|^{\gamma} \!+\! \left(\! \frac{I\! +\! J}{m} \! \right)^{\!\frac{\gamma}2} \! \right) \!\mathrm{d}J  \mathrm{d}w 
	\geq \  \clb\, \langle v, I \rangle^\gamma, \qquad\qquad
	\end{multline}
 It takes the form 
		\begin{multline}\label{clb}
		\clb 
		=  \frac{ \min\left\{ \Ml, \El \right\}}{ 8}   \left(  2^{  4+\pl } \!\left(\frac{ \max\{ \Mu, \D\}}{ \El}\right) \! \left(1+ \left( \frac{  2 (\Mu+\Eu)}{  { \Ml} \left( 2^{\gamma -2 } - \frac{1}{8} \right) } \right)^{\!\frac{2}{\gamma}}\!\right)^{\frac{2+\pl}2}\right)^{\!\frac{-2+\gamma}{\pl}}
		\\ \times \left( 1 +  \left( \frac{  2(\Mu + \Eu)}{  { \Ml} \left( 2^{\gamma -2 } - \frac{1}{8} \right) } \right)^{\!\frac{2}{\gamma}}\right)^{\!-\gamma/2},
		\end{multline}
for $\gamma \in (0,2]$, and if $\gamma=0$ 		then $\clb=2$.
\end{lemma}

\begin{remark}
		It is also important to notice that both constants $\Ml \leq  \| f\|_{ L_0^1}$ and $\El \leq \| f\|_{\dot L_1^1}$  coincide with the semi-norms defined in \eqref{semi-norm_k},  for a non-negative  distribution density $f(t,v,I)$, from \eqref{mass energy} and \eqref{moment epsilon} can not be both zero, otherwise  the lower bound $\clb$  from \eqref{clb} would vanish and the coercive estimate  \eqref{coercive_estimate} would fail, and consequently the existence theory of global in time solutions,  we are about to develop, would also fail.  It is also clear that singular measures can not be solutions in this Banach topology,  as both $\Ml$ and $\El$ cannot be zero. 
\end{remark}

\medskip
\begin{proof}[Proof of Lemma~\ref{lemma lower bound}] 
	{ The point of this proof is to bound from below a convolution of a $f\in L^1_1$, satisfying conditions \eqref{mass energy} and \eqref{moment epsilon} integrated with the weight defined in \eqref{tf}, 
		by a constant depending on parameters from those conditions and on  $\gamma$ multiplied by the Lebesgue weight 
		\begin{equation}\label{brackets2}
		\langle v,I \rangle^\gamma = \left( 1+\frac{1}{2}  \left|v\right|^2 + \frac{I}{m} \right)^{{\gamma}/{2}},
		\end{equation}   
		where the integral's weight  function is given by  
		\begin{equation}\label{g gamma}
		\tB(|v-w|,I,J)= \!\left|v\!-w\right|^{\gamma} \!+\! \left(\! \frac{I\! +\! J}{m} \! \right)^{\!\frac{\gamma}2}  \in \mathbb{R}.
		\end{equation}	
		This task entices to study the control from below and above of \eqref{g gamma} which will generate, after performing  the integration with $f(w, J)$, the desired lower bound involving \eqref{brackets2}. Such calculation  requires  a major extension of the proof proposed in \cite{Alonso-IG-BAMS} for the classical Boltzmann in $\mathbb{R}^d$, as for the current model in  $\mathbb{R}^{3+}$  needs to be adjusted to the topological distances in half spaces. \\
		
		We work out the details of  this proof in two parts after we introduce the notion of upper semi-spheres  suitable for  the estimates.
		\\
		
		We start by   setting  the  following set in  $\mathbb{R}^{3+} = \mathbb{R}^3 \times [0,\infty)$. More precisely, we introduce 
		the  semi-sphere  in the upper half space $\mathbb{R}^{3+}:=\mathbb{R}^3 \times [0, \infty)$  by
		\begin{equation}\label{ball rho}
		B_\rho(v,I) := \left\{ (w, J) \in \mathbb{R}^{3+}: \sqrt{ \frac{1}{2} \left| v-w \right|^2 + \frac{(I+J)}{m} } \leq \rho  \right\},
		\end{equation}
		for any $(v,I) \in \mathbb{R}^{3+}.$\\
		
		Moreover,  the following useful pointwise estimates are sufficient to connect the weight convolution factor in \eqref{g gamma} to the Lebesgue weights \eqref{brackets2}.  On one hand,    by  Minkowski-type inequalities and the concavity of $\gamma/2$-power functions, $0<\gamma\leq 2$, we obtain the following crucial estimate from below for  the integrand \eqref{g gamma},
		\begin{multline}\label{g gamma below}
		\tB(|v-w|,I,J) 
		\geq
		\min\left\{ 1, 2^{1-\gamma} \right\} \left| v \right|^\gamma - \left| w \right|^\gamma  + 2^{\frac{\gamma}{2} -1}  \left( \left( \frac{I}{m}  \right)^{\gamma/2} + \left( \frac{J}{m}  \right)^{\gamma/2} \right)
		\\
		\geq
		2^{\gamma/2 -1} \left( \left( \frac{1}{{2}}\left| v \right|^2\right)^{\gamma/2}  + \left( \frac{I}{m}  \right)^{\gamma/2} \right)
		-  2^{\gamma/2} \left(  \left( \frac{1}{{2}}\left| w \right|^2\right)^{\gamma/2}   +  \left( \frac{J}{m}  \right)^{\gamma/2}\right)
		\\
		\geq
		2^{\gamma -2} \left( \frac{1}{2}\left| v \right|^2  + \frac{I}{m} \right)^{\gamma/2} - 2 \left( \frac{1}{2} \left| w \right|^2  +  \frac{J}{m}  \right)^{\gamma/2}.
		\end{multline}	
		In addition, clearly,	 the concavity of the $\gamma/2$ root function, $0<\gamma\leq 2$, implies for the integrand factor \eqref{g gamma},
		\begin{equation}\label{g gamma 2}
		\tB(|v-w|,I,J) \geq \left(  \frac{1}{2} \left|v-w\right|^2 + \frac{(I+J)}{m} \right)^{\gamma/2}.
		\end{equation}	
		
		On the other hand,  the base of the power expression from the right hand side  in \eqref{g gamma 2} can be estimated from above by invoking 
		 the local Jensen's inequality for convex $k$-power functions}, for any real valued number $k\ge1$,

\begin{equation}\label{pomocna 0}
	\begin{split}
	\left( \frac{1}{2} \left| v-w \right|^2   + \frac{(I+J)}{m} \right)^k & \leq 2^k \left( \frac{1}{2} \left( \left| v \right|^2 +  \left| w \right|^2 \right) + \frac{(I+J)}{m} \right)^k \\
	&\leq 2^{2k-1} \left(  \left( \frac{1}{2} \left| v \right|^2 + \frac{I}{m} \right)^k +\left( \frac{1}{2} \left| w \right|^2 + \frac{J}{m} \right)^k  \right).
	\end{split}
	\end{equation}

	These constructions enable us to first obtain preliminary estimates over the semi-sphere $B_S(v,I)$  for  a parameter $S$ that depends on the data parameters  $ \El$, $ \Mu$, $\D$, $\pl$ and the  radius $\rho_*$  of a   semi-sphere  $B_{\rho_*}\!(0,0)$.    Such estimates will allow us  to control   the integral over  $\mathbb{R}^{3+} $,   by splitting the integration domain  into the semi-sphere  $B_{\rho_*}\!(0,0) $ and its complement  for a good choice of radius $\rho_*$ depending on the data, namely $\rho_* = \rho_*(\Ml, \Mu, \Eu, \gamma)$.  Estimates for the integration over the complement $B^c_{\rho_*}\!(0,0) $ will fix the value of   $\rho_*$ that ensures the lower bound \eqref{lower bound lemma}  with an explicitly calculated constant $\clb$ as a function of the data.\\

	Since the case $\gamma=0$ is trivial, we just consider  $\gamma \in (0,2]$. Setting the shorter notation for our convolution notation in  \eqref{lower bound lemma}  by $\Lambda_\gamma(v,I)$,  and using the inequality \eqref{g gamma 2} for the integrands, we obtain 
	\begin{equation}\label{pomocna 1}
	\begin{split}
	\Lambda_\gamma(v,I) &:= \int_{\mathbb{R}^{3+}}  f(t, w, J) \left( \left| v-w \right|^\gamma + \left( \frac{I+J}{m}  \right)^{\gamma/2} \right) \mathrm{d}J \, \mathrm{d}w
	\\
	&  > \int_{\mathbb{R}^{3+}}  f(t, w, J) \left( \frac{1}{2}\left| v-w \right|^2 +  \frac{(I+J)}{m}  \right)^{\gamma/2} \mathrm{d}J \, \mathrm{d}w =: \bar{\Lambda}_\gamma(v,I).
	\end{split}
	\end{equation}
	Then we  make of use  the semi-sphere defined in \eqref{ball rho} with an arbitrary  radius $S>0$   to  connect the integrals $\bar{\Lambda}_\gamma(v,I)$ and $\bar{\Lambda}_2(v,I)$ which is obtained by setting $\gamma=2$. Namely,  by splitting the integration domain   $\mathbb{R}^{3+}=B_S(v,I) \cup B^c_S(v, I)$   and since the integral $\bar{\Lambda}_\gamma(v,I)$ with respect to the whole space  $\mathbb{R}^{3+}$ is bigger than the same integral with respect to its one  part $B_S(v,I)$, and  noting that 
	$$\left( \frac{1}{2}\left| v-w \right|^2 +  \frac{(I+J)}{m}  \right)^{\gamma/2-1} \geq S^{\gamma-2}\quad \text{ on the semi-sphere} \ B_S(v,I),$$ we obtain
	\begin{equation}\label{pomocna 1.1}
	\begin{split}
	&\bar{\Lambda}_\gamma(v,I) \geq  S^{\gamma-2} \int_{B_S(v,I)}  f(t, w, J) \left( \frac{1}{2} \left| v-w \right|^2+  \frac{(I+J)}{m}  \right) \mathrm{d}J \, \mathrm{d}w, \\
	& \geq  S^{\gamma-2} \left(   \bar{\Lambda}_2(v,I)    -\int_{B^c_S(v,I) }  f(t, w, J) \left( \frac{1}{2} \left| v-w \right|^2+ \frac{(I+J)}{m} \right) \mathrm{d}J \, \mathrm{d}w \right).
	\end{split}
	\end{equation}
	For the first integral $\bar{\Lambda}_2(v,I) $ we develop the square $|v-w|^2$ and use the assumption \eqref{momentum}, as much as the integrability properties on $f$ from conditions \eqref{mass energy}, to bound it from below by the Lebesgue brackets, 
	\begin{align}\label{Lambda 2}
	\bar{\Lambda}_2(v,I) &=   \int_{\mathbb{R}^{3+}}  f(t, w, J) \left( \frac{1}{2} \left( \left| v \right|^2 + \left| w \right|^2\right) +  \frac{(I+J)}{m}  \right)\mathrm{d}J \, \mathrm{d}w \nonumber \\
	&  \geq {\Ml } \left( \frac{1}{2} \left| v \right|^2 +  \frac{I}{m}  \right) +  \El \geq  \El.
	\end{align}	
	On the other hand, for the second integral we invoke integrability conditions \eqref{mass energy} and \eqref{moment epsilon}, 
	and the  pointwise estimate from  \eqref{pomocna 0}, applied to $k=\frac{2+\pl}2>1$, 
	\begin{equation}\label{pomocna 1.2}
	\left( \frac{1}{2} \left| v-w \right|^2   + \frac{(I+J)}{m} \right)^\frac{2+\pl}{2} \leq 2^{1+\pl} \left(  \left( \frac{1}{2} \left| v \right|^2 + \frac{I}{m} \right)^\frac{2+\pl}{2}  +\left( \frac{1}{2} \left| w \right|^2 + \frac{J}{m} \right)^\frac{2+\pl}{2}    \right)
	\end{equation}
	to obtain the control from above of the  integral on the complement of the semi-sphere   	$B^c_S(v,I) $ by the  Lebesgue weight $\langle v,I \rangle^{2+\pl}$, since
	\begin{align}\label{pomocna 1.3}
	&\int_{B^c_S(v,I) } f(t, w, J) \left(\frac{1}{2} \left| v-w \right|^2+  \frac{(I+J)}{m}  \right) \mathrm{d}J \, \mathrm{d}w\nonumber \\
	&\qquad \leq \frac1{S^{\pl} } \int_{B^c_S(v,I) }  \!f(t, w, J) \left(  \frac{1}{2}  \left| v-w \right|^2+ \frac{(I+J)}{m}  \right)^{\frac{2+\pl}2} \mathrm{d}J \mathrm{d}w \nonumber\\
	& \qquad \qquad \leq \frac1{S^{\pl} }  2^{1+\pl} \left( \Mu    \left( \frac{1}{2} \left| v \right|^2 + \frac{I}{m} \right)^\frac{2+\pl}{2} +  \D \right) 
	\nonumber \\ &\qquad\qquad\qquad
	\le \frac 1{S^{\pl}}   2^{1+\pl}    \max\{\Mu, \D \}  \langle v,I \rangle^{2+\pl}.  
	\end{align}	
	Thus, combining \eqref{Lambda 2} and \eqref{pomocna 1.3} for the term in \eqref{pomocna 1.1} we obtain
	\begin{align}\label{pomocna 1.4}
	&\bar{\Lambda}_2(v,I)    -\int_{B^c_S(v,I) } f(t, w, J) \left(\frac{1}{2} \left| v-w \right|^2+ \frac{(I+J)}{m}  \right) \mathrm{d}J \, \mathrm{d}w  \nonumber\\
	&  \ \ \ \ge  \El  -   \frac{  2^{1+\pl}  }{S^{\pl}}     \max\{\Mu, \D \} \langle v,I \rangle^{2+\pl}. 
	\end{align}
	Therefore, in order to end the argument, we need to choose $S$ and  a radius $\rho_*$ for  the     semi-spheres $B_S(v,I)$ and $B_{\rho_*}\!(0,0)$, such that the lower bound for estimate \eqref{pomocna 1.4} is strictly positive.   That means , it is enough to choose $S$ such that 
	\begin{equation*}\label{ }
	\frac{   2^{1+\pl}  }{S^{\pl}}    \max\{\Mu, \D \}  \langle v,I \rangle^{2+\pl} \le  \frac{     2^{1+\pl}     }{S^{\pl}}     \max\{\Mu, \D \}  {  (1+\rho_*^2)^{\frac{2+\pl}2}  } < \frac{ \El }{ 8}
		, 
	\end{equation*}
	for any $(v,I) \in B_{\rho_*}\!(0,0)$, 
	which amounts to choose $S$ 
	\begin{equation}\label{pomocna 1.5}
	S (\rho_*)>    \left(    2^{4+\pl}    \frac{  \max\{\Mu, \D \} }{\El}  (1+\rho_*^2)^{\frac{2+\pl}2}\right)^{\frac1{\pl}} > 1\, . 
	\end{equation}
	Hence,  for $\bar{\Lambda}_\gamma(v,I)  $  from \eqref{pomocna 1.1} we obtain  
	\begin{equation}\label{estimate lb second}
	{\Lambda}_\gamma(v,I) > \bar{\Lambda}_\gamma(v,I) >  \frac{\El }{ 8} \, S(\rho_*)^{\gamma-2}, \quad  \text{for}  \ (v,I) \in B_{\rho_*}\!(0,0).
	\end{equation}

	Now, we are in good conditions to estimate the convolutional form $\Lambda_\gamma(v,I)$  as defined in \eqref{pomocna 1}, or equivalently \eqref{lower bound lemma}, form below in terms of 
	for any $(v,I)\in B_{\rho_*}\!(0,0)$. 
	
	In particular, involving the estimate from below \eqref{g gamma below}  of the integrand $\tB(|v-w|,I,J)$ from \eqref{g gamma}   of  $\Lambda_\gamma(v,I)$, we obtain that the integral  $\Lambda_\gamma(v,I)$ itself is bounded below, after using the  assumption \eqref{mass energy}, by 
	\begin{align}\label{pomocna 1.6}
	\Lambda_\gamma(v, I) &= \int_{\mathbb{R}^{3+}}  f(t, w, J) \left( \left| v-w \right|^\gamma + \left( \frac{I+J}{m}  \right)^{\gamma/2} \right) \mathrm{d}J \, \mathrm{d}w \nonumber \\
	&\quad  \geq   \Ml \left(  2^{\gamma -2} \left( \frac{1}{2}\left| v \right|^2  + \frac{I}{m} \right)^{\gamma/2} \right) \\
	&\qquad  -  2 \int_{\mathbb{R}^{3+}}  f(t, w, J) \left( \frac{1}{2} \left| w \right|^2  +  \frac{J}{m}  \right)^{\gamma/2} \mathrm{d}J \, \mathrm{d}w. \nonumber
	\end{align}
	Then,  the   integral  by the negative sign is easily estimated  from above by splitting the integration domain $\mathbb{R}^{3+} = B_1(0,0) \cup B^c_1(0,0)$ 	by 
	\begin{multline}\label{pomocna 1.7}
	\int_{\mathbb{R}^{3+}}  f(t, w, J) \left( \frac{1}{2} \left| w \right|^2  +  \frac{J}{m}  \right)^{\gamma/2} \mathrm{d}J \, \mathrm{d}w
	\\	\leq 
	\int_{  B_1(0,0)  }  f(t, w, J)  \left( \frac{1}{2} \left| w \right|^2  +  \frac{J}{m}  \right)^{\gamma/2} \mathrm{d}J \, \mathrm{d}w 
	\\
	+ \int_{ B_1^c(0,0)}  f(t, w, J)  \left( \frac{1}{2} \left| w \right|^2  +  \frac{J}{m}  \right)^{\gamma/2}  \mathrm{d}J \, \mathrm{d}w.
	\end{multline}
	The first integral is on $B_1(0,0)$, which implies $\left( \frac{1}{2} \left| w \right|^2  +  \frac{J}{m}  \right) \leq 1$, hence it  
	is controlled by $ \Mu$ using the first inequality  from \eqref{mass energy}.    For the second one on $B^c_1(0,0)$,  since $\gamma\in(0,2]$  then  $\left( \frac{1}{2} \left| w \right|^2  +  \frac{J}{m}  \right)^{\gamma/2} \leq \left( \frac{1}{2} \left| w \right|^2  +  \frac{J}{m}  \right)$,
	and  the second estimate in \eqref{mass energy}  shows that is  bounded by $ \Eu$.  Thus,  the integral from \eqref{pomocna 1.6} is controlled from above by 
	\begin{equation*}
	\int_{\mathbb{R}^{3+}}  f(t, w, J)  \left( \frac{1}{2} \left| w \right|^2  +  \frac{J}{m}  \right)^{\gamma/2} \mathrm{d}J \, \mathrm{d}w \leq  \Mu + \Eu. 
	\end{equation*}	
	It follows then, that  the  integral $\Lambda_\gamma(v,I)$, is controlled from below  after using the  previous considerations, as follows 
	\begin{equation}\label{pomocna 11.1}
	\Lambda_\gamma(v,I)  \geq   2^{ \gamma -2}  { \Ml} \left( \frac{1}{2}\left| v \right|^2 +\frac{I}{m}  \right)^{\gamma/2}  -  2 \left(\Mu + \Eu\right).
	\end{equation}
	
	So now we need to choose $\rho_*$  to be large enough such that whenever	  $(v, I) \in B^c_{\rho_*}\!(0,0)$,    the following condition holds
	\begin{equation}\label{estimate lb first}
	\Lambda_\gamma(v,I) \geq \frac{ \Ml}{ 8}   \,  \left( \frac{1}{2}\left| v \right|^2 +\frac{I}{m}  \right)^{\gamma/2}, \quad \text{for} \ (v, I) \in B^c_{\rho_*}\!(0,0).
	\end{equation}
	This amounts to take
	\begin{equation}\label{rho s}
	\rho_* = \left(  \frac{  2 \left(\Mu + \Eu\right) }{{ \Ml} \left( 2^{\gamma -2} -\frac{1}{ 8} \right) } \right)^{1/\gamma}  \geq 1.
	\end{equation}
	
	We conclude the proof by gathering  estimates \eqref{estimate lb first} and \eqref{estimate lb second},
	\begin{equation}\label{Lambda mod 2}
	\begin{split}
	\Lambda_\gamma(v,I) &\geq  \left(  \frac{ \El}{8} \, S^{\gamma-2}   \mathbbm{1}_{B_{\rho_*}\!(0,0)}(v, I) + \frac{ \Ml}{8} \left( \frac{1}{2}\left| v \right|^2 +\frac{I}{m}  \right)^{\gamma/2}\mathbbm{1}_{ B^c_{\rho_*}\!(0,0)}(v, I) \right)\\
	&\geq   \frac{ \min\left\{ \Ml, \El \right\}}{8} \, S^{\gamma-2}  \left(     \mathbbm{1}_{B_{\rho_*}\!(0,0)}(v, I) +  \left( \frac{1}{2}\left| v \right|^2 +\frac{I}{m}  \right)^{\gamma/2} \mathbbm{1}_{ B^c_{\rho_*}\!(0,0)}(v, I) \right),
	\end{split}
	\end{equation}
	the last inequality is because $S=S(\rho_*)\geq1$  from \eqref{pomocna 1.5}.
	Therefore, it is possible to  find an explicit  constant $\clb>0$ such that the lower bound \eqref{lower bound lemma} holds. \\

	In the sequel we construct the  constant $\clb$ that can be view as an analog to Poincar\'e constant in the classical Sobolev embedding theorem in connection to the diffusion problems.  For   the first term in \eqref{Lambda mod 2}  we obtain
	\begin{equation}\label{clb 1}
	\langle v, I \rangle^2 = 	1 + \frac{1}{2}\left| v \right|^2  + \frac{I}{m} \leq 1 + \rho_*^2, \quad  \text{for} \ (v, I) \in B_{\rho_*}\!(0,0).
	\end{equation}	
	The second one  corresponding to $(v, I) \in B^c_{\rho_*}\!(0,0)$ with  $\rho_*\geq1$ by \eqref{rho s}, it follows
	\begin{equation*}
	\frac{1}{2}\left| v \right|^2  + \frac{I}{m} \geq \rho_*^2 \geq \frac{1}{\rho_*^2},
	\end{equation*}
	or
	\begin{equation}
	\rho_*^2 \left( \frac{1}{2}\left| v \right|^2  + \frac{I}{m}  \right) \geq 1.
	\end{equation}
	Therefore, we get  
	\begin{equation}\label{clb 2}
	\langle v, I \rangle^\gamma   \leq \left(1+\rho_*^2\right)^{\gamma/2}  \left( \frac{1}{2}\left| v \right|^2  + \frac{I}{m}  \right)^{\gamma/2},\quad  \text{for} \   (v, I) \in B^c_{\rho_*}\!(0,0).
	\end{equation}	
	
	Gathering \eqref{clb 1} and \eqref{clb 2}, we obtain the following estimate,
	\begin{multline}
	\langle v, I \rangle^\gamma   = \langle v, I \rangle^\gamma  	 \left(  \mathbbm{1}_{B_{\rho_*}\!(0,0)}(v, I) +   \mathbbm{1}_{ B^c_{\rho_*}\!(0,0)}(v, I)   \right)
	\\
	\leq \left( 1+ \rho_*^2 \right)^{\gamma/2} \left(   \mathbbm{1}_{B_{\rho_*}\!(0,0)}(v, I) +    \left( \frac{1}{2}\left| v \right|^2  + \frac{I}{m}  \right)^{\gamma/2}\mathbbm{1}_{ B^c_{\rho_*}\!(0,0)}(v, I)    \right).
	\end{multline}
	Therefore, going back to the final estimate \eqref{Lambda mod 2}, we obtain,
	\begin{equation}
	\begin{split}
	\Lambda \geq   \frac{ \min\left\{ \Ml, \El \right\}}{8} \,  \frac{ S(\rho_*)^{\gamma-2}}{\left( 1+ \rho_*^2 \right)^{\gamma/2} } \langle v, I \rangle^\gamma =: c_{lb}  \langle v, I \rangle^\gamma\, .
	\end{split}
	\end{equation}
Using  inequalities \eqref{pomocna 1.5}  and \eqref{rho s}  we obtain the final expression for the  constant $c_{lb}$ as announced in \eqref{clb},  which completes the proof of Lemma~\ref{lemma lower bound}. 
\end{proof}

The following corollary holds immediately by the definition of the Banach space  $L^1_k(\mathbb{R}^{3+})$ of integrable functions with respect to the Lebesgue weight, as defined in \eqref{brackets}, \eqref{space L_k^1} and \eqref{norm}.
\begin{corollary}[Coercive Estimate]\label{coercive_bound} The loss collision operator from \eqref{BE}, acting on  any function $f \in L_{k}^1$, for $k\ge1+\gamma/2$, $ \gamma\in[0,2]$, satisfies
	\begin{equation}\label{coercive_estimate}
	 \int_{\mathbb{R}^{3+}} Q^-(f,f) \, \langle v, I \rangle^{2k} \ \mathrm{d}I \mathrm{d}v  
	 \geq \   \kappa^{lb} \clb\, \|f\|_{L^1_{k+\frac{\gamma}2}},  
	\end{equation}
	with $\clb$ independent of  moment order $k$  stated in \eqref{clb} and $\kappa^{lb}$ is from \eqref{kappas}.
\end{corollary}

\begin{proof} 
For   the transition function $\mathcal{B}$  satisfying the assumption \ref{trans prob rate ass}, the following lower  bound for the collision frequency defined in \eqref{collision frequency} holds
\begin{equation*}
\nu[f](v,I) \geq \kappa^{lb} \int_{{\mathbb{R}^{3+}} } {f _*}  \tilde{B} \mathrm{d}I_* \mathrm{d}v_* \geq \kappa^{lb} c_{lb} \langle v, I \rangle^{\gamma},
\end{equation*}
with the constant $\kappa^{lb}$ coming out from the integration over the compact set $K$ is introduced in \eqref{kappas} and the last inequality is from the lower bound \eqref{lower bound lemma}. Now the  loss collision operator  \eqref{loss op} can be bounded from below as follows
\begin{equation*}
\begin{split}
\int_{\mathbb{R}^{3+}} Q^-(f,f) \, \langle v, I \rangle^{2k} \ \mathrm{d}I \mathrm{d}v &= \int_{(\mathbb{R}^{3+})^2 \times \K} f(t,v,I) \, \nu[f](v,I) \, \langle v, I \rangle^{2k} \ \mathrm{d}I \mathrm{d}v \\
&\ge \ \kappa^{lb} \, \clb\,  \int_{\mathbb{R}^{3+}} f\, \langle v, I \rangle^{2k+\gamma} \mathrm{d}I \mathrm{d}v  \ = \   \kappa^{lb} \clb\, \|f\|_{L^1_{k+\frac{\gamma}2}},  
\end{split}
\end{equation*}
which concludes the proof.
\end{proof}

\section{Fundamental lemmas for  the gain operator}\label{Sec: fund lemmas}

This section is devoted to obtain the control of the positive contributions of the collision operator moments for the polyatomic case. The following  results   are obtained by gathering  identities and estimates that exalt the nature of binary collisional operator in weak form as a dissipative single species particle  mixing operator whose collisional transition probability component varying on a compact manifold  is controlled  by an averaging  operator on such manifold whose volume remains invariant by the particle. 

The most fundamental result, The Polyatomic Compact Manifold Averaging Lemma,  is presented in Lemma~\ref{lemma povzner}.
Such calculation  is performed in the weak formulation of the polyatomic collisional form after presented the conservation of energy decomposition identity
that enables this proof of  Lemma~\ref{lemma povzner}.  This Lemma can be viewed as a sharper form of the Povzner Lemma developed by \cite{Bob97} for the classical Boltzmann equation for hard spheres in three dimensions, by \cite{GambaBobPanf04} for inelastic hard spheres in three dimensions \cite{GambaPanfVil09} for hard potentials in any dimension above or equal to three, and recently revisited in  \cite{Alonso-IG-BAMS}. 
The presentation that follows in unedited for the Boltzmann equation  model   of polyatomic gases. 

\subsection{The energy identity}

We first define the total energy of the two colliding molecules using the Lebesgue weight \eqref{brackets}.

\begin{definition}[The total energy in the Lebesgue weight form] Let $v'$, $v'_*$, $I'$ and $I'_*$ be functions of $v$, $v_*$, $I$, $I_*$, $r$, $R$ and $\sigma$ as given in \eqref{velocities} and \eqref{micro int energies}. Then we define the total energy in the Lebesgue weight \eqref{brackets} form as follows
	\begin{equation}\label{E bracket}
	\begin{split}
	E^{\langle \rangle} &:= \langle v, I \rangle^2 + \langle v_*, I_* \rangle^2 = \langle v', I' \rangle^2 + \langle v'_*, I'_* \rangle^2 = 2 +  \left| V \right|^2  + \frac{E}{m},
	\end{split}
	\end{equation}
	with $E$ from \eqref{micro CL energy mass-rel vel}.
\end{definition}

In order to encode behavior of a polyatomic gas, we first need to understand  energy recombination during a collision process, using  transformations \eqref{velocities} and \eqref{micro int energies}.  
This knowledge is crucial for  expressing pre-collisional quantities  $\langle v', I' \rangle^2$ and $\langle v'_*, I'_* \rangle^2 $  in terms of particular partitions of  the total energy, as shows the following Lemma.

\begin{lemma}[Energy Identity Decomposition]\label{lemma energy-identity} 	Let $v'$, $v'_*$, $I'$ and $I'_*$ be defined in collision transformations  \eqref{velocities} and \eqref{micro int energies}.
There exists convex conjugate factors $p = p(v, v_*, I, I_*, R)$ and $q = q(v, v_*, I, I_*,  R)$, i.e. $p+q=1$, and a function $\lambda=\lambda(v, v_*, I, I_*, R)$ such that  the following representation holds
\begin{equation*}
\langle v', I' \rangle^2  = E^{\langle \rangle} \left( p+ \lambda  \hat{V} \cdot \sigma\right),  \qquad \langle v'_*, I'_* \rangle^2  = E^{\langle \rangle}  \left( q -  \lambda  \hat{V} \cdot \sigma \right).
\end{equation*}
Moreover, this representation preserves  the total  molecular energy,
	\begin{align}\label{kin+int-energy}
	\langle v', I' \rangle^2  + \langle v'_*, I'_* \rangle^2 = E^{\langle \rangle} \left( (p+ \lambda  \hat{V} \cdot \sigma) +( q -  \lambda  \hat{V} \cdot \sigma) \right)\equiv
	E^{\langle \rangle}.
	\end{align}
\end{lemma}
\begin{proof}
	We consider  partitions of the energy $E^{\langle \rangle}$ obtained  by introducing  convex combinations associated to functions $\Theta$ and $\Sigma$ that may depend on $v, v_*, I, I_*$ and $R$,  as follows
	\begin{enumerate}
		\item[(i)] for $\Theta \in [0,1]$ we have
		\begin{equation}\label{Theta}
		\Theta E^{\langle \rangle} = 1+ \left| V \right|^2  \quad \Rightarrow \quad \left( 1 - \Theta \right) E^{\langle \rangle}  = 1 + \frac{E}{m},
		\end{equation}
		\item[(ii)] for $\Sigma \in [0,1]$ we get
		\begin{equation}\label{Sigma}
		\Sigma	\left( 1 - \Theta \right)  E^{\langle \rangle} = 1+ R \frac{E}{m} \quad \Rightarrow \quad \left( 1 - \Sigma \right)  \left( 1 - \Theta \right) E^{\langle \rangle}  = (1-R) \frac{E}{m}.
		\end{equation}
	\end{enumerate}
	Now, using collisional rules \eqref{velocities} and \eqref{micro int energies} yield the associated Lebesgue weights for the calculation of total molecular energy of the postcollisional (primed) states
	\begin{equation*}
	\begin{split}
	\langle v', I' \rangle^2 &= 1 + \frac{1}{2} \left| V \right|^2 + \frac{1}{2} R \frac{ E}{m} + \sqrt{\frac{R E}{m}} \left|V\right| \hat{V} \cdot \sigma +  r (1-R) \frac{E}{m},\\
	\langle v'_*, I'_* \rangle^2 &= 1 + \frac{1}{2} \left| V \right|^2 + \frac{1}{2} R \frac{ E}{m} - \sqrt{\frac{R E}{m}} \left|V\right| \hat{V} \cdot \sigma +  (1-r) (1-R) \frac{E}{m},
	\end{split}
	\end{equation*}
	which can be rewritten in terms of functions $\Theta$ and $\Sigma$  as in \eqref{Theta}-\eqref{Sigma}, the parameter $r \in [0,1]$ from \eqref{micro int energies}  and the dot product $\hat{V} \cdot \sigma$  as follows
	\begin{equation}\label{post-coll}
	\begin{split}
	\langle v', I' \rangle^2 &= E^{\langle \rangle}  \left( \frac{1}{2} \Theta   + \frac{1}{2} \Sigma (1-\Theta) +r (1-\Sigma)(1-\Theta) \right) \\ & \qquad \qquad \qquad +  \sqrt{(\Theta E^{\langle \rangle} - 1) (\Sigma (1-\Theta) E^{\langle \rangle} -1) }  \,   \hat{V} \cdot \sigma,\\
	\langle v'_*, I'_* \rangle^2 &= E^{\langle \rangle}  \left( \frac{1}{2} \Theta   + \frac{1}{2} \Sigma (1-\Theta) +(1-r) (1-\Sigma)(1-\Theta) \right) \\ & \qquad \qquad \qquad -  \sqrt{(\Theta E^{\langle \rangle} - 1) (\Sigma (1-\Theta) E^{\langle \rangle} -1) }  \,   \hat{V} \cdot \sigma\, .
	\end{split}
	\end{equation}
	Now set  the convex factors from \eqref{post-coll},  to be 
	\begin{equation*}
	\begin{split}
	p &:= \frac{1}{2} \Theta   + \frac{1}{2} \Sigma (1-\Theta) +r (1-\Sigma)(1-\Theta)  = \frac{s}{2} + r (1-s),\\
	q &:= \frac{1}{2} \Theta   + \frac{1}{2} \Sigma (1-\Theta) +(1-r) (1-\Sigma)(1-\Theta)  =  \frac{s}{2} + (1-r)(1-s),
	\end{split}
	\end{equation*}
	where the dependence of $p$ and $q$ upon velocities $v, v_*$, internal energies $I, I_*$ and variable $R$ is through the function $s:=s(v, v_*, I, I_*, R)$ defined with
	\begin{equation}\label{s}
	s= \Theta   +  \Sigma (1-\Theta) \quad \Rightarrow \quad (1-s) = (1-\Sigma)(1-\Theta),
	\end{equation} 
	with $\Theta$, $\Sigma$ from \eqref{Theta}-\eqref{Sigma}. Since $\Theta, \Sigma \in [0,1]$ it also follows
	\begin{equation}\label{s range}
	s \in [0,1], \quad \text{for any} \quad v, v_* \in \mathbb{R}^3, \ I,I_* \in [0,\infty), \ \text{and} \ R\in[0,1].
	\end{equation}
	Clearly  $p$ and $q$ add up to unity.
	In addition set  
	\begin{equation}\label{lambda_factor}
	\lambda :=  \sqrt{(\Theta E^{\langle \rangle} - 1) (\Sigma (1-\Theta) E^{\langle \rangle} -1) }. 
	\end{equation}
	
	Hence, adding the  two left hand sides of  identities from \eqref{post-coll}, the conservation of the total, i.e. kinetic plus internal molecular energy is given by 
	\begin{align*}
	\langle v', I' \rangle^2  + \langle v'_*, I'_* \rangle^2 &= E^{\langle \rangle} (p+\lambda \hat{V} \cdot \sigma + q - \lambda \hat{V} \cdot \sigma) = 
	E^{\langle \rangle} (p+q) \\ &=  E^{\langle \rangle} =\langle v, I \rangle^2  + \langle v_*, I_* \rangle^2  ,
	\end{align*}
	if we recall that the total molecular energy for a polyatomic state interacting  (or colliding) pairs $(v,I)$ and  $(v_*,I_*)$ is given by 
	$ E^{\langle \rangle}:=  \langle v, I \rangle^2  + \langle v_*, I_* \rangle^2 $. 
	Thus, the energy identity  \eqref{kin+int-energy} holds.
\end{proof}

\subsection{The Polyatomic Compact Manifold  Averaging  Lemma} 

The energy identity \eqref{kin+int-energy}  allows to find a  dissipation effect of the  collision operator. Namely, we will prove that $k$-th moment of the gain term decreases with respect to $k$, allowing the moment of the same order   $k$ of the loss term to prevail in the dynamics, when sufficiently large order of moments $k$ is taken into account. The decay of the gain term is attained by averaging $k^{\text{th}}$-power of the postcollisional total molecular energies, that is  $\langle v', I' \rangle^{\tk} + \langle v'_*, I'_* \rangle^{\tk}$. Due to an additional variable $I$ in the polyatomic gas model,  the averaging needs to be performed with respect to the compact manifold that contains a domain of the two parameters: \emph{(i)} one angular  parameter (scattering direction) $\sigma$ that splits the kinetic energy on molecular velocities,  \emph{(ii)} one additional parameter $r$  that distributes the total internal energy among colliding molecules, and parameter $R$ that splits the total molecular energy into the kinetic and internal part, whose averaging does not contribute to the decay of the gain part and integration gives the constant. This result can be viewed as an extension of the angular averaging  Povzner lemma  used for classical elastic and inelastic collisional theory for single of multiple mixture of monatomic gases.

\begin{lemma}[The Polyatomic Compact Manifold  Averaging  Lemma]\label{lemma povzner}
	Let $v'$, $v'_*$, $I'$ and $I'_*$ be given as  in \eqref{velocities} and \eqref{micro int energies}. Suppose  that functions $	b(\sigma \cdot \hat{u})$, $\dgu$  and $\egu$ satisfy the integrability conditions \eqref{ass b}, \eqref{ass d}  and \eqref{ass e}, respectively.  Then the following estimate holds
	\begin{multline}\label{povzner estimate}
	\int_{\K}\left(  \langle v', I' \rangle^{\tk}+ \langle v'_*, I'_* \rangle^{\tk} \right) b(\hat{u} \cdot \sigma)\, \dgu \, \varphi_\alpha (r)  \, \egu \, \psi_\alpha(R) \, (1-R) R^{1/2}   \,\mathrm{d}R\, \mathrm{d}r \, \mathrm{d}\sigma   
	\\
	\leq  \mathcal{C}_{\hk} \left( \langle v, I \rangle^2 + \langle v_*, I_* \rangle^2 \right)^{\hk},
	\end{multline}
	with the contracting constant $\mathcal{C}_{\hk}$, that is $\mathcal{C}_k\searrow 0$, as $k\rightarrow \infty$. In addition, there exists a $\bar{k}_*>1$ such that  
	\begin{equation}\label{k* const}
	 \mathcal{C}_{\hk}   <	  \kappa^{lb}, \quad \text{for all} \ k > \bar{k}_*,
	\end{equation}	
	 where $ \kappa^{lb}$ is given in \eqref{kappas}.  
	
	Moreover, when  $b(\sigma \cdot \hat{u}) \in  L^p(S^2; \mathrm{d}\sigma)$ and $\dgu \varphi_\alpha(r) \in L^p([0,1]; \mathrm{d}r)$, $p \in (1,\infty]$,  the contracting constant $ \mathcal{C}_{\hk} $ can be explicitly computed with the  known decay rate, 
		\begin{multline}\label{povz const infty}
			\mathcal{C}_{\hk} \leq 2 \, \CgauR \|b\|_{L^{p}(\mathrm{d}\sigma)}   \|\dgun \varphi_\alpha\|_{L^{p}(\mathrm{d}r)}  (4\pi)^{1/p'} 
			\\ \times
			\left(  \frac{1}{kp'+1} +  \frac{2 kp'}{(kp'+1)(kp'+2)} \left( 1 - \left(\frac{1}{2}\right)^{kp'+2}  \right) \right)^{1/p'}. 
				\end{multline}
			with $p$ and $p'$ being the classical conjugates $\frac{1}{p}+\frac{1}{p'}=1$.
\end{lemma}

\begin{proof}
	In order to prove this Lemma, we  first use energy identity and representation \eqref{post-coll}, that generated the $\lambda$ factor in \eqref{lambda_factor}, 
	\begin{equation*}
	\lambda :=  \sqrt{(\Theta E^{\langle \rangle} - 1) (\Sigma (1-\Theta) E^{\langle \rangle} -1) }. 
	\end{equation*}
	Using the Young inequality we get an estimate
	\begin{equation*}
	\lambda \ \leq \ \frac{1}{2} \left(\Theta E^{\langle \rangle}  + \Sigma (1-\Theta) E^{\langle \rangle} -2 \right)
	\leq  \frac{ \left(\Theta   + \Sigma (1-\Theta) \right)}{2}   E^{\langle \rangle}   =  \frac{s}{2} \, E^{\langle \rangle},
	\end{equation*}
	with the function $s \in [0,1]$ defined in \eqref{s}-\eqref{s range}. Thus,
	\begin{equation*}
	\pm\, \lambda \,  \hat{V} \cdot \sigma \ \leq\   \frac{ \left(\Theta   + \Sigma (1-\Theta) \right)}{2} | \hat{V} \cdot \sigma | \, E^{\langle \rangle}  = \frac{s}{2} \, | \hat{V} \cdot \sigma |\, E^{\langle \rangle} .
	\end{equation*}
	Therefore,   the convex form \eqref{post-coll} can be estimated pointwise 
	\begin{equation}\label{post-coll-est-1}
	\begin{split}
	\langle v', I' \rangle^2 & \leq E^{\langle \rangle}  \left( \left( \Theta   +  \Sigma (1-\Theta) \right) \left(\frac{1+ | \hat{V} \cdot \sigma| }{2} \right)+ (1-\Sigma)(1-\Theta) r \right) \\  & =  E^{\langle \rangle}  \left( s \left(\frac{1+ | \hat{V} \cdot \sigma| }{2} \right)+ (1-s) r \right), \\
	\langle v'_*, I'_* \rangle^2 & \leq  E^{\langle \rangle}  \left(\left( \Theta   +  \Sigma (1-\Theta) \right) \left(\frac{1+| \hat{V} \cdot \sigma| }{2} \right) + (1-\Sigma)(1-\Theta) (1-r) \right)
	\\  & =  E^{\langle \rangle}  \left( s \left(\frac{1+ | \hat{V} \cdot \sigma| }{2} \right)+ (1-s) (1-r) \right).
	\end{split}
	\end{equation}
	Moreover,  since functions 
	\begin{equation}\label{pomocna 111}
	s \left(\frac{1+ | \hat{V} \cdot \sigma| }{2} \right)+ (1-s) r, \quad \text{and} \quad s \left(\frac{1+ | \hat{V} \cdot \sigma| }{2} \right)+ (1-s) (1-r)
	\end{equation}	
	are linear with respect to $s$ and the range of $s$ is $[0,1]$   by \eqref{s range}, it follows that      the maximum of both functions \eqref{pomocna 111} is attained at the boundary i.e. for either $s=0$ or $s=1$. Therefore,  
	we can write
	\begin{equation*}
	s \left(\frac{1+ | \hat{V} \cdot \sigma| }{2} \right)+ (1-s) r 
	\leq \max\left\{ \frac{1+ | \hat{V} \cdot \sigma | }{2},  r \right\}, 
	\end{equation*}
	and 
	\begin{equation*}
	s \left(\frac{1+ | \hat{V} \cdot \sigma| }{2} \right)+ (1-s) (1-r) 
	\leq \max\left\{ \frac{1+ | \hat{V} \cdot \sigma| }{2},  1-r \right\}.
	\end{equation*}
	This allows to estimate \eqref{post-coll-est-1}  as follows
	\begin{equation*}
	\langle v', I' \rangle^2  \leq E^{\langle \rangle} \max\left\{ \frac{1+ | \hat{V} \cdot \sigma | }{2},  r \right\}, \qquad
	\langle v'_*, I'_* \rangle^2 \leq  E^{\langle \rangle}  \max\left\{ \frac{1+ | \hat{V} \cdot \sigma | }{2},  1-r \right\}.
	\end{equation*}
	With these estimates and using the  symmetry properties associated to the partition function with respect to $r$ by the fact $\varphi_\alpha(r) = \varphi_\alpha(1-r)$ and  $\dgu= \dgun(1-r)$, the left-hand side of \eqref{povzner estimate}  becomes 
	\begin{multline}\label{pomocna 2}
	\int_{\K}  \left(  \langle v', I' \rangle^{\tk}+ \langle v'_*, I'_* \rangle^{\tk} \right) b(\hat{u} \cdot \sigma)\,\dgu \, \varphi_\alpha (r)  \, \egu \, \psi_\alpha(R) \, (1-R) R^{1/2}   \,\mathrm{d}R\, \mathrm{d}r \, \mathrm{d}\sigma   
	\\
	\leq 2 \left(E^{\langle \rangle} \right)^{\hk} \int_{\K}  \left( \max\left\{ \frac{1+ | \hat{V} \cdot \sigma | }{2},  r \right\}\right)^{\hk}   b(\hat{u} \cdot \sigma)\, \dgu \, \varphi_\alpha (r)\\ \times    \egu \, \psi_\alpha(R) \, (1-R) R^{1/2}   \,\mathrm{d}R\, \mathrm{d}r \, \mathrm{d}\sigma =: \mathcal{K}.  
	\end{multline}
	It is interesting to note that decay of $\mathcal{K}$ from \eqref{pomocna 2} in $k$ is warranted by the averaging over $\sigma$ and $r$, and not necessarily on $R$, which implies that integration with respect to $R$ comes out as a constant. Thus, using the notation \eqref{meas values}, \eqref{pomocna 2} becomes
	\begin{multline}
	\mathcal{K} \leq 2 \, \CgauR \left(E^{\langle \rangle} \right)^{\hk} \int_{S^2} \int_{0}^1 \left( \max\left\{ \frac{1+ | \hat{V} \cdot \sigma | }{2},  r \right\}\right)^{\hk}  b(\hat{u} \cdot \sigma)   \dgu \, \varphi_\alpha (r) \mathrm{d}r \, \mathrm{d}\sigma, 
	\\
	\leq \mathcal{C}_{\hk}  \left(E^{\langle \rangle} \right)^{\hk},
	\end{multline}
	where we have denoted
	\begin{equation}\label{povzner double int}
	\mathcal{C}_{\hk} = 2 \, \CgauR \sup_{\{\hat{V}\in {S}^{2} ,\hat{u} \in {S}^{2}\}} \int_{S^2} \int_{0}^1 \left( \max\left\{ \frac{1+ |\hat{V} \cdot \sigma| }{2},  r \right\}\right)^{\hk} b(\hat{u} \cdot \sigma) \, \dgu\, \varphi_\alpha (r) \mathrm{d}r \, \mathrm{d}\sigma.
	\end{equation}
	
	The following estimates  of the constant $\mathcal{C}_{\hk} $  are inspired and follow the analog ones described  in \cite{Bob97},  \cite{GambaBobPanf04}, \cite{GambaPanfVil09},   \cite{AlonsoLods}, \cite{Gamba13}, and more recently revisited in \cite{Alonso-IG-BAMS} for the classical space homogeneous Boltzmann binary elastic interacting particle  model for monatomic gases in the case of the angular transition $b(\hat{u} \cdot \sigma)$ just an integrable function on the sphere $S^2$. 
	
	Indeed,  if the angular transition $b(\hat{u} \cdot \sigma) \in L^1(S^2; \mathrm{d}\sigma)$   and the partition function  $\dgu \varphi_\alpha(r) \in L^1([0,1]; \mathrm{d}r)$, then   the  integral $\mathcal{C}_{\hk}$ is monotonically decreasing in $\hk$, without necessarily a specific decay rate in the parameter $\hk$.
	
	This statement follows from expressing $\hat V$ in polar coordinates with zenith $\hat{u}$.  That makes  the integral on the sphere that characterizes $\mathcal{C}_{\hk}$ to be a continuous function on in the vectors $\hat{V}$ and $\hat{u}$, whose integrand,  with respect to the measure  $b(\hat{u} \cdot \sigma) \mathrm{d}\sigma \dgu \varphi_\alpha(r)  \mathrm{d}r$    is strictly decreasing in $k>1$ up to a set of measure zero (namely at $\sigma=\{\pm\hat{V}\}$  or at $r=1$).  Then  $\mathcal{C}_{k_{1}}>
	\mathcal{C}_{k_{2}}$ for any $k_{1}<k_{2}$ follows by  taking the supremum   from the continuity in  $\hat{V}$ and $\hat{u}$ property. 
	In particular, {\co  we} conclude by monotone convergence Theorem that the  constant $\mathcal{C}_{\hk}$ is contracting 
	\begin{equation}\label{C hk}
	\mathcal{C}_{\hk} \searrow 0, \quad \text{as} \ \hk \rightarrow \infty.
	\end{equation}
	In particular, with $\kappa^{ub}$ from \eqref{kappas}, it follows that
	\begin{equation*}
	\mathcal{C}_{\hk} < 2 \kappa^{ub}, \quad \text{for any} \ k>1.
	\end{equation*}
In addition, since $\kappa^{lb} \leq \kappa^{ub}$ with the constant  $\kappa^{lb}$ also defined in \eqref{kappas}, we conclude  by \eqref{C hk} there exists $\bar{k}_*$ such that
\begin{equation}\label{ck kappa lb}
	\mathcal{C}_{\hk} < \kappa^{lb}  \quad \text{for any} \ k>\bar{k}_*,
\end{equation}
	where $\bar{k}_*$ is the smallest $k$ such that  \eqref{ck kappa lb} holds. \\
	
	On the other side, for the classical Boltzmann model for  binary interactions it was also shown   in  \cite{Bob97},  \cite{GambaBobPanf04}, \cite{GambaPanfVil09} and   in  \cite{AlonsoLods}, applied to both elastic or inelastic collisions,  that   in the case  when $b\in L^\infty(S^2; \mathrm{d}\sigma)$  or $  b\in L^p(S^2; \mathrm{d}\sigma)$, for $p>1$, it is possible to calculate the decay rate of  $\mathcal{C}_{\hk}$ as a function of  $\hk$.
	
	For a polyatomic gas, we   describe in detail these two cases for the angular transition $b(\hat u\cdot\sigma)$   and the upper bound of partition function $\dgu \varphi_\alpha(r)  $  for which the  constant $\mathcal{C}_{\hk}$ has an explicit decay with respect to $\hk$ and the integrability rate $p$,  as follows.
	
	\begin{enumerate}
		\item[(i)]   
		If the angular transition rate $b(\hat{u} \cdot \sigma) \in L^p(S^2; \mathrm{d}\sigma)$,   and  the partition function $\dgu \varphi_\alpha(r) \in L^\infty([0,1]; \mathrm{d}r)$ with $p>1$,  then  a straight forward application of the  H\"older inequality to the integral apart of  \eqref{povzner double int} yields the estimate 
		\begin{multline}\label{Ckp begin}
		\mathcal{C}_{\hk}\leq 2 \, \CgauR \|b\|_{L^{p}(\mathrm{d}\sigma)}   \|\dgun \varphi_\alpha\|_{L^{p}(\mathrm{d}r)} \\ \times \left(\int_{S^2} \int_{0}^1 \left( \max\left\{ \frac{1+ |\hat{V} \cdot \sigma| }{2},  r \right\}\right)^{\hk p'} \mathrm{d}r \, \mathrm{d}\sigma \right)^{1/p'},
		\end{multline}
		with $p$ and $p'$ the pairs satisfying $\tfrac{1}{p} + \tfrac{1}{p'} = 1$. 
		In the last integral we change variables $\sigma \mapsto \mu$, $\mu=\hat{V}\cdot\sigma$ and obtain
		\begin{multline*}
		\int_{S^2} \int_{0}^1 \left( \max\left\{ \frac{1+ |\hat{V} \cdot \sigma| }{2},  r \right\}\right)^{\hk p'} \mathrm{d}r \, \mathrm{d}\sigma\\= 4\pi \int_{0}^1 \int_{0}^1 \left( \max\left\{ \frac{1+ \mu }{2},  r \right\}\right)^{\hk p'} \mathrm{d}r \, \mathrm{d}\mu.
		\end{multline*}
		Therefore, from \eqref{Ckp begin}, for $\mathcal{C}_{\hk}$ we get the following estimate
		\begin{equation*}
		\mathcal{C}_{\hk} \leq 2 \, \CgauR \|b\|_{L^{p}(\mathrm{d}\sigma)}   \|\dgun \varphi_\alpha\|_{L^{p}(\mathrm{d}r)}  (4\pi)^{1/p'} \mathcal{C}^p_{\hk},
		\end{equation*}
		with
		\begin{equation*}
		\mathcal{C}^p_{\hk} =  \left(\int_{0}^1 \int_{0}^1 \left( \max\left\{ \frac{1+ \mu }{2},  r \right\}\right)^{\hk p'} \mathrm{d}r \, \mathrm{d}\mu\right)^{1/p'}.
		\end{equation*}
		Then we can explicitly compute this constant $ \mathcal{C}^p_{\hk}$, as shown in Appendix Section \ref{App Const}, by setting  $n=kp'$ in the expression \eqref{CinftyApp}, namely,
		\begin{equation}
		\mathcal{C}^p_{\hk} = \left(  \frac{1}{kp'+1} +  \frac{2 kp'}{(kp'+1)(kp'+2)} \left( 1 - \left(\frac{1}{2}\right)^{kp'+2}  \right) \right)^{1/p'}.
		\end{equation}
		\item[(ii)]    However, if the angular transition rate  $b(\hat{u} \cdot \sigma) \in L^\infty(S^2; \mathrm{d}\sigma)$, that is $b(\hat{u} \cdot \sigma)$ is   bounded on  the sphere $S^2$  and   $\dgu \varphi_\alpha(r) \in L^\infty([0,1]; \mathrm{d}r)$,  the decay rate is faster for $k>1$.  Indeed,  \eqref{povzner double int} can be upper bounded by   
		\begin{multline}\label{Cn preko Cinfty}
		\mathcal{C}_{\hk}
		\leq 	2  \,\CgauR \left\| b \right\|_{L^\infty(\mathrm{d}\sigma)}   \left\| \dgun \, \varphi_\alpha \right\|_{L^\infty(\mathrm{d}r)} \\ \times  \sup_{\{\hat{V}\in {S}^{2} ,\hat{u} \in {S}^{2}\}} \int_{S^2} \int_{0}^1 \left( \max\left\{ \frac{1+ |\hat{V} \cdot \sigma| }{2},  r \right\}\right)^{\hk} \mathrm{d}r \, \mathrm{d}\sigma
		\\
		= 8 \pi  \,\CgauR \left\| b \right\|_{L^\infty(\mathrm{d}\sigma)}   \left\| \dgun \, \varphi_\alpha \right\|_{L^\infty(\mathrm{d}r)} 	\mathcal{C}_{\hk}^\infty,
		\end{multline}
		where, after the change of variables $\sigma \mapsto \mu$, $\mu=\hat{V}\cdot\sigma$ in the last expression \eqref{Cn preko Cinfty},  $\mathcal{C}_{\hk}^\infty$ is given with 
		\begin{equation}\label{Cinfty integral}
		\mathcal{C}_{\hk}^\infty :=\int_{0}^{1}  \int_{0}^1 \left( \max\left\{ \frac{1+ \mu }{2},  r \right\}\right)^{\hk} \mathrm{d}r \,  \mathrm{d}\mu.
		\end{equation}
		This double integral \eqref{Cinfty integral} is computed in the Appendix Section \ref{App Const}, setting $n=k$ in \eqref{CinftyApp}, and the final expression is
		\begin{equation}\label{Cinfty}
		\mathcal{C}_{\hk}^\infty= \frac{1}{\hk+1} +  \frac{2 \hk}{(\hk+1)(\hk+2)} \left( 1 - \left(\frac{1}{2}\right)^{\hk+2}  \right), \quad \hk >1.
		\end{equation}
	\end{enumerate}
	
\end{proof}

Therefore, the total energy identity  \eqref{kin+int-energy} enables to obtain a partial crucial result that controls the averaging on the compact manifold $\K$  of the  $k^{\text{th}}$-power of the postcollisional total molecular energies, that is  $\langle v', I' \rangle^{\tk} + \langle v'_*, I'_* \rangle^{\tk}$    by the   $k^{\text{th}}$-power of the molecular energy, i.e.  $E^{ \langle \rangle k}$ time a factor $\mathcal{C}_k$ is '{\em contracting}",  that means it decays as $k$ grows to infinity. 

This result is an imperative for proving  decay of the $k$-th moment of  collision operator gain term when averaged over the compact manifold $\K$. This fact allows for the corresponding $k$-th moment of the loss term to prevail in the dynamics, when sufficiently large order of moment  $k>\bks$ is taken into account  in order  to ensure \eqref{k* const}.  	It is worthwhile to mention that the Averaging Lemma ensures the existence of such $\bks$, since only the contracting constant $\mathcal{C}_{\hk} $ depends on $k$.  This order of moment $\bks$ needed to guarantee this property is studied in   the upcoming Remark \ref{Sec: k*}.

\begin{remark}[Sufficient moment order  to ensure prevail of the loss term]\label{Sec: k*}
	For the single monatomic species, when the averaging is performed only in the scattering direction $\sigma$,  it was  sufficient to take the order of moment $k>1$ to  prove the dominance of the moment associated to the loss term with respect to the same moment of the gain term,  with $k=1$ corresponding to the energy. In the monatomic gas mixture setting, the value of $k=k_*$ depends on the ratio of mass species, and it is shown that $k_*$ grows as this ratio deviates from $1/2$, where $1/2$ corresponds to the single specie case.
	
	 In the current setting, corresponding to  polyatomic gases, the averaging  is performed over the  compact manifold $K$, with respect to the angular scattering direction $\sigma$, as well as to the   parameters $r$ and $R$ that are arguments of  the partition functions. We seek for a sufficient order of moment $\bks$ which secures \eqref{k* const}, 
	 under the additional assumption of 
	\begin{equation}\label{assy Linf}
	b(\hat{u}\cdot \sigma)\in L^\infty(\mathrm{d}\sigma), \quad \text{and} \quad \dgu \varphi_\alpha(r) \in  L^\infty(\mathrm{d}r),
	\end{equation}
	when we can explicitly compute the constant $\mathcal{C}_{\hk}$  from Lemma \ref{lemma povzner}, as shown in \eqref{povz const infty}. We focus on the three models for transition function $\mathcal{B}$ introduced in Section \ref{Sec: tf models}. 
	
	Note that for all the three models the condition of boundedness of the function $\dgun \varphi_\alpha$ is fulfilled when $\alpha\geq 0$, in which case
	\begin{equation*}
	\left\| \dgun \varphi_\alpha \right\|_{L^\infty(\mathrm{d}r)}= 1.
	\end{equation*}
	Therefore, the condition \eqref{k* const} reduces to 
	\begin{equation}\label{threshold}
  \frac{1}{\hk+1} +  \frac{2 \hk}{(\hk+1)(\hk+2)} \left( 1 - \left(\frac{1}{2}\right)^{\hk+2}  \right) =: 	\mathcal{C}^\infty_{\hk} < \frac{1}{2} \frac{ \cgalr \ \CgalR}{ \CgauR } :=  C^*_{\gamma,\alpha}.
	\end{equation}
 To complete the study, it remains to calculate the constants $\cgalr$, $\CgalR$ and  $\CgauR $ for the three models. To that end, we need to determine multiplying functions $\dgl$, $\egl$ and $\egu$.  For the Model 1 we use constants already calculated in \eqref{model 1 const}, taking $m=1$. The Model 2 takes the bounds from \eqref{model 2 plot}, while for the Model 3 we assume \eqref{model 3 plot}. The results are presented in the Figure \ref{Fig}.
\end{remark}

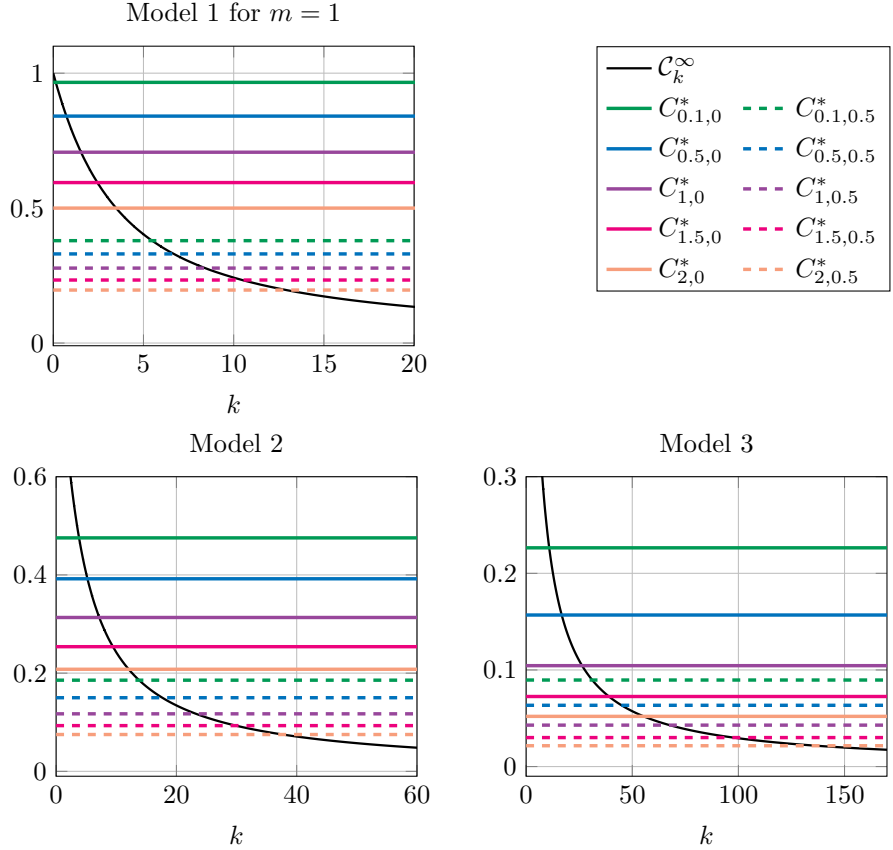
\begin{figure}[H]
	\begin{center}
		\begin{tikzpicture}
		\begin{axis}[
		scale=0.7,
		xmin=0, xmax=20,
		ymin=-0.01, ymax=1.1,
		xmajorgrids,
		ymajorgrids,
		xlabel=$k$,
		align =center, title={Model 1  for $m=1$},
		legend columns=2, 
		legend style={/tikz/column 2/.style={
				column sep=5pt,
			},row sep=0.1cm, at={(2.32,1)}, 
			fill=white,draw=black,nodes=right}]
		\addplot[smooth,color=black,solid,line width=0.2ex]coordinates{
(1.e-6,0.9999997)(0.0150593,0.9962461)(0.0301176,0.9925142)(0.0602342,0.985115)(0.1204674,0.970571)(0.2409337,0.942468)(0.4818664,0.889957)(0.4969247,0.8868295)(0.511983,0.8837194)(0.5420996,0.877551)(0.6023328,0.8654182)(0.7227991,0.8419433)(0.9637319,0.7979638)(0.9800568,0.7951193)(0.9963817,0.7922913)(1.0290315,0.7866843)(1.0943312,0.7756631)(1.2249305,0.7543668)(1.4861292,0.7145679)(1.5024541,0.712197)(1.518779,0.7098393)(1.5514289,0.7051629)(1.6167285,0.6959638)(1.7473279,0.6781609)(2.0085266,0.6447885)(2.0237696,0.6429278)(2.0390127,0.6410762)(2.0694989,0.6374005)(2.1304713,0.6301569)(2.252416,0.6160884)(2.4963055,0.5895316)(2.5115486,0.5879382)(2.5267917,0.5863524)(2.5572779,0.583203)(2.6182502,0.5769927)(2.740195,0.5649155)(2.9840845,0.5420593)(2.9990285,0.5407124)(3.0139726,0.5393713)(3.0438608,0.5367071)(3.1036371,0.5314486)(3.2231898,0.5212042)(3.4622951,0.5017491)(3.4772392,0.5005768)(3.4921833,0.4994095)(3.5220715,0.4970895)(3.5818478,0.4925079)(3.7014005,0.4835719)(3.9405058,0.4665625)(3.9567165,0.4654491)(3.9729272,0.4643405)(4.0053486,0.462138)(4.0701914,0.4577901)(4.1998771,0.4493175)(4.4592484,0.4332173)(4.4754591,0.4322467)(4.4916698,0.4312801)(4.5240912,0.429359)(4.588934,0.4255645)(4.7186197,0.4181608)(4.977991,0.4040573)(4.9931199,0.4032623)(5.0082487,0.4024703)(5.0385065,0.400895)(5.099022,0.3977791)(5.2200531,0.391683)(5.4621152,0.3800082)(5.9462394,0.3585464)(5.9626349,0.3578602)(5.9790304,0.3571766)(6.0118214,0.3558167)(6.0774034,0.3531267)(6.2085674,0.3478625)(6.4708955,0.337776)(6.9955516,0.3192136)(7.0116481,0.3186754)(7.0277446,0.318139)(7.0599376,0.3170714)(7.1243235,0.3149571)(7.2530955,0.3108098)(7.5106394,0.3028272)(8.0257272,0.288009)(8.0407419,0.2875984)(8.0557566,0.2871889)(8.0857859,0.2863733)(8.1458446,0.2847558)(8.2659619,0.2815741)(8.5061967,0.275416)(8.9866661,0.2638634)(10.0286688,0.2418299)(11.0014347,0.2243178)(11.9550641,0.2094361)(12.9897573,0.1953642)(13.9552137,0.1838336)(15.0017339,0.1727761)(16.0291175,0.16314)(16.9872644,0.1550724)(18.0264751,0.1471772)(18.9964491,0.1404995)(20.0474868,0.1339149)(21.079388,0.1280237)(22.0420525,0.1229762)(23.0857807,0.1179346)(24.0602722,0.1135865)(25.0156271,0.109624)(26.0520458,0.1056263)(27.0192277,0.1021499)(28.0674735,0.0986314)(29.0965827,0.0954051)(30.0564551,0.0925804)(31.0973914,0.0897003)(32.0690909,0.0871688)(33.0216538,0.084822)(34.0552805,0.0824144)(35.0196705,0.0802881)(36.0651243,0.0781035)(37.0413413,0.0761683)(37.9984217,0.0743619)(39.036566,0.0724969)(40.0054735,0.0708387)(41.0554448,0.0691253)(42.0862795,0.0675219)(43.0478775,0.0660918)(44.0905393,0.0646081)(45.0639643,0.0632818)(46.0182528,0.0620333)(47.053605,0.0607334)(48.0197206,0.0595685)(49.0668999,0.0583553)(50.0949426,0.0572115)(51.0537486,0.0561843)(52.0936184,0.0551112)(53.0642515,0.0541459)(54.1159483,0.0531374)(55.1485086,0.0521832)(56.1118321,0.0513233)(57.1562195,0.0504225)(58.13137,0.0496095)(59.087384,0.0488375)(60.1244619,0.0480268)(61.0923029,0.0472941)(62.1412078,0.0465249)(63.1709761,0.0457936)(64.1315076,0.045132)(65.173103,0.0444358)(66.1454615,0.0438049)(67.0986835,0.0432037)(68.1329694,0.0425696)(69.0980184,0.0419946)(70.1441313,0.0413886)(71.1210074,0.0408383)(72.078747,0.0403127)(73.1175503,0.0397578)(74.0871169,0.0392535)(75.1377473,0.0387212)(76.1692412,0.0382125)(77.1314983,0.0377498)(78.1748192,0.0372607)(79.1489033,0.0368153)(80.1038508,0.0363889)(81.1398622,0.0359373)(82.1066368,0.0355259)(83.1544752,0.0350905)(84.1831771,0.0346734)(85.1426422,0.0342931)(86.1831711,0.03389)(87.1544632,0.0335222)(88.1066188,0.0331693)(89.1398382,0.0327947)(90.1038208,0.0324528)(91.1488672,0.03209)(92.1246769,0.0317585)(93.08135,0.0314401)(94.1190869,0.0311019)(95.0875871,0.0307927)(96.137151,0.0304646)(97.1675784,0.0301491)(98.1287691,0.0298607)(99.1710235,0.0295541)(100.1440412,0.0292735)(101.0979223,0.0290035)(102.1328672,0.0287162)(103.0985754,0.0284532)(104.1453474,0.0281735)(105.1729828,0.0279042)(106.1313814,0.0276577)(107.1708439,0.0273952)(108.1410696,0.0271546)(109.1923591,0.0268987)(110.224512,0.026652)(111.1874282,0.026426)(112.2314082,0.0261852)(113.2061514,0.0259643)(114.1617581,0.0257513)(115.1984285,0.0255242)(116.1658622,0.0253158)(117.2143598,0.0250938)(118.2437207,0.0248796)(119.2038449,0.0246831)(120.2450329,0.0244734)(121.2169841,0.0242809)(122.1697988,0.0240951)(123.2036773,0.0238967)(124.168319,0.0237144)(125.2140245,0.02352)(126.1904933,0.0233413)(127.1478255,0.0231688)(128.1862215,0.0229845)(129.1553808,0.0228151)(130.2056038,0.0226343)(131.2366903,0.0224596)(132.1985401,0.022299)(133.2414536,0.0221275)(134.2151304,0.0219697)(135.1696706,0.0218172)(136.2052746,0.0216541)(137.1716419,0.0215041)(138.219073,0.0213439)(139.2473675,0.0211889)(140.2064252,0.0210463)(141.2465468,0.0208938)(142.2174316,0.0207535)(143.2693802,0.0206036)(144.3021922,0.0204585)(145.2657675,0.0203249)(146.3104066,0.0201821)(147.2858089,0.0200505)(148.2420747,0.0199232)(149.2794042,0.0197869)(150.247497,0.0196613)(151.2966537,0.0195271)(152.3266737,0.019397)(153.287457,0.0192773)(154.3293041,0.0191491)(155.3019144,0.0190309)(156.2553882,0.0189165)(157.2899258,0.0187939)(158.2552266,0.018681)(159.3015912,0.01856)(160.3288193,0.0184428)(161.2868106,0.0183348)(162.3258657,0.0182191)(163.2956841,0.0181125)(164.3465662,0.0179983)(165.3783118,0.0178876)(166.3408207,0.0177855)(167.3843933,0.0176761)(168.3587292,0.0175752)(169.3139285,0.0174774)(170.3501916,0.0173726)(171.317218,0.0172758)(172.3653082,0.0171722)(173.3942618,0.0170716)(174.3539786,0.0169789)(175.3947593,0.0168795)(176.3663031,0.0167877)(177.3187105,0.0166987)(178.3521816,0.0166032)(179.316416,0.0165151)(180.3617142,0.0164206)(181.3377756,0.0163333)(182.2947004,0.0162487)(183.3326891,0.0161578)(184.301441,0.016074)(185.3512567,0.0159841)(186.3819359,0.0158968)(187.3433783,0.0158162)(188.3858845,0.0157297)(189.3591539,0.0156499)(190.3132868,0.0155724)(191.3484834,0.0154891)(192.3144434,0.0154123)(193.3614671,0.0153298)(194.3893543,0.0152497)(195.3480047,0.0151758)(196.3877189,0.0150964)(197.3581963,0.015023)(198.4097376,0.0149443)(199.4421423,0.0148678)(200.4053102,0.0147972)(201.449542,0.0147214)(202.4245369,0.0146513)(203.3803953,0.0145832)(204.4173176,0.0145101)(205.385003,0.0144425)(206.4337523,0.0143699)(207.463365,0.0142994)(208.4237409,0.0142343)(209.4651807,0.0141643)(210.4373837,0.0140995)(211.3904501,0.0140367)(212.4245803,0.0139691)(213.3894738,0.0139066)(214.4354311,0.0138395)(215.4121516,0.0137774)(216.3697356,0.0137171)(217.4083833,0.0136523)(218.3777943,0.0135923)(219.4282692,0.0135279)(220.4596074,0.0134653)(221.4217089,0.0134074)(222.4648742,0.0133452)(223.4388027,0.0132876)(224.3935947,0.0132317)(225.4294505,0.0131715)(226.3960695,0.0131158)(227.4437523,0.013056)(228.4722986,0.0129978)(229.4316081,0.012944)(230.4719814,0.0128862)(231.4431179,0.0128327)(232.3951179,0.0127806)(233.4281817,0.0127246)(234.3920087,0.0126728)(235.4368996,0.0126171)(236.4125536,0.0125656)(237.3690711,0.0125154)(238.4066525,0.0124615)(239.374997,0.0124116)(240.4244054,0.0123579)(241.4546772,0.0123057)(242.4157122,0.0122574)(243.4578111,0.0122054)(244.4306732,0.0121573)(245.3843987,0.0121105)(246.419188,0.0120601)(247.3847406,0.0120135)(248.431357,0.0119633)(249.4588368,0.0119145)(250.4170799,0.0118693)(251.4563867,0.0118207)(252.4264568,0.0117757)(253.4775908,0.0117273)(254.5095881,0.0116802)(255.4723487,0.0116366)(256.5161731,0.0115897)(257.4907607,0.0115462)(258.4462118,0.0115039)(259.4827267,0.0114584)(260.4500048,0.0114162)(261.4983467,0.0113708)(262.5275521,0.0113266)(263.4875207,0.0112857)(264.5285531,0.0112417)(265.5003487,0.0112009)(266.4530078,0.0111612)(267.4867307,0.0111185)(268.4512168,0.0110789)(269.4967668,0.0110362)(270.4730799,0.0109967)(271.4302565,0.0109583)(272.468497,0.0109169)(273.4375006,0.0108785)(274.4875681,0.0108373)(275.518499,0.0107971)(276.4801931,0.0107598)(277.5229511,0.0107197)(278.4964723,0.0106826)(279.4508569,0.0106464)(280.4863053,0.0106074)(281.452517,0.0105713)(282.4997925,0.0105324)(283.5279314,0.0104945)(284.4868335,0.0104595)(285.5267995,0.0104217)(286.4975287,0.0103866)(287.5493217,0.010349)(288.5819782,0.0103122)(289.5453979,0.0102782)(290.5898814,0.0102415)(291.5651281,0.0102076)(292.5212382,0.0101745)(293.5584122,0.0101388)(294.5263494,0.0101057)(295.5753505,0.0100701)(296.6052149,0.0100355)(297.5658426,0.0100033)(298.6075341,0.0099687)(299.5799889,0.0099366)(300.533307,0.0099053)(301.567689,0.0098716)(302.5328343,0.0098403)(303.5790433,0.0098067)(304.6061158,0.0097739)(305.5639515,0.0097435)(306.602851,0.0097107)(307.5725138,0.0096803)(308.6232403,0.0096476)(309.6548303,0.0096157)(310.6171836,0.0095861)(311.6606006,0.0095543)(312.6347809,0.0095247)(313.5898246,0.0094959)(314.6259322,0.0094649)(315.5928029,0.0094361)(316.6407375,0.0094051)(317.6695355,0.0093749)(318.6290968,0.0093468)(319.6697219,0.0093166)(320.6411102,0.0092886)(321.5933619,0.0092613)(322.6266774,0.0092319)(323.5907562,0.0092046)(324.6358988,0.0091751)(325.6118046,0.0091478)(326.5685739,0.0091212)(327.606407,0.0090925)(328.5750033,0.0090659)(329.6246634,0.0090373)(330.655187,0.0090093)(331.6164738,0.0089833)(332.6588244,0.0089554)(333.6319382,0.0089295)(334.5859155,0.0089042)(335.6209566,0.0088769)(336.5867609,0.0088516)(337.633629,0.0088244)(338.6613606,0.0087978)(339.6198554,0.0087731)(340.659414,0.0087465)(341.6297359,0.0087218)(342.6811216,0.0086953)(343.7133707,0.0086693)(344.676383,0.0086453)(345.7204592,0.0086193)(346.6952985,0.0085953)(347.6510013,0.0085718)(348.687768,0.0085465)(349.6552978,0.008523)(350.7038915,0.0084977)(351.7333487,0.008473)(352.693569,0.00845)(353.7348532,0.0084253)(354.7069006,0.0084024)(355.6598114,0.00838)(356.693786,0.0083559)(357.6585239,0.0083335)(358.7043256,0.0083094)(359.6808905,0.0082869)(360.6383189,0.0082651)(361.6768111,0.0082415)(362.6460665,0.0082196)(363.6963857,0.008196)(364.7275684,0.008173)(365.6895143,0.0081516)(366.732524,0.0081286)(367.7062969,0.0081072)(368.6609333,0.0080864)(369.6966335,0.0080638)(370.6630969,0.0080429)(371.7106241,0.0080204)(372.7390148,0.0079984)(373.6981687,0.007978)(374.7383864,0.007956)(375.7093674,0.0079356)(376.7614121,0.0079136)(377.7943203,0.0078921)(378.7579918,0.0078721)(379.802727,0.0078506)(380.7782255,0.0078306)(381.7345874,0.0078111)(382.7720132,0.0077901)(383.7402021,0.0077705)(384.7894549,0.0077495)(385.8195711,0.0077289)(386.7804506,0.0077098)(387.8223938,0.0076892)(388.7951003,0.0076701)(389.7486703,0.0076514)(390.783304,0.0076313)(391.748701,0.0076126)(392.7951618,0.0075924)(393.822486,0.0075728)(394.7805735,0.0075545)(395.8197247,0.0075348)(396.7896392,0.0075165)(397.8406176,0.0074967)(398.8724593,0.0074774)(399.8350643,0.0074595)(400.8787331,0.0074402)(401.8531652,0.0074223)(402.8084606,0.0074048)(403.8448199,0.0073859)(404.8119424,0.0073684)(405.8601288,0.0073494)(406.8891786,0.007331)(407.8489916,0.0073138)(408.8898684,0.0072953)(409.8615084,0.0072781)(410.8140119,0.0072613)(411.8475792,0.0072432)(412.8119097,0.0072264)(413.8573041,0.0072082)(414.8334617,0.0071913)(415.7904827,0.0071749)(416.8285675,0.0071571)(417.7974156,0.0071406)(418.8473275,0.0071228)(419.8781028,0.0071054)(420.8396413,0.0070893)(421.8822437,0.0070719)(422.8556093,0.0070557)(423.8098383,0.0070399)(424.8451311,0.0070228)(425.8111872,0.007007)(426.8583071,0.0069899)(427.8862904,0.0069732)(428.845037,0.0069577)(429.8848474,0.0069409)(430.855421,0.0069254)(431.9070584,0.0069086)(432.9395593,0.0068922)(433.9028233,0.006877)(434.9471513,0.0068606)(435.9222424,0.0068453)(436.878197,0.0068304)(437.9152154,0.0068143)(438.882997,0.0067994)(439.9318424,0.0067832)(440.9615513,0.0067675)(441.9220234,0.0067529)(442.9635593,0.0067371)(443.9358584,0.0067224)(444.889021,0.0067081)(445.9232474,0.0066926)(446.8882371,0.0066782)(447.9342905,0.0066627)(448.9111072,0.0066483)(449.8687873,0.0066342)(450.9075312,0.006619)(451.8770384,0.0066049)(452.9276094,0.0065896)(453.9590438,0.0065747)(454.9212414,0.0065609)(455.9645029,0.0065459)(456.9385276,0.0065321)(457.8934157,0.0065185)(458.9293677,0.0065039)(459.8960828,0.0064903)(460.9438618,0.0064756)(461.9725042,0.0064612)(462.9319099,0.0064479)(463.9723794,0.0064335)(464.9436121,0.0064202)(465.8957082,0.0064071)(466.9288682,0.006393)(467.8927914,0.0063799)(468.9377784,0.0063657)(469.9135286,0.0063526)(470.8701423,0.0063398)(471.9078197,0.0063259)(472.8762605,0.006313)(473.925765,0.0062991)(474.956133,0.0062855)(475.9172642,0.0062728)(476.9594592,0.0062592)(477.9324174,0.0062465)(478.8862391,0.0062341)(479.9211246,0.0062208)(480.8867733,0.0062083)(481.9334859,0.0061949)(482.9610619,0.0061818)(483.9194011,0.0061696)(484.9588041,0.0061565)(485.9289704,0.0061442)(486.9802005,0.006131)(488.012294,0.0061181)(488.9751507,0.0061061)(490.0190713,0.0060932)(490.9937551,0.0060811)(491.9493023,0.0060694)(492.9859133,0.0060567)(493.9532876,0.0060449)(495.0017257,0.0060321)(496.0310272,0.0060197)(496.991092,0.0060081)(498.0322205,0.0059956)(499.0041124,0.005984)(499.0196731,0.0059838)(499.0352338,0.0059836)(499.0663553,0.0059832)(499.1285982,0.0059825)(499.253084,0.005981)(499.5020557,0.005978)(499.5176164,0.0059779)(499.5331771,0.0059777)(499.5642986,0.0059773)(499.6265415,0.0059766)(499.7510273,0.0059751)(499.7665881,0.0059749)(499.7821488,0.0059747)(499.8132703,0.0059743)(499.8755132,0.0059736)(499.8910739,0.0059734)(499.9066346,0.0059732)(499.9377561,0.0059729)(499.9533168,0.0059727)(499.9688775,0.0059725)(499.9844383,0.0059723)(499.999999,0.0059721)
		};
		\addlegendentry{$\mathcal{C}^\infty_{k}$}
		\addlegendimage{empty legend}
		\addlegendentry{}
		\addplot[smooth,color=ForestGreen,line width=0.3ex] coordinates { (0,0.965936) (200,0.965936)};\label{a=0}
		\addlegendentry{$C^*_{0.1,0}$}
		\addplot[smooth,color=ForestGreen,dashed, line width=0.3ex] coordinates { (0,0.379322) (200,0.379322)};\label{a=0.5}
		\addlegendentry{$C^*_{0.1,0.5}$}
		\addplot[smooth,color=RoyalBlue,line width=0.3ex] coordinates { (0, 0.840896) (200, 0.840896) };
		\addlegendentry{$C^*_{0.5,0}$}
		\addplot[smooth,color=RoyalBlue,dashed, line width=0.3ex] coordinates { (0, 0.330219) (200, 0.330219) };
		\addlegendentry{$C^*_{0.5,0.5}$}
		\addplot[smooth,color=Purple,  line width=0.3ex] coordinates { (0, 0.707107) (200, 0.707107) };
		\addlegendentry{$C^*_{1,0}$}
		\addplot[smooth,color=Purple, dashed, line width=0.3ex] coordinates { (0, 0.27768) (200, 0.27768) };
		\addlegendentry{$C^*_{1,0.5}$}
		\addplot[smooth,color=RubineRed, line width=0.3ex] coordinates { (0, 0.594604) (200, 0.594604) };
		\addlegendentry{$C^*_{1.5,0}$}
		\addplot[smooth,color=RubineRed, dashed, line width=0.3ex] coordinates { (0, 0.2335) (200, 0.2335) };
		\addlegendentry{$C^*_{1.5,0.5}$}
		\addplot[smooth,color=Melon,line width=0.3ex] coordinates { (0,0.5) (200,0.5)};
		\addlegendentry{$C^*_{2,0}$}
		\addplot[smooth,color=Melon,dashed,line width=0.3ex] coordinates { (0,0.19635) (200,0.19635)};
		\addlegendentry{$C^*_{2,0.5}$}
		\end{axis}
		\end{tikzpicture}
		
		\begin{tikzpicture}
		\begin{axis}[
		scale=0.7,
		xmin=0, xmax=60,
		ymin=-0.01, ymax=0.6,
		xmajorgrids,
		ymajorgrids,
		xlabel=$k$,
		align =center, title={Model 2 },
		legend columns=2, 
		legend style={/tikz/column 2/.style={
				column sep=5pt,
			},row sep=0.1cm, at={(1,1)}, 
			fill=white,draw=black,nodes=right}]
		\addplot[smooth,color=black,solid,line width=0.2ex]coordinates{
(1.e-6,0.9999997)(0.0150593,0.9962461)(0.0301176,0.9925142)(0.0602342,0.985115)(0.1204674,0.970571)(0.2409337,0.942468)(0.4818664,0.889957)(0.4969247,0.8868295)(0.511983,0.8837194)(0.5420996,0.877551)(0.6023328,0.8654182)(0.7227991,0.8419433)(0.9637319,0.7979638)(0.9800568,0.7951193)(0.9963817,0.7922913)(1.0290315,0.7866843)(1.0943312,0.7756631)(1.2249305,0.7543668)(1.4861292,0.7145679)(1.5024541,0.712197)(1.518779,0.7098393)(1.5514289,0.7051629)(1.6167285,0.6959638)(1.7473279,0.6781609)(2.0085266,0.6447885)(2.0237696,0.6429278)(2.0390127,0.6410762)(2.0694989,0.6374005)(2.1304713,0.6301569)(2.252416,0.6160884)(2.4963055,0.5895316)(2.5115486,0.5879382)(2.5267917,0.5863524)(2.5572779,0.583203)(2.6182502,0.5769927)(2.740195,0.5649155)(2.9840845,0.5420593)(2.9990285,0.5407124)(3.0139726,0.5393713)(3.0438608,0.5367071)(3.1036371,0.5314486)(3.2231898,0.5212042)(3.4622951,0.5017491)(3.4772392,0.5005768)(3.4921833,0.4994095)(3.5220715,0.4970895)(3.5818478,0.4925079)(3.7014005,0.4835719)(3.9405058,0.4665625)(3.9567165,0.4654491)(3.9729272,0.4643405)(4.0053486,0.462138)(4.0701914,0.4577901)(4.1998771,0.4493175)(4.4592484,0.4332173)(4.4754591,0.4322467)(4.4916698,0.4312801)(4.5240912,0.429359)(4.588934,0.4255645)(4.7186197,0.4181608)(4.977991,0.4040573)(4.9931199,0.4032623)(5.0082487,0.4024703)(5.0385065,0.400895)(5.099022,0.3977791)(5.2200531,0.391683)(5.4621152,0.3800082)(5.9462394,0.3585464)(5.9626349,0.3578602)(5.9790304,0.3571766)(6.0118214,0.3558167)(6.0774034,0.3531267)(6.2085674,0.3478625)(6.4708955,0.337776)(6.9955516,0.3192136)(7.0116481,0.3186754)(7.0277446,0.318139)(7.0599376,0.3170714)(7.1243235,0.3149571)(7.2530955,0.3108098)(7.5106394,0.3028272)(8.0257272,0.288009)(8.0407419,0.2875984)(8.0557566,0.2871889)(8.0857859,0.2863733)(8.1458446,0.2847558)(8.2659619,0.2815741)(8.5061967,0.275416)(8.9866661,0.2638634)(10.0286688,0.2418299)(11.0014347,0.2243178)(11.9550641,0.2094361)(12.9897573,0.1953642)(13.9552137,0.1838336)(15.0017339,0.1727761)(16.0291175,0.16314)(16.9872644,0.1550724)(18.0264751,0.1471772)(18.9964491,0.1404995)(20.0474868,0.1339149)(21.079388,0.1280237)(22.0420525,0.1229762)(23.0857807,0.1179346)(24.0602722,0.1135865)(25.0156271,0.109624)(26.0520458,0.1056263)(27.0192277,0.1021499)(28.0674735,0.0986314)(29.0965827,0.0954051)(30.0564551,0.0925804)(31.0973914,0.0897003)(32.0690909,0.0871688)(33.0216538,0.084822)(34.0552805,0.0824144)(35.0196705,0.0802881)(36.0651243,0.0781035)(37.0413413,0.0761683)(37.9984217,0.0743619)(39.036566,0.0724969)(40.0054735,0.0708387)(41.0554448,0.0691253)(42.0862795,0.0675219)(43.0478775,0.0660918)(44.0905393,0.0646081)(45.0639643,0.0632818)(46.0182528,0.0620333)(47.053605,0.0607334)(48.0197206,0.0595685)(49.0668999,0.0583553)(50.0949426,0.0572115)(51.0537486,0.0561843)(52.0936184,0.0551112)(53.0642515,0.0541459)(54.1159483,0.0531374)(55.1485086,0.0521832)(56.1118321,0.0513233)(57.1562195,0.0504225)(58.13137,0.0496095)(59.087384,0.0488375)(60.1244619,0.0480268)(61.0923029,0.0472941)(62.1412078,0.0465249)(63.1709761,0.0457936)(64.1315076,0.045132)(65.173103,0.0444358)(66.1454615,0.0438049)(67.0986835,0.0432037)(68.1329694,0.0425696)(69.0980184,0.0419946)(70.1441313,0.0413886)(71.1210074,0.0408383)(72.078747,0.0403127)(73.1175503,0.0397578)(74.0871169,0.0392535)(75.1377473,0.0387212)(76.1692412,0.0382125)(77.1314983,0.0377498)(78.1748192,0.0372607)(79.1489033,0.0368153)(80.1038508,0.0363889)(81.1398622,0.0359373)(82.1066368,0.0355259)(83.1544752,0.0350905)(84.1831771,0.0346734)(85.1426422,0.0342931)(86.1831711,0.03389)(87.1544632,0.0335222)(88.1066188,0.0331693)(89.1398382,0.0327947)(90.1038208,0.0324528)(91.1488672,0.03209)(92.1246769,0.0317585)(93.08135,0.0314401)(94.1190869,0.0311019)(95.0875871,0.0307927)(96.137151,0.0304646)(97.1675784,0.0301491)(98.1287691,0.0298607)(99.1710235,0.0295541)(100.1440412,0.0292735)(101.0979223,0.0290035)(102.1328672,0.0287162)(103.0985754,0.0284532)(104.1453474,0.0281735)(105.1729828,0.0279042)(106.1313814,0.0276577)(107.1708439,0.0273952)(108.1410696,0.0271546)(109.1923591,0.0268987)(110.224512,0.026652)(111.1874282,0.026426)(112.2314082,0.0261852)(113.2061514,0.0259643)(114.1617581,0.0257513)(115.1984285,0.0255242)(116.1658622,0.0253158)(117.2143598,0.0250938)(118.2437207,0.0248796)(119.2038449,0.0246831)(120.2450329,0.0244734)(121.2169841,0.0242809)(122.1697988,0.0240951)(123.2036773,0.0238967)(124.168319,0.0237144)(125.2140245,0.02352)(126.1904933,0.0233413)(127.1478255,0.0231688)(128.1862215,0.0229845)(129.1553808,0.0228151)(130.2056038,0.0226343)(131.2366903,0.0224596)(132.1985401,0.022299)(133.2414536,0.0221275)(134.2151304,0.0219697)(135.1696706,0.0218172)(136.2052746,0.0216541)(137.1716419,0.0215041)(138.219073,0.0213439)(139.2473675,0.0211889)(140.2064252,0.0210463)(141.2465468,0.0208938)(142.2174316,0.0207535)(143.2693802,0.0206036)(144.3021922,0.0204585)(145.2657675,0.0203249)(146.3104066,0.0201821)(147.2858089,0.0200505)(148.2420747,0.0199232)(149.2794042,0.0197869)(150.247497,0.0196613)(151.2966537,0.0195271)(152.3266737,0.019397)(153.287457,0.0192773)(154.3293041,0.0191491)(155.3019144,0.0190309)(156.2553882,0.0189165)(157.2899258,0.0187939)(158.2552266,0.018681)(159.3015912,0.01856)(160.3288193,0.0184428)(161.2868106,0.0183348)(162.3258657,0.0182191)(163.2956841,0.0181125)(164.3465662,0.0179983)(165.3783118,0.0178876)(166.3408207,0.0177855)(167.3843933,0.0176761)(168.3587292,0.0175752)(169.3139285,0.0174774)(170.3501916,0.0173726)(171.317218,0.0172758)(172.3653082,0.0171722)(173.3942618,0.0170716)(174.3539786,0.0169789)(175.3947593,0.0168795)(176.3663031,0.0167877)(177.3187105,0.0166987)(178.3521816,0.0166032)(179.316416,0.0165151)(180.3617142,0.0164206)(181.3377756,0.0163333)(182.2947004,0.0162487)(183.3326891,0.0161578)(184.301441,0.016074)(185.3512567,0.0159841)(186.3819359,0.0158968)(187.3433783,0.0158162)(188.3858845,0.0157297)(189.3591539,0.0156499)(190.3132868,0.0155724)(191.3484834,0.0154891)(192.3144434,0.0154123)(193.3614671,0.0153298)(194.3893543,0.0152497)(195.3480047,0.0151758)(196.3877189,0.0150964)(197.3581963,0.015023)(198.4097376,0.0149443)(199.4421423,0.0148678)(200.4053102,0.0147972)(201.449542,0.0147214)(202.4245369,0.0146513)(203.3803953,0.0145832)(204.4173176,0.0145101)(205.385003,0.0144425)(206.4337523,0.0143699)(207.463365,0.0142994)(208.4237409,0.0142343)(209.4651807,0.0141643)(210.4373837,0.0140995)(211.3904501,0.0140367)(212.4245803,0.0139691)(213.3894738,0.0139066)(214.4354311,0.0138395)(215.4121516,0.0137774)(216.3697356,0.0137171)(217.4083833,0.0136523)(218.3777943,0.0135923)(219.4282692,0.0135279)(220.4596074,0.0134653)(221.4217089,0.0134074)(222.4648742,0.0133452)(223.4388027,0.0132876)(224.3935947,0.0132317)(225.4294505,0.0131715)(226.3960695,0.0131158)(227.4437523,0.013056)(228.4722986,0.0129978)(229.4316081,0.012944)(230.4719814,0.0128862)(231.4431179,0.0128327)(232.3951179,0.0127806)(233.4281817,0.0127246)(234.3920087,0.0126728)(235.4368996,0.0126171)(236.4125536,0.0125656)(237.3690711,0.0125154)(238.4066525,0.0124615)(239.374997,0.0124116)(240.4244054,0.0123579)(241.4546772,0.0123057)(242.4157122,0.0122574)(243.4578111,0.0122054)(244.4306732,0.0121573)(245.3843987,0.0121105)(246.419188,0.0120601)(247.3847406,0.0120135)(248.431357,0.0119633)(249.4588368,0.0119145)(250.4170799,0.0118693)(251.4563867,0.0118207)(252.4264568,0.0117757)(253.4775908,0.0117273)(254.5095881,0.0116802)(255.4723487,0.0116366)(256.5161731,0.0115897)(257.4907607,0.0115462)(258.4462118,0.0115039)(259.4827267,0.0114584)(260.4500048,0.0114162)(261.4983467,0.0113708)(262.5275521,0.0113266)(263.4875207,0.0112857)(264.5285531,0.0112417)(265.5003487,0.0112009)(266.4530078,0.0111612)(267.4867307,0.0111185)(268.4512168,0.0110789)(269.4967668,0.0110362)(270.4730799,0.0109967)(271.4302565,0.0109583)(272.468497,0.0109169)(273.4375006,0.0108785)(274.4875681,0.0108373)(275.518499,0.0107971)(276.4801931,0.0107598)(277.5229511,0.0107197)(278.4964723,0.0106826)(279.4508569,0.0106464)(280.4863053,0.0106074)(281.452517,0.0105713)(282.4997925,0.0105324)(283.5279314,0.0104945)(284.4868335,0.0104595)(285.5267995,0.0104217)(286.4975287,0.0103866)(287.5493217,0.010349)(288.5819782,0.0103122)(289.5453979,0.0102782)(290.5898814,0.0102415)(291.5651281,0.0102076)(292.5212382,0.0101745)(293.5584122,0.0101388)(294.5263494,0.0101057)(295.5753505,0.0100701)(296.6052149,0.0100355)(297.5658426,0.0100033)(298.6075341,0.0099687)(299.5799889,0.0099366)(300.533307,0.0099053)(301.567689,0.0098716)(302.5328343,0.0098403)(303.5790433,0.0098067)(304.6061158,0.0097739)(305.5639515,0.0097435)(306.602851,0.0097107)(307.5725138,0.0096803)(308.6232403,0.0096476)(309.6548303,0.0096157)(310.6171836,0.0095861)(311.6606006,0.0095543)(312.6347809,0.0095247)(313.5898246,0.0094959)(314.6259322,0.0094649)(315.5928029,0.0094361)(316.6407375,0.0094051)(317.6695355,0.0093749)(318.6290968,0.0093468)(319.6697219,0.0093166)(320.6411102,0.0092886)(321.5933619,0.0092613)(322.6266774,0.0092319)(323.5907562,0.0092046)(324.6358988,0.0091751)(325.6118046,0.0091478)(326.5685739,0.0091212)(327.606407,0.0090925)(328.5750033,0.0090659)(329.6246634,0.0090373)(330.655187,0.0090093)(331.6164738,0.0089833)(332.6588244,0.0089554)(333.6319382,0.0089295)(334.5859155,0.0089042)(335.6209566,0.0088769)(336.5867609,0.0088516)(337.633629,0.0088244)(338.6613606,0.0087978)(339.6198554,0.0087731)(340.659414,0.0087465)(341.6297359,0.0087218)(342.6811216,0.0086953)(343.7133707,0.0086693)(344.676383,0.0086453)(345.7204592,0.0086193)(346.6952985,0.0085953)(347.6510013,0.0085718)(348.687768,0.0085465)(349.6552978,0.008523)(350.7038915,0.0084977)(351.7333487,0.008473)(352.693569,0.00845)(353.7348532,0.0084253)(354.7069006,0.0084024)(355.6598114,0.00838)(356.693786,0.0083559)(357.6585239,0.0083335)(358.7043256,0.0083094)(359.6808905,0.0082869)(360.6383189,0.0082651)(361.6768111,0.0082415)(362.6460665,0.0082196)(363.6963857,0.008196)(364.7275684,0.008173)(365.6895143,0.0081516)(366.732524,0.0081286)(367.7062969,0.0081072)(368.6609333,0.0080864)(369.6966335,0.0080638)(370.6630969,0.0080429)(371.7106241,0.0080204)(372.7390148,0.0079984)(373.6981687,0.007978)(374.7383864,0.007956)(375.7093674,0.0079356)(376.7614121,0.0079136)(377.7943203,0.0078921)(378.7579918,0.0078721)(379.802727,0.0078506)(380.7782255,0.0078306)(381.7345874,0.0078111)(382.7720132,0.0077901)(383.7402021,0.0077705)(384.7894549,0.0077495)(385.8195711,0.0077289)(386.7804506,0.0077098)(387.8223938,0.0076892)(388.7951003,0.0076701)(389.7486703,0.0076514)(390.783304,0.0076313)(391.748701,0.0076126)(392.7951618,0.0075924)(393.822486,0.0075728)(394.7805735,0.0075545)(395.8197247,0.0075348)(396.7896392,0.0075165)(397.8406176,0.0074967)(398.8724593,0.0074774)(399.8350643,0.0074595)(400.8787331,0.0074402)(401.8531652,0.0074223)(402.8084606,0.0074048)(403.8448199,0.0073859)(404.8119424,0.0073684)(405.8601288,0.0073494)(406.8891786,0.007331)(407.8489916,0.0073138)(408.8898684,0.0072953)(409.8615084,0.0072781)(410.8140119,0.0072613)(411.8475792,0.0072432)(412.8119097,0.0072264)(413.8573041,0.0072082)(414.8334617,0.0071913)(415.7904827,0.0071749)(416.8285675,0.0071571)(417.7974156,0.0071406)(418.8473275,0.0071228)(419.8781028,0.0071054)(420.8396413,0.0070893)(421.8822437,0.0070719)(422.8556093,0.0070557)(423.8098383,0.0070399)(424.8451311,0.0070228)(425.8111872,0.007007)(426.8583071,0.0069899)(427.8862904,0.0069732)(428.845037,0.0069577)(429.8848474,0.0069409)(430.855421,0.0069254)(431.9070584,0.0069086)(432.9395593,0.0068922)(433.9028233,0.006877)(434.9471513,0.0068606)(435.9222424,0.0068453)(436.878197,0.0068304)(437.9152154,0.0068143)(438.882997,0.0067994)(439.9318424,0.0067832)(440.9615513,0.0067675)(441.9220234,0.0067529)(442.9635593,0.0067371)(443.9358584,0.0067224)(444.889021,0.0067081)(445.9232474,0.0066926)(446.8882371,0.0066782)(447.9342905,0.0066627)(448.9111072,0.0066483)(449.8687873,0.0066342)(450.9075312,0.006619)(451.8770384,0.0066049)(452.9276094,0.0065896)(453.9590438,0.0065747)(454.9212414,0.0065609)(455.9645029,0.0065459)(456.9385276,0.0065321)(457.8934157,0.0065185)(458.9293677,0.0065039)(459.8960828,0.0064903)(460.9438618,0.0064756)(461.9725042,0.0064612)(462.9319099,0.0064479)(463.9723794,0.0064335)(464.9436121,0.0064202)(465.8957082,0.0064071)(466.9288682,0.006393)(467.8927914,0.0063799)(468.9377784,0.0063657)(469.9135286,0.0063526)(470.8701423,0.0063398)(471.9078197,0.0063259)(472.8762605,0.006313)(473.925765,0.0062991)(474.956133,0.0062855)(475.9172642,0.0062728)(476.9594592,0.0062592)(477.9324174,0.0062465)(478.8862391,0.0062341)(479.9211246,0.0062208)(480.8867733,0.0062083)(481.9334859,0.0061949)(482.9610619,0.0061818)(483.9194011,0.0061696)(484.9588041,0.0061565)(485.9289704,0.0061442)(486.9802005,0.006131)(488.012294,0.0061181)(488.9751507,0.0061061)(490.0190713,0.0060932)(490.9937551,0.0060811)(491.9493023,0.0060694)(492.9859133,0.0060567)(493.9532876,0.0060449)(495.0017257,0.0060321)(496.0310272,0.0060197)(496.991092,0.0060081)(498.0322205,0.0059956)(499.0041124,0.005984)(499.0196731,0.0059838)(499.0352338,0.0059836)(499.0663553,0.0059832)(499.1285982,0.0059825)(499.253084,0.005981)(499.5020557,0.005978)(499.5176164,0.0059779)(499.5331771,0.0059777)(499.5642986,0.0059773)(499.6265415,0.0059766)(499.7510273,0.0059751)(499.7665881,0.0059749)(499.7821488,0.0059747)(499.8132703,0.0059743)(499.8755132,0.0059736)(499.8910739,0.0059734)(499.9066346,0.0059732)(499.9377561,0.0059729)(499.9533168,0.0059727)(499.9688775,0.0059725)(499.9844383,0.0059723)(499.999999,0.0059721)
		};
		\addplot[smooth,color=ForestGreen,line width=0.3ex] coordinates { (0,0.475426) (200,0.475426)};
		\addplot[smooth,color=ForestGreen,dashed, line width=0.3ex] coordinates { (0,0.185617) (200,0.185617)};
		\addplot[smooth,color=RoyalBlue,line width=0.3ex] coordinates { (0, 0.392121) (200, 0.392121) };
		\addplot[smooth,color=RoyalBlue,dashed, line width=0.3ex] coordinates { (0, 0.149901) (200, 0.149901) };
		\addplot[smooth,color=Purple,  line width=0.3ex] coordinates { (0, 0.3133) (200, 0.3133) };
		\addplot[smooth,color=Purple, dashed, line width=0.3ex] coordinates { (0, 0.117116) (200, 0.117116) };
		\addplot[smooth,color=RubineRed, line width=0.3ex] coordinates { (0, 0.25383) (200, 0.25383) };
		\addplot[smooth,color=RubineRed, dashed, line width=0.3ex] coordinates { (0, 0.0930847) (200, 0.0930847) };
		\addplot[smooth,color=Melon,line width=0.3ex] coordinates { (0, 0.207851) (200, 0.207851)};
		\addplot[smooth,color=Melon,dashed,line width=0.3ex] coordinates { (0,0.0749582) (200,0.0749582)};
		\end{axis}
		\end{tikzpicture}
		\quad 
		\begin{tikzpicture}
		\begin{axis}[
		scale=0.7,
		xmin=0, xmax=170,
		ymin=-0.01, ymax=0.3,
		xmajorgrids,
		ymajorgrids,
		xlabel=$k$,
		align =center, title={Model 3},
		legend columns=2, 
		legend style={/tikz/column 2/.style={
				column sep=5pt,
			},row sep=0.1cm, at={(1,1)}, 
			fill=white,draw=black,nodes=right}]
		\addplot[smooth,color=black,solid,line width=0.2ex]coordinates{
(1.e-6,0.9999997)(0.0150593,0.9962461)(0.0301176,0.9925142)(0.0602342,0.985115)(0.1204674,0.970571)(0.2409337,0.942468)(0.4818664,0.889957)(0.4969247,0.8868295)(0.511983,0.8837194)(0.5420996,0.877551)(0.6023328,0.8654182)(0.7227991,0.8419433)(0.9637319,0.7979638)(0.9800568,0.7951193)(0.9963817,0.7922913)(1.0290315,0.7866843)(1.0943312,0.7756631)(1.2249305,0.7543668)(1.4861292,0.7145679)(1.5024541,0.712197)(1.518779,0.7098393)(1.5514289,0.7051629)(1.6167285,0.6959638)(1.7473279,0.6781609)(2.0085266,0.6447885)(2.0237696,0.6429278)(2.0390127,0.6410762)(2.0694989,0.6374005)(2.1304713,0.6301569)(2.252416,0.6160884)(2.4963055,0.5895316)(2.5115486,0.5879382)(2.5267917,0.5863524)(2.5572779,0.583203)(2.6182502,0.5769927)(2.740195,0.5649155)(2.9840845,0.5420593)(2.9990285,0.5407124)(3.0139726,0.5393713)(3.0438608,0.5367071)(3.1036371,0.5314486)(3.2231898,0.5212042)(3.4622951,0.5017491)(3.4772392,0.5005768)(3.4921833,0.4994095)(3.5220715,0.4970895)(3.5818478,0.4925079)(3.7014005,0.4835719)(3.9405058,0.4665625)(3.9567165,0.4654491)(3.9729272,0.4643405)(4.0053486,0.462138)(4.0701914,0.4577901)(4.1998771,0.4493175)(4.4592484,0.4332173)(4.4754591,0.4322467)(4.4916698,0.4312801)(4.5240912,0.429359)(4.588934,0.4255645)(4.7186197,0.4181608)(4.977991,0.4040573)(4.9931199,0.4032623)(5.0082487,0.4024703)(5.0385065,0.400895)(5.099022,0.3977791)(5.2200531,0.391683)(5.4621152,0.3800082)(5.9462394,0.3585464)(5.9626349,0.3578602)(5.9790304,0.3571766)(6.0118214,0.3558167)(6.0774034,0.3531267)(6.2085674,0.3478625)(6.4708955,0.337776)(6.9955516,0.3192136)(7.0116481,0.3186754)(7.0277446,0.318139)(7.0599376,0.3170714)(7.1243235,0.3149571)(7.2530955,0.3108098)(7.5106394,0.3028272)(8.0257272,0.288009)(8.0407419,0.2875984)(8.0557566,0.2871889)(8.0857859,0.2863733)(8.1458446,0.2847558)(8.2659619,0.2815741)(8.5061967,0.275416)(8.9866661,0.2638634)(10.0286688,0.2418299)(11.0014347,0.2243178)(11.9550641,0.2094361)(12.9897573,0.1953642)(13.9552137,0.1838336)(15.0017339,0.1727761)(16.0291175,0.16314)(16.9872644,0.1550724)(18.0264751,0.1471772)(18.9964491,0.1404995)(20.0474868,0.1339149)(21.079388,0.1280237)(22.0420525,0.1229762)(23.0857807,0.1179346)(24.0602722,0.1135865)(25.0156271,0.109624)(26.0520458,0.1056263)(27.0192277,0.1021499)(28.0674735,0.0986314)(29.0965827,0.0954051)(30.0564551,0.0925804)(31.0973914,0.0897003)(32.0690909,0.0871688)(33.0216538,0.084822)(34.0552805,0.0824144)(35.0196705,0.0802881)(36.0651243,0.0781035)(37.0413413,0.0761683)(37.9984217,0.0743619)(39.036566,0.0724969)(40.0054735,0.0708387)(41.0554448,0.0691253)(42.0862795,0.0675219)(43.0478775,0.0660918)(44.0905393,0.0646081)(45.0639643,0.0632818)(46.0182528,0.0620333)(47.053605,0.0607334)(48.0197206,0.0595685)(49.0668999,0.0583553)(50.0949426,0.0572115)(51.0537486,0.0561843)(52.0936184,0.0551112)(53.0642515,0.0541459)(54.1159483,0.0531374)(55.1485086,0.0521832)(56.1118321,0.0513233)(57.1562195,0.0504225)(58.13137,0.0496095)(59.087384,0.0488375)(60.1244619,0.0480268)(61.0923029,0.0472941)(62.1412078,0.0465249)(63.1709761,0.0457936)(64.1315076,0.045132)(65.173103,0.0444358)(66.1454615,0.0438049)(67.0986835,0.0432037)(68.1329694,0.0425696)(69.0980184,0.0419946)(70.1441313,0.0413886)(71.1210074,0.0408383)(72.078747,0.0403127)(73.1175503,0.0397578)(74.0871169,0.0392535)(75.1377473,0.0387212)(76.1692412,0.0382125)(77.1314983,0.0377498)(78.1748192,0.0372607)(79.1489033,0.0368153)(80.1038508,0.0363889)(81.1398622,0.0359373)(82.1066368,0.0355259)(83.1544752,0.0350905)(84.1831771,0.0346734)(85.1426422,0.0342931)(86.1831711,0.03389)(87.1544632,0.0335222)(88.1066188,0.0331693)(89.1398382,0.0327947)(90.1038208,0.0324528)(91.1488672,0.03209)(92.1246769,0.0317585)(93.08135,0.0314401)(94.1190869,0.0311019)(95.0875871,0.0307927)(96.137151,0.0304646)(97.1675784,0.0301491)(98.1287691,0.0298607)(99.1710235,0.0295541)(100.1440412,0.0292735)(101.0979223,0.0290035)(102.1328672,0.0287162)(103.0985754,0.0284532)(104.1453474,0.0281735)(105.1729828,0.0279042)(106.1313814,0.0276577)(107.1708439,0.0273952)(108.1410696,0.0271546)(109.1923591,0.0268987)(110.224512,0.026652)(111.1874282,0.026426)(112.2314082,0.0261852)(113.2061514,0.0259643)(114.1617581,0.0257513)(115.1984285,0.0255242)(116.1658622,0.0253158)(117.2143598,0.0250938)(118.2437207,0.0248796)(119.2038449,0.0246831)(120.2450329,0.0244734)(121.2169841,0.0242809)(122.1697988,0.0240951)(123.2036773,0.0238967)(124.168319,0.0237144)(125.2140245,0.02352)(126.1904933,0.0233413)(127.1478255,0.0231688)(128.1862215,0.0229845)(129.1553808,0.0228151)(130.2056038,0.0226343)(131.2366903,0.0224596)(132.1985401,0.022299)(133.2414536,0.0221275)(134.2151304,0.0219697)(135.1696706,0.0218172)(136.2052746,0.0216541)(137.1716419,0.0215041)(138.219073,0.0213439)(139.2473675,0.0211889)(140.2064252,0.0210463)(141.2465468,0.0208938)(142.2174316,0.0207535)(143.2693802,0.0206036)(144.3021922,0.0204585)(145.2657675,0.0203249)(146.3104066,0.0201821)(147.2858089,0.0200505)(148.2420747,0.0199232)(149.2794042,0.0197869)(150.247497,0.0196613)(151.2966537,0.0195271)(152.3266737,0.019397)(153.287457,0.0192773)(154.3293041,0.0191491)(155.3019144,0.0190309)(156.2553882,0.0189165)(157.2899258,0.0187939)(158.2552266,0.018681)(159.3015912,0.01856)(160.3288193,0.0184428)(161.2868106,0.0183348)(162.3258657,0.0182191)(163.2956841,0.0181125)(164.3465662,0.0179983)(165.3783118,0.0178876)(166.3408207,0.0177855)(167.3843933,0.0176761)(168.3587292,0.0175752)(169.3139285,0.0174774)(170.3501916,0.0173726)(171.317218,0.0172758)(172.3653082,0.0171722)(173.3942618,0.0170716)(174.3539786,0.0169789)(175.3947593,0.0168795)(176.3663031,0.0167877)(177.3187105,0.0166987)(178.3521816,0.0166032)(179.316416,0.0165151)(180.3617142,0.0164206)(181.3377756,0.0163333)(182.2947004,0.0162487)(183.3326891,0.0161578)(184.301441,0.016074)(185.3512567,0.0159841)(186.3819359,0.0158968)(187.3433783,0.0158162)(188.3858845,0.0157297)(189.3591539,0.0156499)(190.3132868,0.0155724)(191.3484834,0.0154891)(192.3144434,0.0154123)(193.3614671,0.0153298)(194.3893543,0.0152497)(195.3480047,0.0151758)(196.3877189,0.0150964)(197.3581963,0.015023)(198.4097376,0.0149443)(199.4421423,0.0148678)(200.4053102,0.0147972)(201.449542,0.0147214)(202.4245369,0.0146513)(203.3803953,0.0145832)(204.4173176,0.0145101)(205.385003,0.0144425)(206.4337523,0.0143699)(207.463365,0.0142994)(208.4237409,0.0142343)(209.4651807,0.0141643)(210.4373837,0.0140995)(211.3904501,0.0140367)(212.4245803,0.0139691)(213.3894738,0.0139066)(214.4354311,0.0138395)(215.4121516,0.0137774)(216.3697356,0.0137171)(217.4083833,0.0136523)(218.3777943,0.0135923)(219.4282692,0.0135279)(220.4596074,0.0134653)(221.4217089,0.0134074)(222.4648742,0.0133452)(223.4388027,0.0132876)(224.3935947,0.0132317)(225.4294505,0.0131715)(226.3960695,0.0131158)(227.4437523,0.013056)(228.4722986,0.0129978)(229.4316081,0.012944)(230.4719814,0.0128862)(231.4431179,0.0128327)(232.3951179,0.0127806)(233.4281817,0.0127246)(234.3920087,0.0126728)(235.4368996,0.0126171)(236.4125536,0.0125656)(237.3690711,0.0125154)(238.4066525,0.0124615)(239.374997,0.0124116)(240.4244054,0.0123579)(241.4546772,0.0123057)(242.4157122,0.0122574)(243.4578111,0.0122054)(244.4306732,0.0121573)(245.3843987,0.0121105)(246.419188,0.0120601)(247.3847406,0.0120135)(248.431357,0.0119633)(249.4588368,0.0119145)(250.4170799,0.0118693)(251.4563867,0.0118207)(252.4264568,0.0117757)(253.4775908,0.0117273)(254.5095881,0.0116802)(255.4723487,0.0116366)(256.5161731,0.0115897)(257.4907607,0.0115462)(258.4462118,0.0115039)(259.4827267,0.0114584)(260.4500048,0.0114162)(261.4983467,0.0113708)(262.5275521,0.0113266)(263.4875207,0.0112857)(264.5285531,0.0112417)(265.5003487,0.0112009)(266.4530078,0.0111612)(267.4867307,0.0111185)(268.4512168,0.0110789)(269.4967668,0.0110362)(270.4730799,0.0109967)(271.4302565,0.0109583)(272.468497,0.0109169)(273.4375006,0.0108785)(274.4875681,0.0108373)(275.518499,0.0107971)(276.4801931,0.0107598)(277.5229511,0.0107197)(278.4964723,0.0106826)(279.4508569,0.0106464)(280.4863053,0.0106074)(281.452517,0.0105713)(282.4997925,0.0105324)(283.5279314,0.0104945)(284.4868335,0.0104595)(285.5267995,0.0104217)(286.4975287,0.0103866)(287.5493217,0.010349)(288.5819782,0.0103122)(289.5453979,0.0102782)(290.5898814,0.0102415)(291.5651281,0.0102076)(292.5212382,0.0101745)(293.5584122,0.0101388)(294.5263494,0.0101057)(295.5753505,0.0100701)(296.6052149,0.0100355)(297.5658426,0.0100033)(298.6075341,0.0099687)(299.5799889,0.0099366)(300.533307,0.0099053)(301.567689,0.0098716)(302.5328343,0.0098403)(303.5790433,0.0098067)(304.6061158,0.0097739)(305.5639515,0.0097435)(306.602851,0.0097107)(307.5725138,0.0096803)(308.6232403,0.0096476)(309.6548303,0.0096157)(310.6171836,0.0095861)(311.6606006,0.0095543)(312.6347809,0.0095247)(313.5898246,0.0094959)(314.6259322,0.0094649)(315.5928029,0.0094361)(316.6407375,0.0094051)(317.6695355,0.0093749)(318.6290968,0.0093468)(319.6697219,0.0093166)(320.6411102,0.0092886)(321.5933619,0.0092613)(322.6266774,0.0092319)(323.5907562,0.0092046)(324.6358988,0.0091751)(325.6118046,0.0091478)(326.5685739,0.0091212)(327.606407,0.0090925)(328.5750033,0.0090659)(329.6246634,0.0090373)(330.655187,0.0090093)(331.6164738,0.0089833)(332.6588244,0.0089554)(333.6319382,0.0089295)(334.5859155,0.0089042)(335.6209566,0.0088769)(336.5867609,0.0088516)(337.633629,0.0088244)(338.6613606,0.0087978)(339.6198554,0.0087731)(340.659414,0.0087465)(341.6297359,0.0087218)(342.6811216,0.0086953)(343.7133707,0.0086693)(344.676383,0.0086453)(345.7204592,0.0086193)(346.6952985,0.0085953)(347.6510013,0.0085718)(348.687768,0.0085465)(349.6552978,0.008523)(350.7038915,0.0084977)(351.7333487,0.008473)(352.693569,0.00845)(353.7348532,0.0084253)(354.7069006,0.0084024)(355.6598114,0.00838)(356.693786,0.0083559)(357.6585239,0.0083335)(358.7043256,0.0083094)(359.6808905,0.0082869)(360.6383189,0.0082651)(361.6768111,0.0082415)(362.6460665,0.0082196)(363.6963857,0.008196)(364.7275684,0.008173)(365.6895143,0.0081516)(366.732524,0.0081286)(367.7062969,0.0081072)(368.6609333,0.0080864)(369.6966335,0.0080638)(370.6630969,0.0080429)(371.7106241,0.0080204)(372.7390148,0.0079984)(373.6981687,0.007978)(374.7383864,0.007956)(375.7093674,0.0079356)(376.7614121,0.0079136)(377.7943203,0.0078921)(378.7579918,0.0078721)(379.802727,0.0078506)(380.7782255,0.0078306)(381.7345874,0.0078111)(382.7720132,0.0077901)(383.7402021,0.0077705)(384.7894549,0.0077495)(385.8195711,0.0077289)(386.7804506,0.0077098)(387.8223938,0.0076892)(388.7951003,0.0076701)(389.7486703,0.0076514)(390.783304,0.0076313)(391.748701,0.0076126)(392.7951618,0.0075924)(393.822486,0.0075728)(394.7805735,0.0075545)(395.8197247,0.0075348)(396.7896392,0.0075165)(397.8406176,0.0074967)(398.8724593,0.0074774)(399.8350643,0.0074595)(400.8787331,0.0074402)(401.8531652,0.0074223)(402.8084606,0.0074048)(403.8448199,0.0073859)(404.8119424,0.0073684)(405.8601288,0.0073494)(406.8891786,0.007331)(407.8489916,0.0073138)(408.8898684,0.0072953)(409.8615084,0.0072781)(410.8140119,0.0072613)(411.8475792,0.0072432)(412.8119097,0.0072264)(413.8573041,0.0072082)(414.8334617,0.0071913)(415.7904827,0.0071749)(416.8285675,0.0071571)(417.7974156,0.0071406)(418.8473275,0.0071228)(419.8781028,0.0071054)(420.8396413,0.0070893)(421.8822437,0.0070719)(422.8556093,0.0070557)(423.8098383,0.0070399)(424.8451311,0.0070228)(425.8111872,0.007007)(426.8583071,0.0069899)(427.8862904,0.0069732)(428.845037,0.0069577)(429.8848474,0.0069409)(430.855421,0.0069254)(431.9070584,0.0069086)(432.9395593,0.0068922)(433.9028233,0.006877)(434.9471513,0.0068606)(435.9222424,0.0068453)(436.878197,0.0068304)(437.9152154,0.0068143)(438.882997,0.0067994)(439.9318424,0.0067832)(440.9615513,0.0067675)(441.9220234,0.0067529)(442.9635593,0.0067371)(443.9358584,0.0067224)(444.889021,0.0067081)(445.9232474,0.0066926)(446.8882371,0.0066782)(447.9342905,0.0066627)(448.9111072,0.0066483)(449.8687873,0.0066342)(450.9075312,0.006619)(451.8770384,0.0066049)(452.9276094,0.0065896)(453.9590438,0.0065747)(454.9212414,0.0065609)(455.9645029,0.0065459)(456.9385276,0.0065321)(457.8934157,0.0065185)(458.9293677,0.0065039)(459.8960828,0.0064903)(460.9438618,0.0064756)(461.9725042,0.0064612)(462.9319099,0.0064479)(463.9723794,0.0064335)(464.9436121,0.0064202)(465.8957082,0.0064071)(466.9288682,0.006393)(467.8927914,0.0063799)(468.9377784,0.0063657)(469.9135286,0.0063526)(470.8701423,0.0063398)(471.9078197,0.0063259)(472.8762605,0.006313)(473.925765,0.0062991)(474.956133,0.0062855)(475.9172642,0.0062728)(476.9594592,0.0062592)(477.9324174,0.0062465)(478.8862391,0.0062341)(479.9211246,0.0062208)(480.8867733,0.0062083)(481.9334859,0.0061949)(482.9610619,0.0061818)(483.9194011,0.0061696)(484.9588041,0.0061565)(485.9289704,0.0061442)(486.9802005,0.006131)(488.012294,0.0061181)(488.9751507,0.0061061)(490.0190713,0.0060932)(490.9937551,0.0060811)(491.9493023,0.0060694)(492.9859133,0.0060567)(493.9532876,0.0060449)(495.0017257,0.0060321)(496.0310272,0.0060197)(496.991092,0.0060081)(498.0322205,0.0059956)(499.0041124,0.005984)(499.0196731,0.0059838)(499.0352338,0.0059836)(499.0663553,0.0059832)(499.1285982,0.0059825)(499.253084,0.005981)(499.5020557,0.005978)(499.5176164,0.0059779)(499.5331771,0.0059777)(499.5642986,0.0059773)(499.6265415,0.0059766)(499.7510273,0.0059751)(499.7665881,0.0059749)(499.7821488,0.0059747)(499.8132703,0.0059743)(499.8755132,0.0059736)(499.8910739,0.0059734)(499.9066346,0.0059732)(499.9377561,0.0059729)(499.9533168,0.0059727)(499.9688775,0.0059725)(499.9844383,0.0059723)(499.999999,0.0059721)
		};
		\addplot[smooth,color=ForestGreen,line width=0.3ex] coordinates { (0,0.226394) (200,0.226394)};
		\addplot[smooth,color=ForestGreen,dashed, line width=0.3ex] coordinates { (0,0.0895676) (200,0.0895676)};
		\addplot[smooth,color=RoyalBlue,line width=0.3ex] coordinates { (0, 0.156848) (200, 0.156848) };
		\addplot[smooth,color=RoyalBlue,dashed, line width=0.3ex] coordinates { (0, 0.0634203) (200, 0.0634203) };
		\addplot[smooth,color=Purple,  line width=0.3ex] coordinates { (0, 0.104433) (200, 0.104433) };
		\addplot[smooth,color=Purple, dashed, line width=0.3ex] coordinates { (0, 0.0428829) (200, 0.0428829) };
		\addplot[smooth,color=RubineRed, line width=0.3ex] coordinates { (0, 0.072523) (200, 0.072523) };
		\addplot[smooth,color=RubineRed, dashed, line width=0.3ex] coordinates { (0, 0.0300118) (200, 0.0300118) };
		\addplot[smooth,color=Melon,line width=0.3ex] coordinates { (0,0.0519627) (200,0.0519627)};
		\addplot[smooth,color=Melon,dashed,line width=0.3ex] coordinates { (0,0.0215725) (200,0.0215725)};
		\end{axis}
		\end{tikzpicture}
	\end{center}
	\caption{ Study of the constant $C^*_{\gamma,\alpha}$ defined in \eqref{threshold}, for the physical values of   $\alpha=0$ (\ref{a=0}) and $\alpha=0.5$ (\ref{a=0.5}), while varying $\gamma \in (0,2]$.}
	\label{Fig}
\end{figure}

\subsection{The upper bound for the transition potential function}

The second important estimate  enable us to control the positive contributions of the moments of the collisional polyatomic operator are simply derived from obtaining  local upper bound  estimates that  establish the control of the potentials \eqref{tf} by  their associated  Lebesgue weights \eqref{brackets}. This is a much simpler argument that the one needed to obtain the coercive  lower bounds calculated in the previous Section~\ref{Coerciveness}.  \\
\begin{lemma}[Upper bound]\label{lemma upper bound} For any $\gamma \in (0,2]$ the following inequality holds
	\begin{equation}\label{up tf}
	\tf \leq 2^{\frac{3 \gamma}{2} - 1} \left(   \langle v, I \rangle^\gamma   + \langle v_*, I_* \rangle^\gamma  \right),
	\end{equation}
	for  $v, v_* \in \mathbb{R}^3$ and $I,I_* \in [0,\infty)$.
\end{lemma}
\
\begin{proof}
	We first write
	\begin{multline}\label{pom}
	|v-v_*|^\gamma + \left(\frac{I+I_*}{m}\right)^{\gamma/2} \leq 2^\gamma \left( \left( \frac{1}{4}  |v-v_*|^2  \right)^{\gamma/2} + \left(\frac{I+I_*}{m}\right)^{\gamma/2} \right)
	\\
	\leq 2^{\frac{3 \gamma}{2} -1 } \left( \frac{1}{4} \left|v-v_*\right|^2 + \frac{I + I_*}{m}  \right)^{\gamma/2}.
	\end{multline}
	Then, using  
	\begin{equation*}
	\frac{1}{4} \left|v-v_*\right|^2 + \frac{I}{m} + \frac{I_*}{m} \leq  \frac{1}{2} \left| v \right|^2  +\frac{1}{2} \left| v_* \right|^2  +   \frac{I}{m}  +  \frac{I_*}{m} \leq \langle v, I \rangle^2   + \langle v_*, I_* \rangle^2,
	\end{equation*}
	and since  $\gamma/2 \leq 1$, we can estimate 
	\begin{equation}\label{pomocna 4}
	\left( \frac{1}{4} \left|v-v_*\right|^2 + \frac{I}{m} + \frac{I_*}{m} \right)^{\gamma/2} \leq   \langle v, I \rangle^\gamma   + \langle v_*, I_* \rangle^\gamma.
	\end{equation}
	Combining this result with \eqref{pom}, we conclude the proof.
\end{proof}

\section{$L^1_k$ moments a priori estimates }\label{Sec: poly mom}


The Boltzmann collisional model for polyatomic gases becomes particularly challenging when we gather a priori estimates for the propagation and generation of  $k-$Lebesgue moment  generating the Banach space $L^1_k(\mathbb{R}^{3+})$ defined in \eqref{space L_k^1} and \eqref{norm}. 
More precisely, the biggest challenge is to find sufficient conditions for  transition  functions and their corresponding probability measures and bounds, as discussed in Section~\ref{Sec: Suff}, in order to be able to generate an Ordinary Differential Inequality (ODI) that describes the times evolution of  the  solution  norm of $\$f(\cdot, t\|_{L^1_k(\mathbb{R}^{3+})}.$   In fact,  the previous polyatomic compact manifold averaging Lemma \ref{lemma povzner} together with the requirement \eqref{k* const} are sufficient  conditions to show that the evolution of the $k$-th polynomial moment of the collision operator will become negative for $k>k_*$ and $\gamma>0$, due to the fact that the negative contribution of the collision operator moments are superlinear with respect to their positive contribution. 

  Thus, the following Lemma is a non-trivial 
adaptation form the one proven in \cite{Alonso-IG-BAMS} for the classical binary elastic Cauchy Theory of the Boltzmann equation in the  Banach $k-$Lebesgue moment weighted space $L^1_k(\mathbb{R}^3)$ to the Cauchy initial value problem Boltzmann  for polyatomic gases \eqref{Cauchy},  whose solution posed as a probability density in now, for the polyatomic case  the Banach $k-$Lebesgue moment space $L^1_k(\mathbb{R}^{3+})$  with transition probability functions  described in Section~\ref{Sec: Suff}.

\begin{lemma}[Moments bound for the collision operator]\label{lemma Q}
	Let $f \in L^1_1$ satisfying assumptions from Lower bound Lemma \ref{lemma lower bound} and condition \eqref{lower bound lemma}. Moreover, suppose that the transition function $\mathcal{B}$ satisfies Assumption \ref{Sec: Ass B}. Then for any  $\gamma \in (0,2]$, the following inequality holds
	\begin{equation}\label{poly coll op}
	\int_{ \mathbb{R}^{3+}} Q(f,f)(v,I) \, \langle v, I \rangle^{\tk} \mathrm{d} I \mathrm{d}v = \mathfrak{m}_k\left[  Q(f,f) \right]  \leq - A_{\ks}  \left\|f\right\|_{L^1_k}^{1+ \frac{\gamma/2}{k} }    +  B_k \left\|f\right\|_{L^1_k},
	\end{equation}
	for large enough $k$ such that
	\begin{equation}\label{k from moment bound}
	k> \ks, \qquad \ks=\max\left\{  \bks, 1+\gamma,  1+\pl/2  \right\} \quad \text{finite}
	\end{equation}
	where $\bks$, defined in \eqref{ck kappa lb}, is such that inequality \eqref{k* const} holds, and $\delta>0$ is from  condition \eqref{lower bound lemma}.

 For such $\ks$,  the  constant  $ A_{\ks} \geq A_k >0$ is independent of $k$,  and $B_k>0$,
	\begin{align}\label{moments const}
 A_{\ks} & := {    \frac{ \clb }{2} \left(\kappa^{lb}  - \mathcal{C}_{\ks} \right) }   \, \left\|f\right\|_{L^1_0}^{-\frac{\hg}{k_*}},  \\
	B_k 	&     = \mathcal{C}_k 2^{\frac{3 \gamma+k}{2}} 
	\max\left\{   \left( \frac{ \kappa^{lb} \,  2^{\frac{3 \gamma +k}{2}} }{A_{\ks}} \right)^{\frac{\theta_{k}}{1-\theta_k}}   \frac{\eta_{k}^{\frac{1}{1-\theta_k}}}{\left\|f\right\|_{L^1_1}},  1 \right\} \left\|f\right\|_{L^1_1},\nonumber
	\end{align} 
	{   with
		\begin{equation}\label{eta i theta k}
		\eta_{k} :=  \left\|f\right\|_{L^1_0}^\frac{1+\hg}{k+\hg}   \left\|f\right\|_{L^1_1}^\frac{k-1}{k-1+\hg}, \qquad \theta_k = \frac{\hg}{k-1+\hg} +  \frac{k-1}{k+\hg}<1,
		\end{equation}	
	}
	and  $ \clb$ is the lower bound  constant for the  convolution of $f$ and power law function of order $\gamma$ in terms of the $\gamma$-power of Lebesgue weight  from \eqref{lower bound lemma} that enables the super linear behavior for moments of order $k> \ks$, $\mathcal{C}_{\hk} $, monotone decreasing in $k$,  is the constant from Lemma \ref{lemma povzner}  and $\kappa^{lb}$ from \eqref{kappas}.
\end{lemma}

\begin{proof}
For the test function
	$$\chi(v,I)= \langle v, I \rangle^{\tk},$$
the	weak form  \eqref{weak form} yields
	\begin{multline}\label{mathcal W}
\mathcal{W} := \int_{ \mathbb{R}^{3+}} Q(f,f)(v,I) \, \langle v, I \rangle^{\tk} \mathrm{d} I \mathrm{d}v
\\	= \frac{1}{2} \int_{({\mathbb{R}^{3+}})^2\times \K} \frac{f f_*}{I^\alpha I_*^\alpha}
	\left( \langle v', I' \rangle^{\tk}+ \langle v'_*, I'_* \rangle^{\tk} - \langle v, I \rangle^{\tk} - \langle v_*, I_* \rangle^{\tk} \right)
	\mathrm{d}A.
	\end{multline}
	We now use the form of transition probability rate \eqref{trans prob rate ass}. Because of the integrability properties of all multiplying functions involved, we can separate the integral $\mathcal{W}$  into the gain $\mathcal{W}^+$ 
			\begin{multline}\label{gain}
	\mathcal{W}^+
	= \frac{1}{2} \int_{({\mathbb{R}^{3+}})^2 \times \K } f f_*
	\left( \langle v', I' \rangle^{\tk}+ \langle v'_*, I'_* \rangle^{\tk} 
	\right) 
\mathcal{B}(v,v_*,I,I_*,r,R,\sigma) \\ \times \varphi_\alpha(r) \,  \psi_\alpha(R)\, (1-R) R^{1/2}  \, \mathrm{d}R \,\mathrm{d}r \, \mathrm{d} \sigma\,   \mathrm{d}I_* \mathrm{d} v_* \mathrm{d}I \mathrm{d}v,
	\end{multline}
and loss part $\mathcal{W}^-$,
	\begin{multline}\label{loss}
	\mathcal{W}^-
	= \frac{1}{2} \int_{({\mathbb{R}^{3+}})^2 \times \K } f f_*
	\left( \langle v, I \rangle^{\tk}+ \langle v_*, I_* \rangle^{\tk} 
	\right) 
\mathcal{B}(v,v_*,I,I_*,r,R,\sigma) \\ \times \varphi_\alpha(r) \,  \psi_\alpha(R) \, (1-R) R^{1/2} \, \mathrm{d}R \, \mathrm{d}r \,  \mathrm{d} \sigma\,  \mathrm{d}I_* \mathrm{d} v_* \mathrm{d}I \mathrm{d}v,
	\end{multline}
	so that 
	\begin{equation}\label{W}
	\mathcal{W}=\mathcal{W}^+-\mathcal{W}^-.
	\end{equation}
	We treat each term separately. For the gain part  weak form, we use the bound from above stated in \eqref{trans prob rate ass},
\begin{multline*}
\mathcal{W}^+
\leq \frac{1}{2}  \int_{({\mathbb{R}^{3+}})^2  } f f_* \tB(v,v_*,I,I_*)  \int_{\K}
\left( \langle v', I' \rangle^{\tk}+ \langle v'_*, I'_* \rangle^{\tk} 
\right) 
b(\hat{u} \cdot \sigma) \, \dgu\, \egu \\ \times  \varphi_\alpha(r)  \,   \psi_\alpha(R)\,  (1-R) R^{1/2} \, \mathrm{d}R \,  \mathrm{d}r \,  \mathrm{d} \sigma\,  \mathrm{d}I_* \mathrm{d} v_* \mathrm{d}I \mathrm{d}v.
\end{multline*}
The  Averaging Lemma \ref{lemma povzner}  estimates the primed quantities averaged over the compact set  $\K$, and for the gain term it yields
\begin{multline}\label{W+}
\mathcal{W}^+
\leq  \frac{\mathcal{C}_{\hk}  }{2}      
 \int_{({\mathbb{R}^{3+}})^2 } f f_*   \left(  \tf \right)  \\ \times \left(  \left( \langle v, I \rangle^2 + \langle v_*, I_* \rangle^2 \right)^{\hk} \right)
\mathrm{d}I_* \mathrm{d} v_* \mathrm{d}I \mathrm{d}v,
\end{multline}
Now the polynomial inequalities from Lemmas  \ref{binomial} and \ref{moment products} yield
\begin{multline}\label{polynomial 1}
\left(\left\langle v, I \right\rangle^2 + \left\langle v_*, I_* \right\rangle^2 \right)^{\hk}  
 \\ \leq \left\langle v, I \right\rangle^{\tk} + \left\langle v_*, I_* \right\rangle^{\tk} + \sum_{\ell=1}^{\ell_{k}} \left( \begin{matrix}
k \\ \ell 
\end{matrix} \right) \left( \left\langle v, I \right\rangle^{2\ell} \left\langle v_*, I_* \right\rangle^{2(k- \ell)}+ \left\langle v, I \right\rangle^{2(k- \ell)} \left\langle v_*, I_* \right\rangle^{2\ell} \right),\\
\leq \left\langle v, I \right\rangle^{\tk} + \left\langle v_*, I_* \right\rangle^{\tk} + \tilde{c}_k \left( \left\langle v, I \right\rangle^2 \left\langle v_*, I_* \right\rangle^{2(k- 1)}+ \left\langle v, I \right\rangle^{2(k- 1)} \left\langle v_*, I_* \right\rangle^2 \right),
\end{multline}
with 
\begin{equation}\label{C tilde}
  \sum_{\ell=1}^{\ell_{k}} \left( \begin{matrix}
k \\ \ell 
\end{matrix} \right)\le  2^{\frac{k+2}2} -1  =: \tilde{c}_k,  \qquad  \ \  \text{for}\   \ell_{k} = \lfloor\frac{k + 1}{2}\rfloor.
\end{equation}
Thus, the bound for $\mathcal{W}^+$ becomes
\begin{multline}\label{gain 2}
\mathcal{W}^+
\leq    \frac{\mathcal{C}_{\hk} }{2}     
 \int_{({\mathbb{R}^{3+}})^2 } f f_*  \left( \tf \right)
\left( \left\langle v, I \right\rangle^{\tk} + \left\langle v_*, I_* \right\rangle^{\tk} 
\right. \\ \left.
+ \tilde{c}_k \left( \left\langle v, I \right\rangle^2 \left\langle v_*, I_* \right\rangle^{2(k- 1)}+ \left\langle v, I \right\rangle^{2(k- 1)} \left\langle v_*, I_* \right\rangle^2 \right) \right) \mathrm{d}I_* \mathrm{d} v_* \mathrm{d}I \mathrm{d}v.
\end{multline}

Now we turn to the loss term  $\mathcal{W}^-$ defined in \eqref{loss}. We first use the bound from below of the transition function $\mathcal{B}$  as stated in \eqref{trans prob rate ass}, and obtain 
	\begin{multline}\label{loss 2}
\mathcal{W}^-
\geq 
\,   \frac{\kappa^{lb}}{2}   
\int_{({\mathbb{R}^{3+}})^2 }  f f_* \left(\tf\right)\\
\times \left( \langle v, I \rangle^{\tk}+ \langle v_*, I_* \rangle^{\tk} 
\right) \mathrm{d}I_* \mathrm{d} v_* \mathrm{d}I \mathrm{d}v,
\end{multline}
where the constant $\kappa^{lb}$ is defined in \eqref{kappas}.    Gathering the estimates for the gain term \eqref{gain 2} and for the loss term \eqref{loss 2}, the weak form $\mathcal{W}$ from \eqref{W} becomes
		\begin{multline}\label{W 2}
	\mathcal{W}
	\leq  \frac{1}{2} \int_{({\mathbb{R}^{3+}})^2 \times \K } f f_* \left(  \tf  \right) 
	\ \left\{  - \tilde{A}_{\ks}  \left( \langle v, I \rangle^{\tk}+ \langle v_*, I_* \rangle^{\tk} 
	\right)  \right. \\ \left. + \tilde{B}_k \left( \left\langle v, I \right\rangle^2 \left\langle v_*, I_* \right\rangle^{2(k- 1)} + \left\langle v, I \right\rangle^{2(k- 1)} \left\langle v_*, I_* \right\rangle^2 \right)  \right\} \mathrm{d}I_* \mathrm{d} v_* \mathrm{d}I \mathrm{d}v,
	\end{multline}
with  the uniform in $k$ constant $\tilde{A}_{\ks} $, for  $\ks$ chosen  in  \eqref{k from moment bound},  defined by 
\begin{equation}\label{Ak tilde}
 \tilde{A}_{\ks}  =   \kappa^{lb} - \mathcal{C}_{k_*}   >  \tilde{A}_{k}  >\, 0,
\end{equation}
i.e. strictly  positive, for large enough $k>\ks$, by virtue of \eqref{k* const}. On the other hand, for $\tilde{c}_k $ as defined in \eqref{C tilde}, the constant
  $\tilde{B}_k $ bounded for each fixed  $k$
\begin{equation}\label{Bk tilde}
\tilde{B}_k = \tilde{c}_k \, {  \mathcal{C}_{\hk}    } \geq 0.
\end{equation}

In order to reach the conclusion of this Lemma, we remind on the  moment notation 
\begin{equation*}
\mathfrak{m}_{k}(t) := \left\| f \right\|_{L^1_k}(t).
\end{equation*}
from the Definition \ref{def poly moment}.

Now  for \eqref{W 2} we make of use the upper bound \eqref{up tf} and the lower bound $c_{lb}$ from (\ref{lower bound lemma}, \ref{clb})  for the term
\begin{equation*}
\tf\, , 
\end{equation*}
combined  with the Coercive Lemma  estimate \eqref{coercive_estimate},    the right hand side  in  \eqref{W 2} is controlled by  
	\begin{equation}\label{mm pom}
\mathcal{W}
\leq  - \clb \, \tilde{A}_{\ks}  \, \mathfrak{m}_{k+\hg}  + 2^{\frac{3 \gamma}{2}-1} \tilde{B}_k\left( \mathfrak{m}_{1+{\hg}} \, \mathfrak{m}_{k-1}  +  \mathfrak{m}_{k-1+{\hg}} \, \mathfrak{m}_{1}   \right), \ \hk \geq \ks.
\end{equation}
For the second term we  use monotonicity of moments \eqref{monotonicity of norm} 

implying
	\begin{equation*}
	\mathfrak{m}_{k-1+{\hg}} \leq \mathfrak{m}_{k}, \quad \text{since} \  0< \gamma \leq 2.
	\end{equation*}
Then we invoke  arguments of  \cite{Alonso-IG-BAMS} that involve moment interpolation formulas 
\begin{equation}\label{mom interpolation}
 \mathfrak{m}_j \leq  \mathfrak{m}_{j_1}^\tau \,  \mathfrak{m}_{j_2}^{1-\tau}, \quad 0\leq j_1 \leq j \leq j_2, \quad 0< \tau < 1, \quad j=\tau j_1 + (1-\tau) j_2,
\end{equation}
and Young's inequality  estimate, for a given parameter $\epsilon$ to be chosen, 
\begin{equation}\label{Y}
\left| a \, b\right| \leq \frac{1}{p \epsilon^{p/q} } \left|a\right|^p + \frac{\epsilon}{q} \left|b\right|^q, \quad \text{for} \ \  \epsilon>0 \ \ \text{and} \ \ \frac{1}{p} + \frac{1}{q} =1.
\end{equation}
We first interpolate the moment $ \mathfrak{m}_{1+{\hg}}$ by applying \eqref{mom interpolation} for the following choice
$j=1+{\hg}, \quad j_1 = 1, \quad j_2=k+\hg, \quad \tau = \frac{k-1}{k-1+\hg},$
leading to the inequality
\begin{equation}\label{mom int 1}
 \mathfrak{m}_{1+{\hg}} \leq  \mathfrak{m}_{1}^\frac{k-1}{k-1+\hg} \mathfrak{m}_{k+\hg}^\frac{\hg}{k-1+\hg}.
\end{equation}
Next, we interpolate the moment  $ \mathfrak{m}_{k-1}$, by taken 
$j=k-1, \quad j_1 = 0, \quad j_2=k+\hg, \quad \tau = \frac{1+\hg}{k+\hg}, $
leading to
\begin{equation}\label{mom int 2}
\mathfrak{m}_{k-1} \leq  \mathfrak{m}_{0}^\frac{1+\hg}{k+\hg}  \mathfrak{m}_{k+\hg}^\frac{k-1}{k+\hg} \, .
\end{equation}
Therefore, \eqref{mom int 1} and \eqref{mom int 2} yield
\begin{equation}\label{mom int 3}
 \mathfrak{m}_{1+{\hg}} \mathfrak{m}_{k-1} \leq \eta_{k}  \, \mathfrak{m}_{k+\hg}^{\theta_{k}},   
\end{equation}
with the factor $\eta_k$ depending solely on the zeroth and first moment of the initial data $f_0(v)$,  
\begin{equation}\label{eta k}
 \eta_{k} := \eta_{k}( \mathfrak{m}_0,  \mathfrak{m}_1) =  \mathfrak{m}_{0}^\frac{1+\hg}{k+\hg}   \mathfrak{m}_{1}^\frac{k-1}{k-1+\hg},
\end{equation}
and
\begin{equation}\label{theta k}
 \theta_k = \frac{\hg}{k-1+\hg} +  \frac{k-1}{k+\hg} = 1 - \frac{k-1}{(k-1+\hg)(k+\hg)} < 1.
\end{equation}

Taking estimates  (\ref{mom int 3},\ref{eta k}, \ref{theta k})  and recalling  that $ \tilde{B}_k= \tilde{c}_k \, {  \mathcal{C}_{\hk}\le   2^{\frac{k+2}{2}}  \mathcal{C}_{\hk}\   } $ from \eqref{Bk tilde} and \eqref{C tilde},  then inequality \eqref{mm pom}  can be  estimated  from above by
\begin{equation}\label{mm pom 2}
\mathcal{W}
\leq  - \clb \, \tilde{A}_{\ks}  \, \mathfrak{m}_{k+\hg}  + 2^{\frac{3 \gamma +k }{2}} \mathcal{C}_k\left( \eta_{k}  \, \mathfrak{m}_{k+\hg}^{\theta_{k}},    +  \mathfrak{m}_{k} \, \mathfrak{m}_{1}   \right), \ \hk \geq \ks.
\end{equation}

 Next, in order to obtain a Bernoulli type of ODI for the time evolutions of the $k-$ Lebesgue moments associated to an a priori estimate to the solutions of the Cauchy problem \eqref{Cauchy},   it is needed to  absorb  the positive contribution of $\mathfrak{m}_{k+\hg}$   in  the second term of the right hand side of \eqref{mm pom 2} into  the negative one. This is perform  by means  of the  Young inequality \eqref{Y}.    Indeed,   recalling  definition of the constant     the estimate for the term 
$2^{\frac{3 \gamma +k}{2}}\, \eta_{k}  \, \mathfrak{m}_{k+\hg}^{\theta_{k}}, $   is performed by the following choice of parameters in the   inequality \eqref{Y}
\begin{equation*}
a= 2^{\frac{3 \gamma +k}{2}} \, \eta_{k}, \quad b= \mathfrak{m}_{k+\hg}^{\theta_{k}}, \quad q = \frac{1}{\theta_{k}},  \ \text{and} \ p = \frac{1}{1-\theta_{k}} \, ,
\end{equation*}
yields
\begin{equation}\label{Y2}
\begin{split}
2^{\frac{3 \gamma}{2}-1} \tilde{c}_k \eta_{k}  \, \mathfrak{m}_{k+\hg}^{\theta_{k}}
 &\leq    \epsilon^{-\frac{\theta_{k}}{1-\theta_k}}(1-\theta_{k})  \left(2^{\frac{3 \gamma +k }{2}}  \eta_{k} \right)^{\frac{1}{1-\theta_k}} + \epsilon \, \theta_k  \,   \mathfrak{m}_{k+\hg}
 \\ &\leq \epsilon^{-\frac{\theta_{k}}{1-\theta_k}}  \left(2^{\frac{3 \gamma +k }{2}} \tilde{c}_k \eta_{k} \right)^{\frac{1}{1-\theta_k}} + \epsilon  \, \mathfrak{m}_{k+\hg}.
 \end{split}
\end{equation}
Therefore,  \eqref{mm pom 2} becomes
\begin{equation}\label{mm pom 3}
\mathcal{W}
\leq  -  \left( \clb \, \tilde{A}_{\ks} -  \epsilon \, \mathcal{C}_k   \right) \mathfrak{m}_{k+\hg}  +  \tilde{D}_k \, \mathfrak{m}_{1} \, \mathfrak{m}_{k}, \ \hk \geq \ks,
\end{equation}
where the constant $\tilde{D}_k$ is given by the upper bound of the following expression  
\begin{multline}\label{Dk tilde}
  \max\left\{\epsilon^{-\frac{\theta_{k}}{1-\theta_k}}  \left(2^{\frac{3 \gamma}{2}-1}  \tilde{c}_k \eta_{k} \right)^{\frac{1}{1-\theta_k}} \frac{1}{\mathfrak{m}_{1}},  2^{\frac{3 \gamma}{2}-1} \tilde{B}_k \right\}
\\  \leq  \mathcal{C}_k  2^{\frac{3 \gamma +k}{2} }    \max\left\{ \left(\frac{2^{\frac{3 \gamma +k}{2} } }{ \epsilon} \right)^{\frac{\theta_{k}}{1-\theta_k}} \frac{ \eta_{k}^{\frac{1}{1-\theta_k}}}{\mathfrak{m}_{1}},  1 \right\} =: \tilde{D}_k.
\end{multline}
By the inequality \eqref{ck kappa lb} that upper bounds the  constant $\mathcal{C}_k$ for any $k>\bar{k}_*$, we obtain
\begin{equation}\label{mm pom 4}
\mathcal{W}
\leq  -  \left( \clb \, \tilde{A}_{\ks} -   \kappa^{lb}\, \epsilon   \right) \mathfrak{m}_{k+\hg}  +  \tilde{D}_k \, \mathfrak{m}_{1} \, \mathfrak{m}_{k}, \ \hk \geq \ks.
\end{equation}
Next,   choosing $\epsilon$, a positive constant  uniform in $k$, by setting
\begin{equation}\label{eps}
\epsilon := \frac{ \clb \, \tilde{A}_{\ks} }{2 \, \kappa^{lb}} = \frac{ \clb }{2} \left( 1 - \frac{ \mathcal{C}_{\ks}}{\kappa^{lb}} \right) >0,
\end{equation}
where we have used a definition of $\tilde{A}_{\ks}$ from \eqref{Ak tilde}. 

 In particular, the $k-$moment  of the collisional integral  \eqref{mm pom 5}  is bounded by 
\begin{equation}\label{mm pom 5}
\mathcal{W}
\leq  -   \frac{ \clb \, \tilde{A}_{\ks} }{2 }  \mathfrak{m}_{k+\hg}  +  \tilde{D}_k \, \mathfrak{m}_{1} \, \mathfrak{m}_{k}, \quad \hk \geq \ks,
\end{equation}
with $\tilde{D}_k$ as in \eqref{Dk tilde} with specified $\epsilon$ from \eqref{eps}.

Even further, applying  Jensen's inequality  to each  moment $ \mathfrak{m}_{k+\hg} $, $k>1$,  
\begin{multline*}
\int_{ \mathbb{R}^{3+}} f(v,I) \left\langle v, I \right\rangle^{\tk+\gamma} \mathrm{d}I \mathrm{d}v
\\
 \geq \left( \int_{ \mathbb{R}^{3+}} f(v, I) \mathrm{d}I  \mathrm{d}v \right)^{-\frac{\hg}{k}}  \left( \int_{\mathbb{R}^{3+}} f(v, I) \left\langle v, I \right\rangle^{\tk} \mathrm{d}I \mathrm{d}v \right)^{1+\frac{\hg}{k}}, 
\end{multline*}
or in terms of moments,
\begin{equation*}
\mathfrak{m}_{k+\hg} \geq \mathfrak{m}_0^{-\frac{\hg}{k}} \mathfrak{m}_k^{1+\frac{\hg}{k}} = \left\| f \right\|_{L^1_0}^{-\frac{\hg}{k}} \mathfrak{m}_k^{1+\frac{\hg}{k}}.
\end{equation*}
Hence,   In particular, the $k-$moment  of the collisional integral  \eqref{mathcal W} is bounded by   \ \eqref{mm pom} becomes
	\begin{equation*}
		\mathcal{W}
	\leq  -   \frac{ \clb \, \tilde{A}_{\ks} }{2 }  \,  \mathfrak{m}_{0}^{-\frac{\hg}{k}}  \mathfrak{m}_k^{1+\frac{\hg}{k}}+ {
 \tilde{D}_k }\, \mathfrak{m}_{\ho}  \, \mathfrak{m}_{k},
	\end{equation*}
where strictly positive the coercive  constant $c_{lb}$ is from Lemma \ref{lemma lower bound} and also the strictly positive constant  $\tilde{A}_{\ks}$  is from estimate  \eqref{Ak tilde} independent of $k$ as $k_*$ is fixed, depending on nonnegative factor calculated \eqref{k* const} from the Polyatomic Compact Manifold Averaging Lemma~\ref{lemma povzner}.
Denoting
\begin{align}\label{Aks-Bk}
0 <  A_{\ks} & :=      \frac{ \clb \, \tilde{A}_{\ks} }{2 }  \,  \mathfrak{m}_{0}^{-\frac{\hg}{k}} =  \frac{ \clb }{2}  \left( \kappa^{lb} -  \mathcal{C}_{\ks} \right)    \, \left\|f\right\|_{L^1_0}^{-\frac{\hg}{k_*}}, \ \qquad \text{and} \nonumber \\
B_k &:=  \tilde{D}_k \, \mathfrak{m}_{\ho}  \\ 
&
 =  \mathcal{C}_k  2^{\frac{3 \gamma + k}{2} }    \max\left\{ \left(\frac{\kappa^{lb} \, 2^{\frac{3 \gamma}{2} + \frac{k}{2}} }{ A_{k_*}} \right)^{\frac{\theta_{k}}{1-\theta_k}} \frac{\eta_{k}^{\frac{1}{1-\theta_k}}}{\left\|f\right\|_{L^1_1}},  1 \right\} \left\|f\right\|_{L^1_1},\nonumber,
\end{align} 
with $\eta_{k}$ as defined in \eqref{eta k} and $\theta_{k}$ as in \eqref{theta k},   
the last 
 inequality	 concludes the proof of Lemma \eqref{lemma Q} and the constant identities \eqref{moments const}.

\end{proof}

When regularizing properties of the collision operator stated in Lemma \ref{lemma Q} are combined with the Boltzmann equation \eqref{BE},   then we obtain ordinary differential inequality for $L^1$ polynomially weighted norms or polynomial moments $\mathfrak{m}_k$  in the sense of Definition \ref{def poly moment}.

\begin{lemma}[Ordinary Differential Inequality for Polynomial Moment]\label{Lemma ODI} If $f$   is a solution of the Boltzmann equation for polyatomic gases \eqref{BE} and  $\left\| f \right\|_{L_k^1}$ its associated norm of order $k$, then, for any $k>\ks$, with $\ks$ finite from \eqref{k from moment bound}, and $\gamma\in (0,2]$ the following estimate holds
	\begin{equation}\label{ODI poly}
	\frac{\mathrm{d} }{\mathrm{d}t} \left\| f \right\|_{L_k^1}  = \mathfrak{m}_k\left[  Q(f,f) \right] 
	\leq  - {A}_{\ks} \left\| f \right\|_{L_k^1}^{1+\frac{\hg}{k}}    +  {B}_k \left\| f \right\|_{L_k^1},
	\end{equation}
	where  both  ${A}_{\ks}$ and  ${B}_k$   are positive constants form  Lemma~\ref{lemma Q}, equations \eqref{Ak tilde} and \eqref{Bk tilde}, respectively.
	\end{lemma}

\begin{remark}\label{remark coercive factor}  The  constant $A_{\ks} = \frac{ \clb }{2}  \left( \kappa^{lb} -  \mathcal{C}_{\ks} \right)    \, \left\|f\right\|_{L^1_0}^{-\frac{\hg}{k_*}}, $ defined in \eqref{moments const} for $\ks$  specify in  \eqref{k from moment bound},  
  can be  identified as the $k$ independent {\em coercive   factor.}
This strictly positive factor, which controls the lower bound to the absorption term on the moments inequality \eqref{ODI poly}, provides a sufficient condition to proceed next with  Theorem~\ref{theorem bound on norm}  yielding the global in time  propagation and generation of $k^{th}$-moments of any order  $k>\ks$, provided that the initial data $f_0(v,I)$  has positive mass, positive energy, as much as  satisfies conditions of  the Lower Bound Lemma~\ref{lemma lower bound}.    

While these estimates are obtained without  the need of  entropy estimates,  yet, if the initial entropy is bounded,  the constructed solutions will have well defined  entropy \eqref{entropy} that remains bounded  for all times  by the initial one. 
\end{remark}

\begin{proof}
	In order to get ODI for the evolution of  $ \left\| f \right\|_{L_k^1}\!(t) = \pmk(t)$ defined in \ref{def poly moment},  we integrate  the Boltzmann equation \eqref{BE} over the space $(v,I) \in \mathbb{R}^{3+} $ against test function 
	$$\chi(v,I)= \langle v, I \rangle^{\tk}.$$
	Using the weak form \eqref{weak form}, we get
	\begin{equation*}
		\frac{\mathrm{d} }{\mathrm{d}t} \left\| f \right\|_{L_k^1} =
\frac{\mathrm{d} \, \pmk}{\mathrm{d}t} 
	=\int_{ \mathbb{R}^{3+}} Q(f,f)(v,I) \, \langle v, I \rangle^{\tk} \mathrm{d} I \mathrm{d}v.
	\end{equation*}
	Applying Lemma \ref{lemma Q} on the right-hand side we get the desired estimate.
\end{proof}

This differential inequality by means of a comparison principle for ODEs implies the following Theorem.

\begin{theorem}[Generation and propagation of polynomial moments]\label{theorem bound on norm}
If $f$ a solution of the Boltzmann equation \eqref{BE}  with the transition function from Assumption \ref{Sec: Ass B} , and the constants 	   ${A}_{\ks} >0 $ and $ 0\leq{B}_k$ bounded , for all $k>\ks$,  as defined in \eqref{ODI poly}, 
 then the following properties hold.	
	\begin{itemize}
		\item[1.](Generation)   Let the initial data $f_0(v,I)\in\L^1_{\ks}(\mathbb R^{3+)}$, i.e. $\mathfrak{m}_0[f](0)<\infty$,  then there is a constant $\mathfrak{C}^{\mathfrak{m}}$  uniformly in $k>\ks$,  $\ks$ from \eqref{k from moment bound}, such that for any $\gamma\in(0,2]$,
		\begin{align}\label{poly gen max t-2}
	&\qquad\  \ \mathfrak{m}_k[f](t)  \leq  \mathfrak{B}_{k} \ \max\{ 1, t^{-\frac{k}{{\hg}}}\},  \quad \forall t>0,  \quad \text{with} 
\\
	&\ \ \mathfrak{B}_{k} =  \left(\frac{{B}_k}{{A}_{\ks}}\right)^{\frac{k}{{\hg}}} 
\max\left\{ \left(\frac{{\gamma}{B}_k}{2k}  \right)^{-\frac{2k}{{\gamma}}} e^{\frac{B_k}{2}}, \left( 1- e^{-\frac{{\gamma} {B}_k}{ 2k} }\right)^{-\frac{2k}{{\gamma}}} \right\}, \ \   \text {for any} \ k> \ks. \nonumber
\end{align}

		\item[2.](Propagation) Moreover, if $\mathfrak{m}_k[f](0)<\infty$, then 
		\begin{equation}\label{poly propagation}
		\mathfrak{m}_k[f](t)  \leq \max\left\{\left(\frac{{B}_k}{{A}_{\ks}}\right)^{\frac{2k}{{\gamma}}}, \mathfrak{m}_k[f](0) \right\},  \quad\text{for all}\ t\geq 0.
		\end{equation}
	\end{itemize}
\end{theorem}

\begin{proof}
	The aim of the proof is to associate an   ODE of Bernoulli type  to the derived ODI \eqref{ODI poly} from Lemma \ref{Lemma ODI}.  A comparison or maximum principle argument for Ordinary Differential Equation argument enables us to establish that that solution to the initial value to the following family of ODEs, for each fixed $k\ge \ks$, and time $t>0$,
	\begin{equation}\label{ODI poly-2}
	\begin{split}
&	\left\{ 
	\begin{split}
&	\dfrac {\mathrm{d}}{\mathrm{d}t} 	\mathfrak{m}_k[f](t) =  \mathfrak{m}_k [Q(f,f)](t), \\[5pt]
&	\mathfrak{m}_k[f](0) =y_0,  
	\end{split}	\right.\\
&	\text{with} \ \  \mathfrak{m}_k [Q(f,f)](t)   \leq   - {A}_{\ks} \left\| f \right\|_{L_k^1}^{1+\frac{\hg}{k}}(t)    +  {B}_k \left\| f \right\|_{L_k^1}(t).
\end{split}
	\end{equation}	
	Thus, we solve the auxiliary upper ODE with the same initial data  as in \eqref{ODI poly-2},  
		\begin{equation}\label{upper ODE}
			\begin{split}
		&	\left\{ 
		\begin{split}
	 y'(t)&= - a \,y(t)^{1+c}+  b \,  y(t)\, , \ \qquad \text{for} \  a , \ b >0, \\[5pt]
	y(0)&=y_0,
		\end{split}	\right.\\
	&
	\text{with }\ \  \ a:= {A}_{\ks}, \ b:={B}_k, \ c:= \gamma/(2k),     \ \text{for each same}\  k>{\ks} \ \text{fixed}.
	\end{split}
	\end{equation}	
	
	That means any solution   $y(t)$ of the initial value problem to   \eqref{upper ODE} controls  from above any solution $ \mathfrak{m}_k[f](t),$   to the initial value problem for the ODI from \eqref{ODI poly-2}, with the same initial data for both problems.  That implies $y(0)\geq  \mathfrak{m}_k[f](0) >0$

%
	   In addition solutions to  $y(t)$ of the super  linear initial value problem given in \eqref{upper ODE}, are of  Bernoulli type, that is they   have a unique explicit solution  of the form
	   \begin{equation}\label{ode comparison}
	y(t) = \left(\frac{a}{b}\left(1-e^{-c \, b\, t}\right) + y_0^{-c} e^{-c\,b\,t}  \right)^{-\frac{1}{c}}.
	\end{equation}
	    Yet,  a simple argument allow us to describe the solution analytical properties. Indeed, the ODE  \eqref{upper ODE} has a unique stationary point at 
	  $(b/a)^{1/c}$, since  the initial data $y_0$ positive and finite,  will monotonically increase  whenever $y_0<(b/a)^{1/c}$, monotonically decrease whenever $y_0>(b/a)^{1/c}$, and both cases will converge for larger time $t$,  to the stationary value  when the initial data $y_0\equiv (b/a)^{1/c}$, making the  uniform in time  estimate
		\begin{equation}\label{ODE Ber prop up}
	y(t) \leq \max\{y(0),   (b/a)^{1/c}\}, \qquad \forall t \geq 0.
	\end{equation}
	Then,  by the maximum principle for ODE's the  solution of  \eqref{upper ODE}, controls from above    the solution of moments ODIs \eqref{ODI poly}. 
	
	Therefore, if $ 0<\mathfrak{m}_k[f](0)<y_0$, 
	for $k>\ks$, with $\ks$ from \eqref{k from moment bound},  estimate \eqref{ODE Ber prop up} implies 
	\begin{equation}\label{ODE Ber prop}
		\mathfrak{m}_k[f](t)  \leq  \max\left\{\left(\frac{{B}_k}{{A}_{\ks}}\right)^{\frac{2k}{{\gamma}}}, \mathfrak{m}_k[f](0) \right\}, \quad \text{for all}\  t>0,  
	\end{equation}
	with which means the polynomial $k$-moments are bounded uniformly in time, if initially so. Thus, the   propagation result \eqref{poly propagation} follows.

However,  the  moment's generation property is different since one needs to show 
$\mathfrak{m}_k[f](t)  \le y(t)$, with $k>\ks$, with an  initial data corresponds to a lower order moment that $k$,  that is  $y_0:=\mathfrak{m}_{\ks}[f](0)$ finite, but  not necessarily the value  of  $\mathfrak{m}_k[f](0)$  which could be infinity.   In this case we use  the exact from of the solution  of the  Bernoulli type ODE \eqref{ode comparison} in order to search for an  upper bound  to  $y(t)$ independent of the value of   $y_0$. In fact estimating from above the right hand side of the $y(t)$ solution \eqref{ode comparison} independently from the initial data $y_0$,  yields
	\begin{equation*}
	y(t) <  \left(\frac{a}{b}\left(1-e^{-c \, b\, t}\right) \right)^{-\frac{1}{c}}   \leq  \left(\frac{a}{b}  \right)^{-\frac{1}{c}} \begin{cases}
	\left(c \, b \right)^{-\frac{1}{c}}  e^{\frac{b}{2} }  \ t^{-\frac{1}{c}}, & t<1 \\
	\left(1-e^{-c \,b }  \right)^{-\frac{1}{c}}, & t\geq 1.
	\end{cases}
	\end{equation*} 
	Therefore,  an upper bound  to $\mathfrak{m}_k[f](t) $ uniformly in time readily follows, for any $k>\ks$ and initial finite data for just  $\mathfrak{m}_{\ks}[f](0)$,  clearly yields 
	{\begin{equation*}
	\mathfrak{m}_k[f](t)  \leq  \mathfrak{B}_{k} \ \max\{ 1, t^{-\frac{k}{{\hg}}}\}, \quad \forall t>0,
	\end{equation*}}
	where the constant is 
	\begin{equation*}
	\mathfrak{B}_{k} =  \left(\frac{{B}_k}{{A}_{\ks}}\right)^{\frac{k}{{\hg}}} 
 \max\left\{ \left(\frac{{\gamma}{B}_k}{2k}  \right)^{-\frac{2k}{{\gamma}}} e^{\frac{B_k}{2}}, \left( 1- e^{-\frac{{\gamma} {B}_k}{ 2k} }\right)^{-\frac{2k}{{\gamma}}} \right\}, \qquad  \text {for any} \ k> \ks, 
 \end{equation*}
 as stated in \eqref{poly gen max t-2}. Thus,  Theorem's~\ref{theorem bound on norm}  proof is now complete.

\end{proof}

\section{Existence and Uniqueness Theory}\label{Section Ex Uni proof}

It this Section, we establish an existence and uniqueness theorem that solves the Cauchy problem 
\begin{equation}\label{Cauchy}
\left\{ 
\begin{split}
&\partial_t f(t,v,I)=Q(f,f)(v,I)\\
&f(0,v,I)=f_0(v,I),
\end{split}
\right.
\end{equation}
under the Assumption \ref{Sec: Ass B} on the transition function $\mathcal{B}$.

\noindent{\bf The invariant region $\Omega$ to solve the Cauchy problem for Boltzmann equation for polyatomic gases.}
Our goal is to set conditions on initial data that  ensures existence and uniqueness of the solution to the Cauchy problem \eqref{Cauchy} associated to the Boltzmann equation  under conditions described in Section \ref{Sec: Suff}.

These conditions will include moments with physical interpretation of mass and energy of the gas, and the imposed  restrictions will be physically relevant. For instance, we will necessitate  bounded mass and energy both from below and above, thus excluding zero and infinitely large mass and energy. Moreover, we will require bounded moment  of order 
\begin{equation}\label{ks intro}
\ks:= \max\left\{ \bks, 1+\gamma, 1+ \pl/2 \right\},
\end{equation}	
for $\gamma \in (0,2]$ related to the  potential of the transition function \eqref{tf}, $\pl>0$ is from the lower bound \eqref{moment epsilon} and $\bks$ from \eqref{k from moment bound} is sufficiently large to ensure the prevail of the polynomial moments of loss term with respect to those same moments of the gain term. Such $\bks$ depends on $\gamma$, $\alpha$ and the model of transition function at hand.    
Now we define the bounded, convex and closed subset $\Omega \subseteq L_{\ho}^1$,
\begin{multline}\label{inv region}
\Omega =  \Big\{ f(v, I) \in L_{\ho}^1: f \geq 0 \ \text{in} \ (v, I),  \quad \int_{\mathbb{R}^{3+}}  v \, f \, \mathrm{d}I \,  \mathrm{d}v=0,  
\\
\exists \,  C_0, C_1, \text{ such that} \ \forall t\geq0, \
\mathfrak{m}_0[f](t) = C_0,  \ \mathfrak{m}_{1}[f](t) = C_1,   
\\
\mathfrak{m}_{\ks}[f](t) \leq \Cg, \ \text{with} \ \ks \ \text{from } \eqref{ks intro}   \  \Big\}.
\end{multline}
The value of $\ks$ which determines how many moments need to be bounded initially in order to guarantee existence and uniqueness to the Cauchy problem \eqref{Cauchy} is strongly related to the evolution of    of collision operator $Q(f,f)$, $\mathfrak{m}_k[Q(f,f)]$ . 
 More precisely, after the estimates  obtained in Sections~\ref{Coerciveness}, \ref{Sec: poly mom} and \ref{Sec: fund lemmas},  the a priori estimates 
 applied in  Lemma~\ref{Lemma ODI}, enable us to study an upper bound for  $\mathfrak{m}_k[Q(f,f)]$. That is, equivalent to study the map associated to the right hand side of Ordinary Differential Inequality   \eqref{ODI poly}, 
 
 More precisely,  we set $x:=
 \mathfrak{m}_{k}\!(t)$, and consider the   map  $ \mathcal{L}_{ k}(x): [0,\infty) \rightarrow \mathbb{R}, $
\begin{equation}\label{Lg} 
 \mathfrak{m}_{k}[Q(f,f)] = \int_{ \mathbb{R}^{3+}} Q(f,f)(v,I) \, \langle v, I \rangle^{\tk} \mathrm{d} I \mathrm{d}v  \le  \ \mathcal{L}_{ k}(x):=  - A_{\ks} \, x^{1+\frac{\hg}{k}} + B_{k} x,   
\end{equation}
  where $A_{\ks}$ is strictly  positive and independent of $k$,  and  $B_{k} $ non-negative  constants for any $\gamma>0$ for $k>\ks$ with $\ks$ 
  from \eqref{ks intro}.\\

 The next result follows.

\begin{theorem}[Existence and Uniqueness]\label{theorem existence uniqueness} Assume that $f(0,v, I)=f_0(v,I) \in \Omega$. Then the Boltzmann equation \eqref{Cauchy} for the transition function $\mathcal{B}$ under the Assumption \ref{Sec: Ass B}   has the unique solution in $\mathcal{C}\left(\left[0, \infty \right), \Omega \right) \cap \mathcal{C}^1\left(\left(0,\infty  \right), L_{\ho}^1 \right)$.
\end{theorem}

\begin{proof} The goal is apply general ODE theory from Appendix \ref{Appendix exi and uni}, that is to study collision operator $Q$ as mapping $Q: \Omega \rightarrow L_{\ho}^1$, and to show
	\begin{enumerate}
		\item H\"{o}lder continuity condition
	\begin{equation}\label{Holder}
	\left\| Q(f,f) - Q(g,g) \right\|_{L^1_{\ho}} \leq C_H   \left\| f-g \right\|_{L^1_{\ho}}^{1/2},
	\end{equation}
	\item Sub-tangent condition
		\begin{equation*}
	\lim\limits_{h\rightarrow 0+} \frac{\text{dist}\left(f + h Q(f,f), \Omega \right)}{h} =0,
	\end{equation*}
	where
	\begin{equation*}
	\text{dist}\left(\mathbb{H},\Omega\right)=\inf_{\omega \in \Omega} \left\| h-\omega \right\|_{L_{\ho}^1},
	\end{equation*}
	\item One-sided Lipschitz condition
	\begin{equation}\label{Lipschitz}
	\left[  Q(f,f) - Q(g,g), f-g \right] \leq C_L  \left\| f-g \right\|_{L^1_{\ho}},
	\end{equation}
	where  brackets $\left[  \cdot,  \cdot \right] $ by Remark \ref{Lip remark} become
	\begin{multline*}
	\left[ Q(f,f) - Q(g,g), f- g\right] 
	\\
	=\lim_{h\rightarrow 0^-}\frac{\left(\left\| \left(f-g\right) + h \left( Q(f,f)-Q(g,g)\right) \right\|_{L_{\ho}^1} - \left\| \left(f-g\right) \right\|_{L_{\ho}^1} \right)}{h} 
	\\
	\leq  \int_{ \mathbb{R}^{3+}} \left( Q(f,f)(v,I) - Q(g,g)(v,I)\right) \, \text{sign}\left(f(v,I) - g(v,I)\right) \left\langle v, I \right\rangle^2 \mathrm{d}I \, \mathrm{d}v.
	\end{multline*}
		\end{enumerate}
First, we check $Q: \Omega \rightarrow L_{\ho}^1$ is well defined. Indeed, for  any $f \in \Omega$, using $\left| \cdot \right|= \cdot \, \text{sign}(\cdot)$
\begin{multline*}
\left\| Q(f,f) \right\|_{L^1_{\ho}} = \int_{ \mathbb{R}^{3+}}  Q(f,f)(v,I) \ \text{sign}\left( Q(f,f)(v,I) \right) \, \langle v, I \rangle^2 \, \mathrm{d} I \, \mathrm{d}v
\\
\leq \frac{1}{2} \int_{({\mathbb{R}^{3+}})^2 \times K }  \frac{ f f_*}{(I I_*)^\alpha} \left( \langle v', I' \rangle^2 + \langle v'_*, I'_* \rangle^2 + \langle v, I \rangle^2 + \langle v_*, I_* \rangle^2\right) \mathrm{d}A
\end{multline*}
by virtue of the weak form \eqref{weak form} for the test function 
$$\chi(v,I) = \text{sign}\left( Q(f,f)(v,I) \right) \, \langle v, I \rangle^2.$$
Using microscopic energy law \eqref{micro CL} and the form of transition function \eqref{trans prob rate ass}   with the multiplying functions from above together with the upper bound \eqref{up tf}, as much as monotonicity of moments \eqref{monotonicity of norm} we get
\begin{multline*}
\left\| Q(f,f) \right\|_{L^1_{\ho}} \leq 2^{\frac{3 \gamma}{2} - 1}  \kappa^{ub} 
  \int_{({\mathbb{R}^{3+}})^2  }  f f_* \left( \langle v, I \rangle^2 + \langle v_*, I_* \rangle^2\right)  \left( \langle v, I \rangle^\gamma  +  \langle v_*, I_* \rangle^\gamma \right) \mathrm{d}I_*\, \mathrm{d}v_* \mathrm{d}I \, \mathrm{d}v
  \\ =   2^{\frac{3 \gamma}{2} }  \kappa^{ub} \left( \left\| f \right\|_{L^1_{\ho+\hg}} \left\| f \right\|_{L^1_{0}}  +  \left\| f \right\|_{L^1_{\ho}} \left\| f \right\|_{L^1_{\hg}}  \right) \leq  2^{\frac{3 \gamma}{2} + 1}  \kappa^{ub} \,\left\| f \right\|_{L^1_{\ho+\hg}} \left\| f \right\|_{L^1_{\ho}},
\end{multline*}
with $ \kappa^{ub}$ from \eqref{kappas}.
Since $f\in\Omega$, the right hand  side is bounded, and thus $Q(f,f) \in L_{\ho}^1$.\\

 Then, the proof consists in three parts. \\

\noindent\emph{Part I: H\"{o}lder continuity condition.} We first rewrite difference of the two collision operators acting on distribution functions $f$ and $g$ as collision operator on sums and differences of these two distribution functions,
\begin{equation}\label{coll op difference}
Q(f,f) - Q(g,g) = \frac{1}{2} \left( Q(f+g, f-g) + Q(f-g, f+g)\right),
\end{equation}
by virtue of the bilinear structure of the strong form of collision operator \eqref{collision operator}. 
Using this relation, we write the $L^1_1$ norm
\begin{multline*}
\mathcal{I}_H := \left\| Q(f,f) - Q(g,g) \right\|_{L^1_1} = \int_{ {\mathbb{R}^{3+}}} \left|  Q(f,f) - Q(g,g) \right| \langle v, I \rangle^2 \, \mathrm{d}I \, \mathrm{d}v
\\
\leq \frac{1}{2} \int_{ {\mathbb{R}^{3+}}}   \left(  \left|  Q(f+g, f-g)  \right|  + \left| Q(f-g, f+g) \right| \right) \langle v, I \rangle^2 \, \mathrm{d}I \, \mathrm{d}v.
\end{multline*}
The absolute value of collision operators will be rewritten in terms of the sign function using $\left| \cdot \right| = \cdot \, \text{sign}(\cdot)$, that will be understood as a test function. This implies
\begin{multline*}
\mathcal{I}_H
\leq \frac{1}{2} \int_{ ({\mathbb{R}^{3+}})^2 \times \K}  \left(  \left(f(v, I) + g(v, I)\right)  \left| f(v_*, I_*) - g(v_*, I_*)\right| 
\right. \\ \left.
+  \left| f(v, I) - g(v, I)\right|  \left( f(v_*, I_*) + g(v_*, I_*)\right)  \right) \\
\times \left( \langle v', I' \rangle^2 +  \langle v'_*, I'_* \rangle^2  +  \langle v, I \rangle^2  +  \langle v_*, I_* \rangle^2 \right)\\
\times  \mathcal{B} \,  \varphi_\alpha(r) \, \psi_\alpha(R)  \,(1-R) R^{1/2}  \,  \mathrm{d}R   \, \mathrm{d}r \, \mathrm{d}\sigma \, \mathrm{d}I_*\, \mathrm{d}v_* \mathrm{d}I \, \mathrm{d}v
\\
=
\int_{ ({\mathbb{R}^{3+}})^2 \times \K}  \left(  \left(f(v, I) + g(v, I)\right)  \left| f(v_*, I_*) - g(v_*, I_*)\right| 
\right. \\ \left.
+  \left| f(v, I) - g(v, I)\right|  \left( f(v_*, I_*) + g(v_*, I_*)\right)  \right) 
\, \left( \langle v, I \rangle^2  +  \langle v_*, I_* \rangle^2 \right)\\
\times  \mathcal{B} \, \varphi_\alpha(r) \,  \psi_\alpha(R) \,(1-R) R^{1/2} \, \mathrm{d}R \, \mathrm{d}r  \, \mathrm{d}\sigma  \mathrm{d}I_*\, \mathrm{d}v_* \mathrm{d}I \, \mathrm{d}v,
\end{multline*}
and the last equality is by energy conservation law during collision \eqref{micro CL}. Now we make use of the transition function \eqref{trans prob rate ass},  its bound from above \eqref{up tf}, and the upper bound constant $\kappa^{ub}$ for the integration over the compact manifold $\K$ given in \eqref{kappas}, 
\begin{multline}\label{hol1}
\mathcal{I}_H
\leq  2^{\frac{3 \gamma}{2} - 1}  \kappa^{ub} 
\int_{ ({\mathbb{R}^{3+}})^2 }  \left(  \left(f(v, I) + g(v, I)\right)  \left| f(v_*, I_*) - g(v_*, I_*)\right| 
\right. \\ \left.
+  \left| f(v, I) - g(v, I)\right|  \left( f(v_*, I_*) + g(v_*, I_*)\right)  \right) 
\\
\times \left( \langle v, I \rangle^2  +  \langle v_*, I_* \rangle^2 \right) \left( \langle v, I \rangle^\gamma  +  \langle v_*, I_* \rangle^\gamma \right) \mathrm{d}I_*\, \mathrm{d}v_* \mathrm{d}I \, \mathrm{d}v.
\end{multline}
We rewrite \eqref{hol1} in moment notation,
\begin{multline*}
\mathcal{I}_H
\leq   2^{\frac{3 \gamma}{2}}  \kappa^{ub} 
\left( \left\| f + g \right\|_{L^1_{\ho+\hg}}  \left\| f-g \right\|_{L^1_0}  
+  \left\| f + g \right\|_{L^1_{\ho}}  \left\| f-g \right\|_{L^1_{\hg}} 
\right. \\ \left.
+  \left\| f + g \right\|_{L^1_{\hg}}  \left\| f-g \right\|_{L^1_{\ho}}  
+  \left\| f + g \right\|_{L^1_{0}}  \left\| f-g \right\|_{L^1_{\ho+\hg}}     \right).
\end{multline*}
Furthermore, monotonicity of norms \eqref{monotonicity of norm} yields
\begin{equation*}
\mathcal{I}_H
\leq  2^{\frac{3 \gamma}{2} + 1}  \kappa^{ub}  \, \left\| f-g \right\|_{L^1_{\ho+\hg}}  
\left( \left\| f + g \right\|_{L^1_{\ho+\hg}} 
+  \left\| f + g \right\|_{L^1_{\ho}}   \right).
\end{equation*}
Next, we use interpolation inequality \eqref{interpolation inequality} on $ \left\| f-g \right\|_{L^1_{\ho+\hg}}  $ and get
\begin{equation*}
 \left\| f-g \right\|_{L^1_{\ho+\hg}}   \leq  \left\| f-g \right\|_{L^1_{\ho}}^{1/2}  \left\| f-g \right\|_{L^1_{\ho+\gamma}}^{1/2}.    
\end{equation*}
Moreover, characterization of the set $\Omega$ gives the following bounds
\begin{equation*}
\left\| f-g \right\|_{L^1_{\ho+\gamma}}^{1/2}  \leq \left\| f \right\|_{L^1_{\ho+\gamma}}^{1/2}  + \left\| g \right\|_{L^1_{\ho+\gamma}}^{1/2} \leq 2 \, C_{\ho+\gamma}^{1/2},
\end{equation*}
and 
\begin{equation*}
\left\| f + g \right\|_{L^1_{\ho+\hg}}  \leq 2 \, C_{\ho+\hg}, \  \ \left\| f + g \right\|_{L^1_{\ho}} \leq 2 \, C_{\ho}. 
\end{equation*}
Finally,  denoting 
\begin{equation*}
C_H :=    2^{\frac{3 \gamma}{2} +3}  \kappa^{ub}  \, C_{\ho+\gamma}^{1/2} \left( C_{\ho+\hg}  + C_{\ho}  \right),
\end{equation*}
we get desired estimate \eqref{Holder}.\\

\noindent\emph{Part II: Sub-tangent condition.} We first study the collision frequency
\begin{equation*}
\nu(f)(v,I) := \int_{{\mathbb{R}^{3+}} \times K} f(v_*, I_*) \mathcal{B} \, \varphi_\alpha(r) \psi_\alpha(R) (1-R) R^{1/2}  \, \mathrm{d}R  \, \mathrm{d}r \, \mathrm{d} \sigma  \, \mathrm{d}I_* \, \mathrm{d}v_*.
\end{equation*}
Using the form \eqref{trans prob rate ass} of the transition function $\mathcal{B}$ together with its bound from above \eqref{up tf}, we obtain
\begin{align}\label{coll freq estimate}
\nu(f)(v,I) &\leq  2^{\frac{3 \gamma}{2} - 1}  \kappa^{ub} \int_{{\mathbb{R}^{3+}} } f(v_*, I_*) \left( \langle v, I \rangle^\gamma +  \langle v_*, I_* \rangle^\gamma \right)\, \mathrm{d}I_* \, \mathrm{d}v_*   \nonumber
\\
&\leq   2^{\frac{3 \gamma}{2} - 1}  \kappa^{ub}  \left( C_0 \langle v, I \rangle^\gamma + C_{\hg} \right)  \nonumber
\\& \leq   2^{\frac{3 \gamma}{2} }  \kappa^{ub} \left( C_0 +  C_{\hg} \right) \left( 1+ \left(  \frac{1}{2} \left| v \right|^2 + \frac{I}{m}\right)^{\gamma/2} \right),
\end{align}
with $\kappa^{ub}$ from \eqref{kappas} and using characterization of the invariant region $\Omega$ as in \eqref{inv region}.

The idea of the proof of sub-tangent condition is to prove that for $f\in \Omega$ and for any $\epsilon>0$ there exists $h_1>0$ such that the interval $\mathcal{N}$ centered at $f+h Q(f,f)$ with radius $h \epsilon$ denoted by
\begin{equation}\label{interval}
\mathcal{N}(f+h Q(f,f), h \epsilon),
\end{equation}
has non-empty intersection with $\Omega$ for any $0<h<h_1$, as formulated in the Proposition \ref{Prop sub tangent} below. Then for  such $h_1$ we have
\begin{equation*}
h^{-1} \text{dist}\left( f + h Q(f,f), \Omega \right) \leq \epsilon, 
\end{equation*}
for all $0<h<h_1$,  which concludes the sub-tangent condition. Therefore, it remains to prove the following Proposition \ref{Prop sub tangent}.

\begin{proposition}\label{Prop sub tangent}
	Fix $f\in \Omega$. Then for any $\epsilon>0$ there exists $h_1>0$ such that
\begin{equation}\label{ball sub tangent}
\mathcal{N}(f+h Q(f,f), h \epsilon) \cap \Omega \neq \emptyset,
\end{equation}	
for any $0<h<h_1$.
\end{proposition}
\begin{proof}
	We  recall the definition of the semi-sphere in the velocity-internal energy space \eqref{ball rho},  
	\begin{equation*}
	B_\rho(0, 0) := \left\{ (v, I) \in \mathbb{R}^3 \times [0, \infty): \sqrt{ \frac{1}{2} \left| v \right|^2 + \frac{I}{m} } \leq \rho  \right\}.
	\end{equation*}
	Then  with the help of the  characteristic function of this sphere $B_\rho(0, 0)$, we  define truncated distribution function
\begin{equation}\label{f rho}
f_\rho(t, v, I) := f(t, v, I) \, \mathbbm{1}_{B_\rho(0, 0)}(v, I),  
\end{equation}
and consider the following function
\begin{equation}\label{w}
g_\rho = f + h \, Q(f_\rho, f_\rho),\quad \text{for} \ h>0.
\end{equation}

Our goal is to find $\rho$ and $h$ so that $g_\rho \in \mathcal{N}(f+h Q(f,f), h \epsilon) \cap \Omega $.\\

We first note that for any $f\in \Omega$, its truncation $f_\rho \in \Omega$ as well. Then using definition \eqref{collision operator-pull-out} yields
\begin{multline*}
Q(f_\rho, f_\rho)(v, I) \geq Q^-(f_\rho, f_\rho)(v, I) \\= - f_\rho \int_{ {\mathbb{R}^{3+}} \times K}  \, f_{\rho *}  \, \mathcal{B} \, \varphi_\alpha(r) \, \psi_\alpha(R) \, (1-R) R^{1/2} \,  \mathrm{d}R \, \mathrm{d}r  \, \mathrm{d} \sigma \, \mathrm{d}I_* \, \mathrm{d}v_*,
\end{multline*}
since the gain term is positive. Next, using an upper bound on the collision frequency \eqref{coll freq estimate}, yields the following lower bound estimate on the collision operator acting on $f_\rho$,
\begin{multline*}
Q(f_\rho, f_\rho)(v, I) \geq -  2^{\frac{3 \gamma}{2} }  \kappa^{ub} \left( C_0 +  C_{\hg}  \right) \left( 1+ \left(  \frac{1}{2} \left| v \right|^2 + \frac{I}{m}\right)^{\gamma/2} \right) f_\rho
\\
\geq -   2^{\frac{3 \gamma}{2} }  \kappa^{ub} \left( C_0 + C_{\hg} \right) \left( 1+ \rho^{\gamma} \right) f.
\end{multline*}
Therefore, for $g_\rho$ we can bound
\begin{equation*}
g_\rho \geq f \left( 1 -    2^{\frac{3 \gamma}{2} }  \kappa^{ub} \left( C_0 +  C_{\hg}   \right) \left( 1+ \rho^{\gamma} \right) h \right) \geq 0,
\end{equation*}
for any $h \in(0, \frac{1}{  2^{\frac{3 \gamma}{2} }  \kappa^{ub} \left( C_0 +  C_{\hg}  \right) \left( 1+ \rho^{\gamma} \right)})$.

Next, weak form \eqref{weak form} implies 
\begin{equation*}
\int_{\mathbb{R}^{3+}} Q(f_\rho, f_\rho)(v, I) \, \mathrm{d}I \, \mathrm{d}v =0, \quad \int_{\mathbb{R}^{3+}} Q(f_\rho, f_\rho)(v, I) \langle v, I \rangle^2 \, \mathrm{d}I \, \mathrm{d}v =0,
\end{equation*}
which yields
\begin{equation*}
\mathfrak{m}_0[g_\rho](t) = \mathfrak{m}_0[f](t), \quad \mathfrak{m}_{\ho}[g_\rho](t) = \mathfrak{m}_{\ho}[f](t),
\end{equation*}
independently of $\rho$ and $h$. Then upper and lower bounds for these polynomial moments of  $f$ imply the same kind of estimates for $g_\rho$.

Finally, as anticipated at the opening of this section, we  prove that $L^1_{\ks}$ norm of $g_\rho$ is bounded.

To this end we study the  the map 
$\mathcal{L}_{ k}(x): [0,\infty) \rightarrow \mathbb{R}$ from \eqref{Lg}, 
\begin{equation*}
\mathcal{L}_{ k}(x):= - A_{\ks} x^{1+\frac{\hg}{k}} + B_{k} x,
\end{equation*}
where $A_{\ks}$ and $B_{k}$ are positive constants for any $\gamma>0$, $k>\ks$, as described in \eqref{Lg}.  That means, this map has only one root, denoted by 
\begin{equation}\label{xhat}
\hat x_{k}:=  \left(\frac {B_{k}}{A_{\ks}}\right)^{\frac{\hg}{k + \hg}}, \quad k\geq \ks,
\end{equation}
at which $\mathcal{L}_{ k}(x)$ changes from positive to negative. 
 
 Thus, we  set, for any $x \geq 0$,
\begin{equation*}
\mathcal{L}_{k}(x) \leq \max_{0 \leq x \leq \hat x_{\gamma,k}} \mathcal{L}_{ k}(x) =: \hat{\mathcal{L}}_{ k},
\end{equation*}
since   such maximum can be explicitly computed by
\begin{multline}\label{L st ks}
\hat{\mathcal{L}}_{k}:=     \mathcal{L}_{k}\left( \left(\frac {B_{k} k}{A_{\ks}(k +\gamma/2)  }\right)^{\frac{k}{\hg}}\right) =     \mathcal{L}_{k}\left(    \hat x_{k}^{\frac{k + \hg}{\hg}} \frac{k}{k +\gamma/2}  \right)^{\frac{k}{\hg}}  := K_{\gamma,  k}  
%
 \end{multline}
 and    $K_{\gamma,  k} $ also depends on the initial mass, energy and  $\|f_0\|_{L^1_{\ks}}$.
 %

In particular, from  \eqref{poly coll op} it follows that, for any $k>1$
	\begin{equation*}
\int_{ {\mathbb{R}^{3+}}} Q(f,f)(v,I) \, \langle v, I \rangle^{k} \mathrm{d} I \mathrm{d}v \leq \mathcal{L}_{ k}\left(\mathfrak{m}_{k}[f]\right) \leq  \hat{\mathcal{L}}_{\gamma, k}.
\end{equation*}
Next, define 
\begin{equation}\label{c dva plus dva gama def}
 C_{k} = \hat x_{k} + \hat{\mathcal{L}}_{\gamma, k}.
\end{equation} 
 
For any $f \in \Omega$, we have two possibilities: either $(i)$ $\mathfrak{m}_{k}[f] \leq \hat x_{k}$ or $(ii)$ $\mathfrak{m}_{k}[f] > \hat x_{k}$. 

In the first case,  for the $k$-moment of $g_\rho$ we get
\begin{equation*}
\mathfrak{m}_{k}[g_\rho]  \leq \hat x_{k} + h \int_{ {\mathbb{R}^{3+}}} Q(f_\rho,f_\rho)(v,I) \, \langle v, I \rangle^{2 k} \mathrm{d} I \mathrm{d}v \leq \hat x_{k} + h  \hat{ \mathcal{L}}_{\gamma, k} \leq C_{k},
\end{equation*}
where we have used $h \leq 1$, without loss of generality.

In the second case, we take $\rho=\rho(f)>0$ large enough to ensure  $\mathfrak{m}_{k}[f_\rho] > \hat x_{k}$ as well. In that case, $\mathcal{L}_{k}$ is negative, i.e.
\begin{equation*}
\mathcal{L}_{k}\left( \mathfrak{m}_{k}[f_\rho]   \right) \leq 0.
\end{equation*} 
Therefore,
\begin{equation*}
\mathfrak{m}_{k}[g_\rho] \leq \hat x_{k} \leq C_{k}, \quad k\geq \ks.
\end{equation*}
Therefore, in either case $\mathfrak{m}_{k}[g_\rho]$ is bounded, and moreover we have constructed the constant of boundedness $C_{k}$ for any $k\geq k_*$.

We conclude that $g_\rho \in \Omega$ provided that $0<h<h_*$,
\begin{equation*}
h_* = \min\left\{ 1, \frac{1}{  2^{\frac{3 \gamma}{2} }  \kappa^{ub} \left( C_0 +C_{\hg} \right) \left( 1+ \rho(f)^{\gamma} \right)}\right\}
\end{equation*}

On the other hand, let us show that $g_\rho \in  B(f+h Q(f,f), h \epsilon) $. From the H\"{o}lder estimate \eqref{Holder} we get
\begin{equation*}
h^{-1} \left\|  f + h Q(f, f) - g_\rho \right\|_{L^1_{\ho}} = \left\|  Q(f, f) - Q(f_\rho, f_\rho) \right\|_{L^1_{\ho}} \leq C_H  \left\|  f - f_\rho \right\|_{L^1_{\ho}}^{1/2} \leq \epsilon,
\end{equation*}
for $\rho = \rho(\epsilon)$ large enough. Thus, for this choice of $\rho $, we have $g_\rho \in  B(f+h Q(f,f), h \epsilon) $.\\

Finally, we conclude the proof of the Proposition by choosing 
\begin{equation*}
\rho= \max\left\{\rho(f), \rho(\epsilon)  \right\}, \ \text{and} \ h_1=  \min\left\{ 1, \frac{1}{2 \, \kappa \left( C_0 + C_{\hg} \right) \left( 1+ \rho^{\gamma} \right)} \right\}.
\end{equation*}

\end{proof}

\noindent\emph{Part III: One-sided Lipschitz condition.}  The left hand side of \eqref{Lipschitz} after use of representation \eqref{coll op difference} becomes
\begin{multline*}
\mathcal{I}_L := \left[ Q(f,f) - Q(g,g), f-g \right]
\\
\leq  \int_{ {\mathbb{R}^{3+}}} \left( Q(f,f)(v,I) - Q(g,g)(v,I)\right) \, \text{sign}\left(f(v,I) - g(v,I)\right) \left\langle v, I \right\rangle^2 \mathrm{d}I \, \mathrm{d}v
\\
\leq \frac{1}{2}  \int_{ {\mathbb{R}^{3+}}}\left( Q(f+g, f-g)(v,I) + Q(f-g, f+g)(v,I)\right)  
\\ \times \text{sign}\left(f(v,I) - g(v,I)\right) \left\langle v, I \right\rangle^2 \mathrm{d}I \, \mathrm{d}v.
\end{multline*}
Using the weak form \eqref{weak form}, we get
\begin{multline*}
\mathcal{I}_L 
\leq \frac{1}{4}  \int_{ ({\mathbb{R}^{3+}})^2\times K} \left( \frac{(f+g)(f_*-g_*) }{(II_*)^\alpha} +  \frac{(f-g)(f_*+g_*) }{(II_*)^\alpha}  \right) 
\\
\times \left(  \text{sign}\left(f' - g'\right) \left\langle v', I' \right\rangle^2 +  \text{sign}\left(f'_* - g'_*\right) \left\langle v'_*, I'_* \right\rangle^2 \right. \\ \left. -  \text{sign}\left(f - g\right) \left\langle v, I \right\rangle^2 - \text{sign}\left(f_* - g_*\right) \left\langle v_*, I_* \right\rangle^2 \right) \mathrm{d}A
\end{multline*}
We bound sign function by 1 for the first two terms, and from the last two terms we use $\left| \cdot \right|= \cdot \, \text{sign}(\cdot)$ where applicable,
\begin{multline*}
\mathcal{I}_L 
\leq \frac{1}{4}  \int_{ ({\mathbb{R}^{3+}})^2\times K} \left\{ \left( (f+g) \left|f_*-g_*\right| +  \left|f-g\right| (f_*+g_*)  \right) 
\times \left(  \left\langle v', I' \right\rangle^2 +  \left\langle v'_*, I'_* \right\rangle^2 \right) 
\right. 
\\
\left.
+ \left( (f+g) \left| f_*-g_* \right| -  \left| f-g \right| (f_*+g_*)  \right) 
\left\langle v, I \right\rangle^2
 \right. 
 \\
 \left.
 + \left( - (f+g)\left| f_*-g_* \right| +  \left| f-g \right| (f_*+g_*)  \right) 
 \left\langle v_*, I_* \right\rangle^2
  \right\} \frac{\mathrm{d}A}{(II_*)^\alpha}.
\end{multline*}
Using the energy collision law \eqref{micro CL}, after cancellations of some terms we get 
\begin{multline*}
\mathcal{I}_L 
\leq \frac{1}{2}  \int_{ ({\mathbb{R}^{3+}})^2\times K} \left( (f+g) \left|f_*-g_*\right| \left\langle v, I \right\rangle^2 +  \left|f-g\right| (f_*+g_*) \left\langle v_*, I_* \right\rangle^2 \right)  \frac{\mathrm{d}A}{(II_*)^\alpha}
\\
= \int_{ ({\mathbb{R}^{3+}})^2\times K}   \left|f-g\right| (f_*+g_*) \left\langle v_*, I_* \right\rangle^2  \frac{\mathrm{d}A}{(II_*)^\alpha},
\end{multline*}
the last equality is due to the change of variables \eqref{microreversibility changes star} in the first integral. 
We make use of the transition function $\mathcal{B}$ assumption \eqref{trans prob rate ass} and the upper bound \eqref{up tf},
\begin{multline*}
\mathcal{I}_L 
\leq {C}_K \int_{ ({\mathbb{R}^{3+}})^2 }   \left|f-g\right| (f_*+g_*) \left\langle v_*, I_* \right\rangle^{2} \left( \left\langle v, I \right\rangle^\gamma + \left\langle v_*, I_* \right\rangle^\gamma\right)  \mathrm{d} I_*\, \mathrm{d}v_* \, \mathrm{d}I \, \mathrm{d}v
\\ =  2^{\frac{3 \gamma}{2} - 1}  \kappa^{ub} \left( \left\|f-g\right\| _{L^1_{\hg}} \left\|f+g\right\| _{L^1_{\ho}} +  \left\|f-g\right\| _{L^1_0} \left\|f+g\right\| _{L^1_{\ho+\hg}}  \right) 
\\
\leq  2^{\frac{3 \gamma}{2} }  \kappa^{ub}  \left( C_{\ho} + C_{\ho+\hg} \right)\left\|f-g\right\| _{L^1_\ho} 
\end{multline*}
where $\kappa^{ub}$ is from \eqref{kappas}, and we have used monotonicity of norms \eqref{monotonicity of norm} and definition of the set $\Omega$ from \eqref{inv region}, which concludes the proof.

\end{proof}

\section{Generation and propagation of  exponential moments}\label{Section gen prop exp mom}

In the case of single monatomic gas \cite{Bob97}, \cite{Mouhot06}, \cite{GambaPanfVil09},  \cite{Gamba13},  \cite{GambaTask18},  and more recently in \cite{Alonso-IG-BAMS} and monatomic gas mixtures \cite{IG-P-C}, generation and propagation of polynomial moments implied  the same properties of exponential moments, that {   polynomial }  moments estimates to solutions of the Boltzmann models allow for  the summability of {   polynomial }  moments  if the initial data 
has this summability property.
 The expert reader  can easily observe {   that}  the analog results hold for the Boltzmann equation for polyatomic gases. However the expository proof we present in this section makes modifications that would clarify many points {   of these techniques}  to first time readers. 

The notion of exponential moments as defined in \eqref{exp moment} is  associated  to the finding conditions for the summability of {   polynomial} moments {   (or only moments) }  to be a convergent series. More precisely, 
 the conditions for the  summability of moments propagation, {   if} the initial data has that the same  property, is related to show that there is an associated geometric convergent series of moments, whose   radius of convergence is the rate  of decay  in such exponential moment form.   The natural link between these {   two} objects is provided {   by} the Taylor series of an exponential form, formally written
\begin{align}\label{exp mom def 2}
\mathcal{E}_{s}[f]({\beta},t) &:=   \int_{\mathbb{R}^{3+}} f(t,v,I) \, e^{\beta \left\langle v, I \right\rangle^{2s}} \mathrm{d} I \, \mathrm{d}v\\
&= \int_{\mathbb{R}^{3+} }f(t, v, I) \sum_{k=0}^\infty  \langle v,I\rangle ^{2sk}\frac{\beta^k}{k!} 
 \mathrm{d} I \, \mathrm{d}v=\sum_{k=0}^\infty \frac{\beta^k}{k!} \, \mathfrak{m}_{s k}[f](t),   \quad\text{for} \ t>0,  \nonumber
\end{align}
which we refer as to exponential moments with order $2s$, with $s\in (0,1]$,  $s=1$ corresponding to Gaussians,  and rate $\beta    > 0$.

Indeed, our goal is to show {   that} the constructed solutions  in  Section~\ref{Section Ex Uni proof} can propagate  or generate exponential moments depending on the integrability properties of the exponential moment of the initial data.  Propagation of initial data in this context  means given $f_0\in \Omega$, a order factor $0<s    \leq 1$ and a rate $0<\beta_0$ such that   $\mathcal{E}_{s}[f]({\beta_0},0)$ is finite, then the solution of the Cauchy problem for $f(t, v,I)$ posed in \eqref{Cauchy}  satisfies   $\mathcal{E}_{s}[f]({   \beta},t)$ is finite for all time $t\ge0$ {   and $0<\beta \leq \beta_0$} . 
However, generation of data means the following stronger property:  given just  $f_0\in \Omega$,  (not necessarily  $\mathcal{E}_{s}[f]({\beta_0},0)$ finite for any order $2s$, $ 0<s\leq 1$, and $0<\beta_0$) the solution of the corresponding Cauchy problem \eqref{Cauchy} satisfies  that $\mathcal{E}_{s}[f]({\beta_0},t)$ is finite for all time $t>0$, for some order factor $s$ and rate $\beta$ to be found depending on the data. 

The following Theorem proves the accuracy of these two statements. Their proofs consists in developing ordinary differential inequalities for  the quantities $\mathcal{E}_{s}[f]({\beta},t) $ valid for $t>0$, whose initial data is referred as  to $\mathcal{E}_{s}[f]({\beta},t)\mid_{t=0}$.

\medskip

\begin{theorem}[{\bf Propagation and Generation of  exponential moments}]\label{theorem gen prop ML}
	Let $f$ be the solution of the Cauchy problem \eqref{Cauchy}. The following properties hold.
	\begin{itemize}
		\item[(a)](Propagation) Let $0<s\leq \ho$. Suppose that there exists a  constant $\beta_0 >0$, such that the initial data $f_0(v,I)$ has a bounded exponential moment of order $2s$ and rate $\beta_0$ 
		\begin{equation}\label{initial data exp prop}
		\mathcal{E}_{s}[f]({\beta},t) \mid_{t=0} = \mathcal{E}_{s}[f_0]({\beta_0},0) = \int_{\mathbb{R}^{3+}} f_0(v,I) \, e^{\beta_0 \left\langle v, I \right\rangle^{2s}} \mathrm{d} I \, \mathrm{d}v =: M_P  < \infty, 
		\end{equation}
		then there exist  a  constant $0<\beta \leq \beta_0$ 
		such that
		\begin{equation}\label{exp propagation p}
		\mathcal{E}_{s}[f]({\beta},t)   \leq 3  \, \mathcal{E}_{s}[f_0]({\beta_0},t)  = 3M_P\qquad \forall t\geq 0.
		\end{equation}
		
		\item[(b)](Generation)   If the prescribed initial data is  in the solution set   $\Omega$ defined in \eqref{inv region}, that is 
		\begin{equation}\label{initial data exp gen}
		\mathcal{E}_{s}[f]({\beta},t)\mid_{t=0} = \mathfrak{m}_{\ks}[f_0]  =\int_{\mathbb{R}^{3+}} f_0(v,I) \, \left\langle v, I \right\rangle^{2\ks} \mathrm{d} I \, \mathrm{d}v =: M_G  < \infty, 
		\end{equation}
		then, there exist a rate constant $\beta>0$ and   a positive order  $2s$, with $0<s\leq 1$,  such that the exponential moment is    generated to be bounded uniformly in time by the initial polynomial moment
		\begin{equation}\label{exp generation p}
		\mathcal{E}_{{\hg}}[f](\beta \min\left\{t,1\right\}, t) \leq \mathfrak{m}_{\ks}[f_0]= \int_{\mathbb{R}^{3+}} f_0(v,I)  \left\langle v, I \right\rangle^{2\ks} \mathrm{d} I \, \mathrm{d}v
		=M_G, \quad \forall t> 0.
		\end{equation}   
	\end{itemize}
\end{theorem} 
\begin{proof}

Let $f$ be the solution of the Cauchy problem \eqref{Cauchy}. 


The proof of both items  $(a)$ and $(b)$ in Theorem~\ref{theorem gen prop ML} strongly relies on estimates worked out in the previous sections to prove the  generation and propagation of polynomial moments stated in Theorem \ref{theorem bound on norm}. We will first present the sufficient estimates  and the  proof of each  item it performed in the next two subsections.

In order to obtain information about the summability of moments,  we need here {   to} use  the sharper upper bound of the  binomial expansion  \eqref{polynomial 1},   for any $0<s\leq 1$,  associated to the Lebesgue bracket
\begin{align*}\label{binomial 8}
&\left(\left\langle v, I \right\rangle^2 + \left\langle v_*, I_* \right\rangle^2 \right)^{sk}  
\leq \left(\left\langle v, I \right\rangle^{2s} + \left\langle v_*, I_* \right\rangle^{2 s} \right)^{k}\\
	&\ \leq \left\langle v, I \right\rangle^{2sk} +  \!\left\langle v_*, I_* \right\rangle^{2sk}  
	+ \!\sum_{\ell=1}^{\ell_{k}}\!\! \left( \begin{matrix}
k \\ \ell 
\end{matrix} \!\right)\!\! \left( \left\langle v, I \right\rangle^{2 s \ell} \left\langle v_*, I_* \right\rangle^{ 2(sk  \!- \! s \ell)}\!+ \left\langle v, I \right\rangle^{2( sk  \!-  \! s \ell)} \!\left\langle v_*, I_* \right\rangle^{2s \ell} \right)\!,
\end{align*}
with $\ell_k = \lfloor\frac{k + 1}{2}\rfloor$, which combined with estimate bounds from above \eqref{up tf} and below \eqref{lower bound lemma} for the transition function,  $\tilde{\mathcal{B}}= |v-v_*|^\gamma + \left(\frac{I +I_*}m\right)^{\hg}$, results in the following inequalities
%
for   polynomial moments
\begin{equation*}
\mathfrak{m}_{sk}[f](t)=: \mathfrak{m}_{sk}(t), \quad 0<s\leq \ho, \ k\geq 0,  \ \ \text{ with}\   sk>\ks,
\end{equation*}
with $\ks$ as defined in \eqref{k from moment bound},  where the only modification with respect to the proof of Lemma \ref{Lemma ODI} is  with the estimate on positive contribution.
Then, after applying coercive estimate    \eqref{lower bound lemma}, we need  the analog result as in  \eqref{mm pom} , now in a more accurate form suitable to study exponential  moments propagation and generation properties, namely 
\begin{equation}\label{poly ODI delta q}
\mathfrak{m}_{sk}'(t) \leq  - \Aks  \mathfrak{m}_{sk + \hg} +   \B_{sk} \sum_{\ell=1}^{\ell_{k}} \left( \begin{matrix}
k\\ \ell 
\end{matrix} \right) \left( \mathfrak{m}_{ s \ell + \hg} \mathfrak{m}_{sk -s \ell} + \mathfrak{m}_{sk -s \ell +\hg} \mathfrak{m}_{ s \ell } \right),
\end{equation}
where $\Aks$ and $\B_{sk}$ are  positive constants,
 \begin{align}\label{constants k}
\Aks & =   c_{lb} \tilde{A}_{\ks} = \clb \left(  \kappa^{lb} - \mathcal{C}_{k_*} \right),  \ \text{independent of}  \ sk>\ks, \ \text{and}, \nonumber\\ 
\B_{sk} &:= 2^{\frac{3 \gamma}{2}-1} \mathcal{C}_{sk}   \searrow 0, \ \text{for}\  sk >\ks,
\end{align}  
$ \mathcal{C}_{sk}\leq\kappa^{lb}$ {   for $sk >\ks$,} are from Lemma~\ref{lemma povzner}, estimate \eqref{povzner estimate} and 
  $\ks$ from \eqref{k from moment bound}. The  coercive constant {   satisfies  $\Aks > A_{\ks}$, with $A_{\ks}$}
  	associated   \eqref{Aks-Bk},
   $c_{lb}$ is the one from Lemma~\ref{lemma lower bound}, \eqref{clb}  and   $\tilde{A}_{\ks} $ is from \eqref{Ak tilde}.  Yet we note that $\B_{sk}$ is a much smaller constant than $B_{sk}$, as defined in   \eqref{Aks-Bk}, since the factor $ 2^{\frac{3 \gamma}{2}-1} <    2^{\frac{3\gamma+sk}2} \leq   \max\left\{ \left(\frac{\kappa^{lb} \, 2^{\frac{3 \gamma}{2} + \frac{k}{2}} }{ A_{k_*}} \right)^{\frac{\theta_{k}}{1-\theta_k}} \frac{\eta_{k}^{\frac{1}{1-\theta_k}}}{\left\|f\right\|_{L^1_1}},  1 \right\} \left\|f\right\|_{L^1_1} $, for any $sk>\ks$.
   
We are now in conditions to show both items in  Theorem~\ref{theorem gen prop ML} of this section.

\medskip
\subsection{\emph{End of proof  of Theorem~\ref{theorem gen prop ML}, part (a):  Propagation of exponential moments.}}    
Our goal is to find conditions for the  summability of moments propagation, starting from an initial data   $f_0\in L^1_k(\mathbb{R}^{3+}) $  finite  for all $k\ge\ks$,  that actually satisfies \eqref{initial data exp prop},  that is, there is a rate $0<\beta_0$, such that 
\begin{equation}\label{exp mom def in}
\mathcal{E}_{s}[f]({\beta},t) \mid_{t=0} =    \sum_{k=0}^\infty \frac{\beta_0^k}{k!} \, \mathfrak{m}_{s k}[f_0]\quad\text{is finite}\, .
\end{equation}

 Next, the calculation {   of} an  ODI associated to $\mathcal{E}_{s}[f]({\beta},t)$ is obtained as follows.
  
 Starting  from the Taylor series of an exponential function, one can represent  exponential moment  as  presented in \eqref{exp mom def 2}, set
\begin{equation}\label{exp mom def}
\mathcal{E}_{s}[f]({\beta},t) =  \sum_{k=0}^\infty \frac{\beta^k}{k!} \, \mathfrak{m}_{s k}[f](t).
\end{equation}
We consider its partial sum given above in \eqref{exp mom def},  and another partial sum  with a shift in the moment order $\hg$, namely,
\begin{equation}\label{pomocna 18}
\mathcal{E}_{s}^n = \sum_{k=0}^n \frac{\beta^k}{k!} \, \mathfrak{m}_{s k}, \quad \mathcal{E}_{s;\gamma}^n= \sum_{k=0}^n \frac{\beta^k}{k!} \, \mathfrak{m}_{s k+ \hg},
\end{equation}
where we have omitted to highlight   dependence on $t$ and $\beta$, and relation to $f$, i.e. we have assumed
\begin{equation*}
\mathcal{E}_{s}^n[f]({\beta},t)=:\mathcal{E}_{s}^n, \qquad \mathcal{E}_{s;\gamma}^n[f]({\beta},t):=\mathcal{E}_{s;\gamma}^n,  \qquad  \mathfrak{m}_{s k+ \hg}[f](t)=:\mathfrak{m}_{s k+ \hg}. 
\end{equation*}

The goal if to   prove that there is a $\beta$ independent of time $t$ such  partial sum $\mathcal{E}_{s}^n$ is bounded uniformly in time $t$ and $n$
such that, $\lim_{n\to\infty}\mathcal{E}_{s}^n(\beta, t) =  \mathcal{E}_{s}(\beta, t)$ . 

Taking derivative with respect to time $t$ of \eqref{exp mom def}, we get 
\begin{equation}\label{exp mom odi 1}
\frac{\mathrm{d} \, }{\mathrm{d} t}\mathcal{E}_{s}^n = \sum_{k=0}^n \frac{\beta^{ k}}{k!}   \mathfrak{m}_{s k}' 
= \sum_{k=0}^{k_0-1} \frac{\beta^{ k}}{k!}  \mathfrak{m}_{s k}' +  \sum_{k=k_0}^n \frac{\beta^{ k}}{k!}  \mathfrak{m}_{s k}',
\end{equation}
where $k_0$ is an index that will be determined later on. Since $\mathcal{E}_{s}^n$  is written in terms of $\mathfrak{m}_{sk}$ we  derive ordinary differential inequality (ODI) for polynomial moment $\mathfrak{m}_{sk}$. 

Indeed, from polynomial ODI, also  \eqref{poly ODI delta q}, and therefore  
	\begin{equation}\label{poly ODI prop}
\frac{\mathrm{d}}{\mathrm{d} t} \mathfrak{m}_{s k} \leq - \Aks \mathfrak{m}_{s k+ \hg} +   \B_{s k} \sum_{\ell=1}^{\ell_{k}} \left( \begin{matrix}
k\\ \ell 
\end{matrix} \right) \left( \mathfrak{m}_{s \ell+\hg} \ \mathfrak{m}_{s k- s \ell} + \mathfrak{m}_{s k- s \ell+\hg} \ \mathfrak{m}_{s \ell}\right).
\end{equation}
Making use  inequality \eqref{poly ODI prop} on the second term in \eqref{exp mom odi 1}, yields the new ODI
\begin{multline}\label{pomocna 10}
\frac{\mathrm{d} \, }{\mathrm{d} t}\mathcal{E}_{s}^n[f]   =    \mathcal{E}_{s}^n[Q(f,f)] 
\leq     \sum_{k=0}^{k_0-1} \frac{\beta^{k}}{k!}  \mathfrak{m}_{s k}'  - \Aks  \sum_{k=k_0}^n \frac{\beta^{ k}}{k!} \,   \mathfrak{m}_{s k+ \hg} \\ +   \sum_{k=k_0}^n \B_{sk} \frac{\beta^{ k}}{k!}   \sum_{\ell=1}^{\ell_{k}} \left( \begin{matrix}
k \\ \ell 
\end{matrix} \right) \left( \mathfrak{m}_{s \ell+\hg}  \mathfrak{m}_{s k- s \ell} + \mathfrak{m}_{s k- s \ell+\hg} \mathfrak{m}_{s \ell}\right)\\
=: S_0 - \Aks S_1 +   S_2.
\end{multline}
We estimate each sum $S_0$, $S_1$ and $S_2$ separately.

We first recall the  $k$-polynomial moments global bounds  in the case   of propagation of the initial data from Theorem~\ref{theorem bound on norm}, \eqref{poly propagation},  noting that the  bounds  for $\mathfrak{m}_{sk}'(t)$ easily follow from taking an upper bound to the upper ODE \eqref{upper ODE} to make the comparison argument to obtain  global bounds for  $\mathfrak{m}_{sk}(t)$, namely, 
\begin{align*}
\mathfrak{m}_{sk}  &\leq    \max\left\{\left(\frac{{B}_{    sk}}{{A}_{\ks}}\right)^{\frac{2 {    sk}}{{\gamma}}}\!\!, \mathfrak{m}_{    sk}[f](0) \right\},  \quad \text{and} \\ \mathfrak{m}_{sk}'  &\leq  B_{sk}    \max\left\{\left(\frac{{B}_{    sk}}{{A}_{\ks}}\right)^{\frac{2{    sk}}{{\gamma}}}\!\!, \mathfrak{m}_{    sk}[f](0) \right\}.
\end{align*}
It follows that,  for  any  fixed ${    s}k_0>\ks$,  the following {   constant is} independent of time 
\begin{equation}\label{pomocna 11-2}
{c}_{sk_0} :=    \max_{0\le k\le k_0 {    -1 }} \left\{\mathfrak{m}_{sk}, \mathfrak{m}_{sk}'  \right\}, \quad 0<s\le1\, . 
\end{equation}
{  
Note that this constant ${c}_{sk_0} $ monotonically increases with respect to  $k_0$, for $sk_0>\ks$, by the same property of moments $\mathfrak{m}_{sk}$ in $k$. Therefore,
\begin{equation}\label{csk0 monot}
c_{sk_0} \leq c_{s(k_0 +1)} \leq c_{sk_0 +1}.
\end{equation}
}

\noindent{\bf  Estimate for the term $S_0$ from \eqref{pomocna 10}. }  The first term  of \eqref{pomocna 10} is estimated using 
 the constant $c_{sk_0}$ just  defined in \eqref{pomocna 11-2}, to obtain an upper estimate to both 
\begin{equation}\label{pomocna 9}
\mathfrak{m}_{s k}, \mathfrak{m}_{s k}' \leq c_{sk_0} \quad \text{for all} \ k\in\{0,1,\dots,k_0 {    -1 } \}, \ \ \text{for}\  sk_0\geq \ks. 
\end{equation} 
Hence, an upper bound for  $S_0$  is controlled, with a good choice of the rate $\beta$, by
\begin{equation}\label{S0 estimate}
S_0 \leq  \sum_{k=0}^{k_0-1} \frac{\beta^{k}}{k!}  \mathfrak{m}_{s k}' {   \leq c_{sk_0} }  \sum_{k=0}^{k_0-1} \frac{\beta^{k}}{k!}  \leq c_{sk_0} e^{\beta} \leq 2 \, c_{sk_0},
\end{equation}
for some  $k_0$ fixed to be chosen later, provided that  $\beta$ is small enough to satisfy
\begin{equation}\label{beta first cond}
e^{\beta} \leq 2, \quad\text{or equivalently,}\quad    \beta< \ln 2.
\end{equation}

\noindent{\bf  Estimate for the term $S_1$ from \eqref{pomocna 10}. }   The   term  containing $S_1$ is negative and so it needs to be bounded from below.  Thus, recasting  the partial sum $\mathcal{E}_{s;\hg}^n$ minus the term {   similar to} $S_0$ by
\begin{equation*}
S_1 := \sum_{k=k_0}^n \frac{\beta^{k}}{k!} \,  \mathfrak{m}_{s k+ \hg} = \mathcal{E}_{s;\gamma}^n - \sum_{k=0}^{k_0-1} \frac{\beta^{ k}}{k!} \,\mathfrak{m}_{s k+ \hg}, 
\end{equation*}
and invoking the upper estimate to {   the second addend } in terms of the moments upper bounds   \eqref{pomocna 11-3}, {   namely
\begin{equation*}
\mathfrak{m}_{s k+ \hg} \leq \mathfrak{m}_{s k+ 1} \leq c_{sk_0 +1}, \quad \text{for all} \quad k = 1, \dots k_0-1,
\end{equation*}
}
yields the following lower  bound  to the term $S_1$ in  \eqref{pomocna 10}, as follows
\begin{equation}\label{S1 estimate}
S_1 \geq \mathcal{E}_{s;\gamma}^n - 2 c_{sk_0     +1}.
\end{equation}

\noindent{\bf  Estimate for the term $S_2$ from \eqref{pomocna 10}. } This next term in  \eqref{pomocna 10} can be split into two parts 
\begin{equation*}
S_2 = \sum_{k=k_0}^n \B_{sk}  \frac{\beta^{ k}}{k!}   \sum_{\ell=1}^{\ell_{k}} \left( \begin{matrix}
k \\ \ell 
\end{matrix} \right) \left( \mathfrak{m}_{s \ell+\hg}  \mathfrak{m}_{s k- s \ell} + \mathfrak{m}_{s k- s \ell+\hg} \mathfrak{m}_{s \ell}\right)=: S_{2_1} + S_{2_2}. 
\end{equation*}
that are worked out in the same way.  The sum $S_{2_1}$ is rearrenged by 
\begin{equation*}
S_{2_1} =  \sum_{k=k_0}^n \B_{sk}  \sum_{\ell=1}^{\ell_{k}}\frac{\beta^\ell \, \mathfrak{m}_{s \ell+\hg}}{\ell!}  \frac{\beta^{k-\ell} \, \mathfrak{m}_{s k- s \ell}}{(k-\ell)!} \leq \B_{sk_0} \mathcal{E}_{s;\gamma}^n \, \mathcal{E}_{s}^n,
\end{equation*}
where the last inequality is due to the  monotone decreasing property of   $\B_{sk}\leq \B_{sk_0}$  as defined in \eqref{constants k}.\\

Proceeding in a similar way,  $S_{2_2}$  can be estimated by 
\begin{equation}\label{S2 estimate}
S_{2_2} \leq  2 \, \B_{s k_0} \, \mathcal{E}_{s;\gamma}^n \, \mathcal{E}_{s}^n.\\
\end{equation}

Finally, inserting all the estimates \eqref{S0 estimate} {   with \eqref{csk0 monot}}, \eqref{S1 estimate} and \eqref{S2 estimate} into the right hand side of inequality 
 \eqref{pomocna 10},   the following upper ODI for the partial sum $\mathcal{E}_{s}^n$ follows
\begin{equation}\label{pomocna 12}
\frac{\mathrm{d} \, }{\mathrm{d} t}\mathcal{E}_{s}^n \leq  - \Aks \mathcal{E}_{s;\gamma}^n
+ 2 c_{sk_0     +1} (1+\Aks)    + 2 \, \B_{s k_0} \, \mathcal{E}_{s;\gamma}^n \, \mathcal{E}_{s}^n.
\end{equation}

The next step consists in finding a positive rate constant rate $\beta$ and an upper bound for $\mathcal{E}_{s}^n$ from  this ODI \eqref{pomocna 12}.
To this end, for each  $n \in \mathbb{N}$,  define
\begin{equation}\label{Tn prop}
T_n := \sup\{ t\geq 0: \mathcal{E}_{s}^n(\beta,\tau) \leq  3\, M_P, \ \forall \tau \in [0,t] \}, 
\end{equation}
where $M_P$ bounds the initial exponential moment of order $2s$ and rate $\beta_0$, from \eqref{initial data exp prop}. 

Since we already imposed the condition \eqref{beta first cond}, our task is to  show that we can find a $\beta \leq \min\{\beta_0, \ln 2\}$, such that  $\mathcal{E}_{s}^n(\beta,t)$ is uniformly bounded in $t$ and $n$,   by proving that  the time  for which   partial sums  remain bounded, as defined in 
\eqref{Tn prop},  is  actually unbounded. That means  $T_n=\infty$ for all $n\in \mathbb{N}$.\\

Indeed,  since $0< \beta \leq \beta_0$,  then assumption by \eqref{initial data exp prop} at at $t=0$,  yields 
\begin{equation}\label{Tn bound constant}
\mathcal{E}_{s}^n(\beta,0) = \sum_{k=0}^n \frac{\beta^{ k}}{k!} \mathfrak{m}_{s k}(0) \leq  \sum_{k=0}^\infty \frac{\beta_0^{ k}}{k!} \mathfrak{m}_{s k}(0) = \mathcal{E}_{s}({\beta_0},0)= M_P,
\end{equation}
uniformly in $n$. Thus, as each term $\mathfrak{m}_{s k}(t)$ is continuous in $t$, then $\mathcal{E}_{s}^n(\beta, t)$ is continuous in $t$ as well. Therefore, $\mathcal{E}_{s}^n(\beta, t)< M_P$ on some time interval $[0,t_n)$, with $t_n>0$, which implies that the sequence $T_n$ is well-defined and positive, for every $n\in \mathbb{N}$.

In addition, from the definition \eqref{Tn prop} of $T_n$,  it follows that $\mathcal{E}_{s}^n(\beta, t)\leq 3\, M_P$, for $t\in [0,T_n]$, so right hand side of the  ODI   \eqref{pomocna 12} is control from above by
\begin{equation}\label{pomocna 20}
\frac{\mathrm{d} \, }{\mathrm{d} t}\mathcal{E}_{s}^n \leq  -  \mathcal{E}_{s;\gamma}^n \left( \Aks - 8 \, \B_{s k_0} \, M_P  \right) \\ +  2 c_{sk_0     +1} \left(  1  + \Aks \right).
\end{equation}
 Since, from  \eqref{pomocna 11-2}, $\B_{s k_0}$ converges to zero as $k_0$ goes to infinity, allow us to conclude that 
  there is a  sufficiently large, fixed $k_0 >  \frac{ \ks}{s }$, such that  
\begin{equation}\label{K_1kstar}
\Aks - 8\, \B_{s k_0}  M_P    > \frac{\Aks}{2}.
\end{equation} 

Therefore,  the ODI in \eqref{pomocna 20} is estimated by the following upper ODI  for  the partial sums $\mathcal{E}_{s}^n(t)$
\begin{equation}\label{pomocna 21}
\frac{\mathrm{d} \, }{\mathrm{d} t}\mathcal{E}_{s}^n \leq  -\frac{\Aks}{2}\,  \mathcal{E}_{s;\gamma}^n+  2 { c_{sk_0     +1}} \left(  1  +\Aks \right), \qquad\text{for}\ \  sk_0 > {\ks}, 
\end{equation}
 with $k_0$ defined in \eqref{K_1kstar}.   

It remains to find a lower bound for the $\hg$ shifted partial sum $\mathcal{E}_{s;\gamma}^n$ in terms of $\mathcal{E}_{s}^n$,  for a  suitable rate $\beta$ to be chosen. We start by 
  estimating 
%
by recalling that the conserved quantity of the Boltzmann flow satisfies $ \mathfrak{m}_{0}(t)= \mathfrak{m}_{0}(0) < \mathcal{E}_{s}^n(\beta_0, 0)$. Second, we  note that the rate $\beta$, to be chosen soon,  needs to satisfy   from \eqref{beta first cond}  $\beta<\ln2$. That means, since $0<s\le 1$, then  $e^{\beta^{1-s}}< 2^{1-s} < 2 $.

Therefore, these two observations pave the way  to find a bound from below to the partial sum $\mathcal{E}_{s;\gamma}^n$,   for any $n>  s\ks$, as follows
\begin{multline}\label{ES lower bound}
\mathcal{E}_{s;\gamma}^n =  \sum_{k=0}^n \frac{\beta^{ k}}{k!} \mathfrak{m}_{s k+ \hg} \geq \sum_{k=0}^n \frac{\beta^{ k}}{k!} \int_{\left\{ \left\langle v, I \right\rangle \geq \beta^{-1/2} \right\}} f(t,v,I) \left\langle v, I \right\rangle^{2(sk+\hg)} \mathrm{d}I  \, \mathrm{d}v 
\\ \geq  \beta^{-\gamma/2} \left( \mathcal{E}_{s}^n - \sum_{k=0}^n \frac{\beta^{ k}}{k!}  \int_{\left\{ \left\langle v, I \right\rangle < \beta^{-1/2} \right\}} f(t,v,I) \left\langle v, I \right\rangle^{2sk}  \mathrm{d}I  \,  \mathrm{d}v  \right)
\\ \geq  \beta^{-\gamma/2} \left( \mathcal{E}_{s}^n - \sum_{k=0}^n \frac{\beta_0^{k (1-s)}}{k!} \mathfrak{m}_{0}(t)  \right) 
\geq  \beta^{-\gamma/2} \left( \mathcal{E}_{s}^n - \mathfrak{m}_{0} e^{\beta^{1-s}} \right)\\
  \geq  \beta^{-\gamma/2} \left( \mathcal{E}_{s}^n - 2  \mathcal{E}_{s}^n(\beta_0, 0) \right) = \beta^{-\gamma/2} \left( \mathcal{E}_{s}^n - 2  M_P \right)\, .
\end{multline}
Thus, using this lower estimate to control the negative term  from  \eqref{pomocna 21} yields the Ordinary Differential Inequality 
\begin{equation*}
\frac{\mathrm{d} \, }{\mathrm{d} t}\mathcal{E}_{s}^n \leq  -\frac{\Aks}{2}\, \beta^{-\gamma/2} \mathcal{E}_{s}^n +  {\Aks} \beta^{-\gamma/2} M_P +  2 c_{sk_0     +1} \left(  1  + \Aks \right),  
 \end{equation*}\label{ODE prop}
for $ sk_0 > {\ks},$ or equivalently, the following absorption ODI, 
\begin{equation}
\frac{\mathrm{d}}{\mathrm{d} t}\left(\mathcal{E}_{s}^n -  2M_P\right)    \leq  -\frac{\Aks}{2}\, \beta^{-\gamma/2}\left( \mathcal{E}_{s}^n - 2 M_P\right) +  2 c_{sk_0     +1} \left(  1  + \Aks \right),  \text{for}\ \  sk_0 > {\ks}.
 \end{equation}


Therefore, invoking the  maximum principle for ODI's majorized by an upper associated ODE to \eqref{ODE prop}, we obtain
\begin{align}\label{pomocna 22}
\mathcal{E}_{s}^n(\beta, t) -2M_P  &\leq \max\left\{ \left(\mathcal{E}_{s}^n(\beta_0, 0)-2M_P\right),  \frac{4 c_{    s k_0+1} \left(  1  + \Aks \right)}{ \Aks \, \beta^{-\gamma/2}}  \right\} \nonumber \\  
&\leq   \beta^{\gamma/2}  \, \frac{4 c_{sk_0     +1} \left(  1  + \Aks \right)}{ \Aks },
\end{align}
for any $t\in [0,T_n], $ since $\mathcal{E}_{s}^n(\beta_0, t) =M_p$, implying  the term $ \mathcal{E}_{s}^n(\beta_0, 0)-2M_P= - M_P=-\mathcal{E}_{s}^n(\beta_0, t) <0$.   In particular, inequality \eqref{pomocna 22} yields, by  taking  $\beta =\beta_1$ small enough  such that 
\begin{equation*}
\beta_1^{\gamma/2}   \,  \frac{4 c_{sk_0    +1} \left(  1  + \Aks \right)}{ \Aks \, } \ \le \ {M_P}
\end{equation*} 
or equivalently,
\begin{equation}\label{choose beta1}
\beta_1   \ \le\  \left(\frac{ \Aks  }{4 c_{sk_0     +1} \left(  1  + \Aks \right)}  {M_P} \right)^{\gamma/2} 
\end{equation}

Therefore setting  $\beta$ to be 
 the minimum satisfying  conditions  \eqref{initial data exp prop},  \eqref{beta first cond}, \eqref{K_1kstar} and \eqref{choose beta1}, namely, 
 \begin{equation}\label{beta  cond}
  \beta \le \min\{ \beta_0, \ln 2, \beta_1\}, 
 \end{equation}
the  inequality \eqref{ODE prop} implies that the following strict inequality of the partial sums up to term $n$ are strictly controlled by a $3M_P$,  or equivalently, 
\begin{equation}\label{pomocna 23}
\mathcal{E}_{s}^n(\beta, t) < 3 \mathcal{E}_{s}[f]({\beta_0},0), \quad\text{for any}\   0<\beta\leq  \min\{ \beta_0, \ln 2, \beta_1\},  \ t\in [0,T_n]\, , 
\end{equation}  
  with $\beta$ a constant independent of time $t$, depending on a fixed $k_0\ge\ks$  from \eqref{K_1kstar}. \\

Finally, due to the continuity of $\mathcal{E}_{s}^n(\beta,t)$ with respect to time $t$, this strict inequality actually  holds on a slightly larger time interval $[0, T_n + \varepsilon)$, $\varepsilon >0$. This contradicts the maximality of $T_n$ unless $T_n = + \infty$. Therefore,  
$\mathcal{E}_{s}^n(\beta, t) \leq 3 M_P$ for all $t\geq 0$ and $n\in \mathbb{N}$. Thus, letting $n\rightarrow \infty$ we conclude
\begin{equation*}
\mathcal{E}_{s}[f]({\beta},t) = \lim\limits_{n\rightarrow \infty} \mathcal{E}_{s}^n[f]({\beta},t) \leq 3 M_P =  3  \int_{\mathbb{R}^{3+}} f_0(v,I) \, e^{\beta_0 \left\langle v, I \right\rangle^{2s}} \mathrm{d} I \, \mathrm{d}v, \quad \forall t\geq 0,
\end{equation*}
i.e. the solution $f(t,v)$ to   Boltzmann equation with finite initial exponential moment of order $2s$ and propagates the exponential moments of order $2s$, but at a lower  rate $\beta$ than the initial  rate $\beta_0$ ,  with $\beta$ defined by 
\begin{equation}\label{pomocna 23-3}
\beta =\min \left\{ \beta_0,\,  \ln 2, \,\left(\frac{ \Aks  }{4 c_{sk_0   +1} \left(  1  + \Aks \right)}  \, {\mathcal{E}_{s}({\beta_0},0)} \right)^{\gamma/2} \right\}.
\end{equation}

It is significant  to notice that the  exponential moment rate $\beta$ is proportional to the coercive constant $\Aks = c_{lb} \tilde A_{\ks} $ as described  in \eqref{constants k} at the top of this last section, with  $c_{lb} $ from the functional  lower estimate calculated  in  Lemma~\ref{lemma lower bound}, estimate \eqref{clb}, and  positive factor  $ \tilde A_{\ks} $ defined in \eqref{Ak tilde}, estimating the  moments negative contribution derived {   from} the compact manifold averaging Lemma~\ref{lemma povzner}  crucial for  moments estimates of the gain collisional operator, with $\ks$  defined in \eqref{k from moment bound}. \\

\bigskip

\subsection{\noindent\emph{End of proof  of Theorem~\ref{theorem gen prop ML}, part (b):  Generation of exponential moments} } 

  The proof is more delicate since we seek to {   derive} an ODI for   the exponential moment $\mathcal{E}_{s}[f]({\beta},t)$ where the order factor $s$ and rate $\beta$ need to be found,  
just from an initial data    $f_0\in L^1_{\ks}(\mathbb{R}^{3+})$, that is $\int_{\mathbb{R}^{3+}} f_0(v,I) \, \left\langle v, I \right\rangle^{2\ks} \mathrm{d} I \, \mathrm{d}v =: M_G $              as defined in \eqref{initial data exp gen}.

	We will show that  the generated exponential moment associated to this  polynomial moment initial data has order $\gamma   \in (0,2]$, 
	and  a rate $\beta$ that  actually is  proportional to time $t$, as it provides a continuous  in time  venue to pass from a  data with a  $\ks$-Lesbegue polynomial moment to be able to show the instantaneously in time generation of an exponential moment. This calculation  clearly needs   the  developed estimates on the generation of moments Theorem~\ref{theorem bound on norm},   \eqref{poly gen max t-2}.

Following the ideas developed in  \cite{Gamba13}, \cite{GambaTask18},  and more recently in \cite{Alonso-IG-BAMS} and \cite{IG-P-C},
we start by  associating   an exponential moment of order $\gamma $, with a  rate linearly dependent on time, namely $\beta t$, with $\beta$ depending on $\ks$ from \eqref{k from moment bound}, to the solution $f(t,v,I) \in \Omega$ of the Boltzmann equation in the invariant region defined in \eqref{inv region},
\begin{equation}\label{exp mom gen}
\mathcal{E}_{\hg}[f](\beta t,t) := \int_{\mathbb{R}^{3+}} f(t,v,I) \, e^{\beta t \left\langle v, I \right\rangle^{\gamma}} \mathrm{d}I \mathrm{d}v =\sum_{k=0}^\infty \frac{(\beta t)^k}{k!} \mathfrak{m}_{\gamma k/2}[f](t).
\end{equation}

As in part  $(a)$,  we  also define partial sums and the associated  shifted ones now with time dependent rate, that is,
\begin{equation}\label{hg partial sums}
\!\!\!\!\mathcal{E}^n_{\hg}[f](\beta t,t)\!=\! \! \sum_{k=0}^n \frac{(\beta t)^k}{k!} \mathfrak{m}_{\gamma k/2}[f](t) \ \text{and} \ \mathcal{E}^n_{\hg; \gamma}[f](\beta t,t)\!=\! \! \sum_{k=0}^n \frac{(\beta t)^k}{k!} \mathfrak{m}_{\gamma k/2+ \hg  }[f](t).
\end{equation} 
Proceeding in a similar way as in the  previous Subsection for the propagation of exponential moments proof, we will also relieve notation by omitting explicit dependence on time $t$ and relation to $f$ by setting
\begin{equation*}
\mathcal{E}^n_{\hg}[f](\beta t,t)=: \mathcal{E}^n_{\hg}, \quad\text{and}\quad  \mathcal{E}^n_{\hg; \gamma}[f](\beta t,t):= \mathcal{E}^n_{\hg; \gamma}.
\end{equation*}

Next, taking  the time derivative of $ \mathcal{E}^n_{\hg}$ yields the identity
\begin{equation}\label{exp mom eq}
\frac{\mathrm{d}}{\mathrm{d} t} \mathcal{E}^n_{\hg} = \beta \sum_{k=1}^n \frac{(\beta t)^{k-1}}{(k-1)!} \mathfrak{m}_{\gamma k/2} + \sum_{k=0}^{k_0-1}  \frac{(\beta t)^{k}}{k!}  \mathfrak{m}_{\gamma k/2}' +  \sum_{k=k_0}^n \frac{(\beta t)^{k}}{k!}   \mathfrak{m}_{\gamma k/2}', 
\end{equation}
which, making use of the ODIs for $ \mathfrak{m}_{\gamma k/2}'$   \eqref{poly ODI delta q}, evaluated in   $s:=\hg$, results into
\begin{multline}\label{poly ODI gen}
\frac{\mathrm{d}}{\mathrm{d} t} \mathfrak{m}_{\gamma k/2}\leq - \Aks \mathfrak{m}_{\gamma k/2+ \hg  } \\+   \B_{\gamma k} \sum_{\ell=1}^{\ell_{k}} \left( \begin{matrix}
k\\ \ell 
\end{matrix} \right) \left( \mathfrak{m}_{\gamma \ell/2+\hg} \ \mathfrak{m}_{\gamma k/2- \gamma \ell/2} + \mathfrak{m}_{\gamma k/2- \gamma \ell/2+\hg} \ \mathfrak{m}_{\gamma \ell/2} \right), 
\end{multline}
that allows  identity  \eqref{exp mom eq} to be estimated by above as follows. 

The first and second terms are   obtained  by simply re-indexing  the sum and using the  definition of the corresponding  shifted partial sum;   and the last one is control by above  by the right hand side of the moment ODI \eqref{poly ODI gen}, which together implies
\begin{multline}\label{pomocna 11}
\frac{\mathrm{d}}{\mathrm{d} t} \mathcal{E}^n_{\hg} \leq \beta \,  \mathcal{E}^n_{\hg; \gamma} + \sum_{k=0}^{k_0-1}  \frac{(\beta t)^{k}}{k!}  \mathfrak{m}_{\gamma k/2}' - \Aks  \sum_{k=k_0}^n \frac{(\beta t)^{k}}{k!}  
\mathfrak{m}_{\gamma k/2+ \hg  }\\+    \sum_{k=k_0}^n \frac{(\beta t)^{k}}{k!} \B_{\frac{\gamma k}{2}}  \sum_{\ell=1}^{\ell_{k}} \left( \begin{matrix}
k\\ \ell 
\end{matrix} \right) \left( \mathfrak{m}_{\gamma \ell/2+\hg}  \mathfrak{m}_{\gamma k/2- \gamma \ell/2} + \mathfrak{m}_{\gamma k/2- \gamma \ell/2+\hg} \mathfrak{m}_{\gamma \ell/2}\right)
\\ =: \beta \,  \mathcal{E}^n_{\hg; \gamma} + S_0 - \Aks S_1 +  \left(S_{2_1} + S_{2_2} \right),
\end{multline}
 for  any $k_0\ge \ks$, since $0 <\gamma/2\le 1$.   

We treat each of this terms term separately.

Clearly, as we wrote in part $(a)$ of the proof of Theorem~\ref{theorem gen prop ML},  we need to invoke the global in time estimates for moments of $ \mathfrak{m}_{\gamma k/2}(t)$
under initial conditions  for the generation property as  defined in 
\eqref{poly gen max t-2}.

 Hence, set 
\begin{equation*}
\mathfrak{m}_{    k \gamma/2}  \leq \mathfrak{B}_{    k \gamma/2} \ \max_{t>0}\{ 1, t^{-{k}}\},  \quad  \mathfrak{m}_{    k \gamma/2}'  \leq  B_{    k \gamma/2} \mathfrak{B}_{    k \gamma/2}\ \max_{t>0}\{ 1, t^{-{k}}\}, 
\end{equation*}
and  denote, for any fixed  finite ${    \gamma k_0/2} \geq \ks$, the  positive constant 
\begin{equation}\label{pomocna 11-3}
 \bar c_{\gamma k_0/2}:=   \max_{0\leq k\leq k_0-1}\{\mathfrak{B}_{\gamma k/2},     B_{ k \gamma/2} \mathfrak{B}_{ k \gamma/2} \}\, ,
\end{equation}
where $\mathfrak{B}_{    k \gamma/2} $  was defined in   \eqref{poly gen max t-2}

It is {   worth} to notice that  these  moments upper bounds are time dependent. This is a  crucial feature that enables the  generation of  exponential moments with polynomial moment initial data \eqref{initial data exp gen}  in a short interval of time, and then bootstrap the argument to get global in time estimates.

\medskip

\noindent {\bf Estimate for the  term $S_0$ from \eqref{pomocna 11}.}  From polynomial moment generation estimates \eqref{poly gen max t-2}, one can bound polynomial moments of any order  $k>\ks$, as well as their derivatives by means of \eqref{ODI poly} with bounds  $\mathfrak{B}^{k}$. In particular, set the moments upper bounds to be
\begin{equation*}
\mathfrak{m}_{\gamma k/2}  \leq \mathfrak{B}^{\gamma k/2} \ \max_{t>0}\{ 1, t^{-{k}}\},  \quad  \mathfrak{m}_{\gamma k/2}'  \leq  {    B}_{\gamma k/2} \mathfrak{B}^{\gamma k/2} \ \max_{t>0}\{ 1, t^{-{k}}\}, 
\end{equation*}
and in particular 
\begin{equation}\label{pomocna 11-4}
\mathfrak{m}_{\gamma k/2}  \leq \mathfrak{B}_{\gamma k/2}  t^{-{k}},  \quad\text{and}\quad   \mathfrak{m}_{\gamma k/2}'  \leq  {    B}_{\gamma k/2} \mathfrak{B}_{\gamma k/2}  t^{-{k}},  \qquad\text{for} \ 0< t\le 1. 
\end{equation}

%
%
%
%
%


Then,  for any $ t\in(0,1]$ and $k\le k_0     -1$, and estimating $S_0$ by 

\begin{equation}\label{S0 gen}
S_0:= \sum_{k=0}^{k_0-1}  \frac{(\beta t)^{k}}{k!}  \mathfrak{m}_{\gamma k/2}'(t) \leq  \sum_{k=0}^{k_0-1}  \frac{\beta^{k}}{k!} \mathfrak{B}_{\gamma k/2}   \leq \bar c_{\gamma k_0/2} e^{\beta} \leq 2 \bar c_{\gamma k_0/2},
\end{equation}
with  $\bar c_{\gamma k_0/2}$ defined in \eqref{pomocna 11-3} constant in time,
for  $ t\in(0, 1]$    and $ {    \gamma k_0/2} \ge \ks$, which holds provided that $\beta$  satisfies, at most
\begin{equation}\label{gen alpha 1}
\beta \leq  \ln 2
\end{equation}
exactly as it was required in the proof of  Theorem~\ref{theorem gen prop ML}, part $(a)$, \eqref{beta first cond}.


\noindent {\bf Estimate for the  term $S_1$ from \eqref{pomocna 11}.}  This estimate is crucial  for the control from below of the negative contribution associated to the ODI for    exponential moments  that yields  their generation  argument,   with initial data depending on just $f_0(v,I) \in L^1_{\ks}(\mathbb{R}^{3+})$.  

Indeed,   using the boundedness of the $\gamma/2$-shifted moments, namely  $\mathfrak{m}_{\gamma k/2+ \hg  } \leq \mathfrak{B}^{\gamma k/2+\hg} \ \max_{t>0}\{ 1, t^{-k-1}\}$, and that $\beta$ has been chosen to satisfy  $e^{\beta} <2$, then 
\begin{align}\label{S1 gen}
S_1:= \sum_{k=k_0}^n \frac{(\beta t)^{k}}{k!} \mathfrak{m}_{\gamma k/2+ \hg  } &= \mathcal{E}^n_{\hg; \gamma} -  \sum_{k=0}^{ k_0-1 } \frac{(\beta t)^{k}}{k!} \mathfrak{m}_{\gamma k/2+ \hg  }  \nonumber \\
&\ \ \geq  \mathcal{E}^n_{\hg; \gamma} - 2 \bar c_{\gamma (k_0+1)/2}{\frac{1}{t} } ,
\end{align}
for $\beta$ chosen as in \eqref{gen alpha 1} and  $ \bar c_{\gamma (k_0+1)/2}$ as defined in \eqref{pomocna 11-3}.   

Note that the appearance of the factor $1/t$ multiplying $2 \bar c_{ \gamma(k_0 +1)/2} $ is a major deviation from the analog estimate of the term $S_1$ in the propagation argument of Theorem~\ref{theorem gen prop ML}, part $(b)$, as it will be use later on in this proof.
  
\noindent  {\bf Estimate for the  term $S_2$ from \eqref{pomocna 11}.} Terms $S_{2_1}$ and $S_{2_2}$ are treated in the same way as follows. We will detail calculation for $S_{2_1}$.  First,  reorganize the terms in sum to obtain
\begin{align}\label{S2 gen}
S_{2_1}&:=\sum_{k=k_0}^n \frac{(\beta t)^{k}}{k!}  \B_{\frac{\gamma k}{2}}  \sum_{\ell=1}^{\ell_{k}} \left( \begin{matrix}
k\\ \ell 
\end{matrix} \right)  \mathfrak{m}_{\gamma \ell/2+\hg}  \mathfrak{m}_{\gamma k/2- \gamma \ell/2} \\
&= \sum_{k=k_0 }^n \B_{\frac{\gamma k}{2}} \sum_{\ell=1}^{\ell_{k}} \frac{(\beta t)^{\ell}  \mathfrak{m}_{\gamma \ell/2+\hg}}{\ell!}   \frac{(\beta t)^{k-\ell}\mathfrak{m}_{\gamma k/2- \gamma \ell/2} }{(k-\ell)!} \leq \B_{\frac{\gamma     k_0}{2}}  \mathcal{E}^n_{\hg; \gamma} \, \mathcal{E}^n_{\hg}. \nonumber
\end{align}
where,  as defined in \eqref{constants k}, $\B_{\frac{\gamma k}{2}} = 2^{\frac{3 \gamma}{2}-1}{\mathcal{C}}_{\frac{\gamma k}2}  $ monotonically  decays with respect to $k  \geq k_0$,  which implies $\B_{\frac{\gamma k}{2}}  \leq \B_{\frac{\gamma k_0}{2}}$  for  $\frac{2}{\gamma}\ks< k_0 \leq k\le n$.  The  estimate for the term $S_{2_2}$ in obtain in a similar way.

Gathering and inserting   estimates  \eqref{S0 gen},  \eqref{S1 gen} and \eqref{S2 gen} into \eqref{pomocna 11},     with  $\beta$ satisfying \eqref{gen alpha 1}, that is $\beta<\ln2$, then all partial sums from \eqref{pomocna 11} satisfy the  following a priori ODIs, for any $\frac{2}{\gamma} \ks < k_0 \sout{\le k}\le n$, 
\begin{align}\label{pomocna 25}
\frac{\mathrm{d}}{\mathrm{d} t} \mathcal{E}^n_{\hg} &\leq \beta \,  \mathcal{E}^n_{\hg; \gamma} + 2 c_{\gamma k_0/2} - \Aks \left(  \mathcal{E}^n_{\hg; \gamma} - 2 \bar c_{ \gamma(k_0 +1)/2}  \frac{1}{t} \right) +   \B_{\frac{\gamma}{2}  { k_0}}  \, \mathcal{E}^n_{\hg; \gamma} \, \mathcal{E}^n_{\hg} 
\nonumber\\ 
&=\mathcal{E}^n_{\hg; \gamma} \left(\beta-\Aks +    \B_{\frac{\gamma     k_0}{2}}    \mathcal{E}^n_{\hg}    \right)   +  2 \bar c_{ \gamma(k_0 +1)/2} \left( \frac{t+  \Aks}{t} \right).
\end{align}

Now, fix $\gamma\in (0,2]$,  and  $\beta<\ln2$,  while  seting the time interval $[0, \bar{T}_n ]$  with the upper end given  by 
\begin{equation*}
\bar{T}_n := \sup\left\{  t\in [0,1]: \mathcal{E}^n_{\hg}[f](\beta t,t)\leq 4 M_G \right\}.
\end{equation*}

Note that this
$ \bar{T}_n $ is well defined. 
Indeed, taking the conserved quantity
\begin{equation}\label{M0 gen}
M_G:= \left\| f \right\|_{L^1_1} = \int_{\mathbb{R}^{3+}} f_0(v,I) \langle v, I\rangle^{    2} \, \mathrm dI \mathrm dv ,   \qquad { \text{with}\  f_0(v,I) \in \Omega,  } 
\end{equation}
and noting  that just for  $t=0$, $\mathcal{E}^n_{\hg}(0,0) \leq   \mathcal{E}_{\hg}(0,0)=\mathfrak{m}_{0}(0)<\left\| f \right\|_{L^1_1}\!(0)   < 4 M_G$, then  
 by the continuity of partial sums $\mathcal{E}^n_{\hg}$ with respect to $t$  imply that  $\mathcal{E}^n_{\hg}(\beta t,t) \leq 4 M_G$ on a slightly larger time interval $t\in[0,t_n)$, $t_n>0$, and thus $0 < \bar{T}_n     \leq 1$.

In the next step, we  need to show that there is a rate  $0< \beta< \ln 2$ such that   $ \mathcal{E}^n_{\hg}(t)$ is bounded for $t\in [0, \bar{T}_n ]$. 
It is clear that for any $t\in [0, \bar{T}_n ]$,  the evaluation in the partial sum  $\mathcal{E}^n_{\hg}(\beta t,t)\leq 4 M_G$,  since $\bar{T}_n\leq 1$ and so   $t^{-1}\geq 1$. Therefore,  from \eqref{pomocna 25}, the ODI for the partial sums with time dependent rates $\beta t$ yields  the following upper estimate for  the time derivative of $ \mathcal{E}^n_{\hg}(t)$, namely,  for any time $t$, such that  $\ \ 0\le t \le \bar{T}_n\le 1$, 
\begin{equation}\label{pomocna 25.b}
\frac{\mathrm{d}}{\mathrm{d} t} \mathcal{E}^n_{\hg}(t) \leq   - \mathcal{E}^n_{\hg; \gamma}(t) \left(  \! - \beta  +\Aks -\B_{\frac{\gamma}2  k_0} \, 4 M_G\right) +   2 \bar c_{ \gamma(k_0 +1)/2} \frac{1+  \Aks}{t} \, .
\end{equation}

Now, since  $ \B_{\frac{\gamma k_0}2 }  $ converges to zero as  $ {    \frac{2}{\gamma} }\ks< k_0$ grows, we can choose large enough number of moments $k_0$,  such that 
\begin{equation}\label{choice beta rate k00} 
\frac{\Aks} 2 >  4 \B_{\frac{\gamma}2  k_0}  M_G,
\end{equation}
with $M_G + =\|f\|_{L^1_1}$ from \eqref{M0 gen},     then just taking  a small enough $\beta$ satisfying 
\begin{equation}\label{choice beta rate k0} 
  0< \beta < \frac{\Aks}2 
\, ,
 \end{equation}
yields the following upper bound from \eqref{pomocna 25.b} 
\begin{equation}\label{K3}
\frac{\mathrm{d}}{\mathrm{d} t} \mathcal{E}^n_{\hg} \leq  - \frac{\Aks}{2} \mathcal{E}^n_{\hg; \gamma} +  2\bar c_{ \gamma(k_0 +1)/2}\left(\frac {1+  \Aks}{t} \right).
\end{equation}

Finally, the shifted moment  $\mathcal{E}^n_{\hg; \gamma}(\beta t,t) $   can be bounded as follows
\begin{equation*}
\mathcal{E}^n_{\hg; \gamma}(\beta t,t) = \sum_{k=1}^{n+1} \frac{ (\beta t)^k \mathfrak{m}_{\gamma k/2}(t) }{k!} \frac{k}{\beta t} \geq  \frac{1}{\beta t}  \sum_{k=2}^{n} \frac{ (\beta t)^k \mathfrak{m}_{\gamma k/2}(t) }{k!} \geq  \frac{\mathcal{E}^n_{\hg}(\beta t,t) - M_G}{\beta t},
\end{equation*}
hence, the ODI \eqref{pomocna 25.b} yields the reduced one
\begin{equation}\label{pomocna 25.c}
\frac{\mathrm{d}}{\mathrm{d} t} \mathcal{E}^n_{\hg} \leq  - \frac{\Aks}{2\beta t} \left( \mathcal{E}^n_{\hg} - M_G -  \frac{2\beta}{\Aks} 2\bar c_{ \gamma(k_0 +1)/2} (1+  \Aks) \right).
\end{equation}

Now,  recalling the value of $M_G$ from the initial data from \eqref{M0 gen}, and that $\hg<1$, and so $\bar c_{ \gamma(k_0 +1)/2}\leq \bar c_{k_0 +1}$,    one can choose $\beta$ small enough such that 
\begin{equation}\label{choice beta rate}
M_G +  4\beta \bar c_{ k_0+1} \left(\frac{1+  \Aks}{\Aks} \right)  < 2 \|f\|_{L^1_{\ks}}
\end{equation} 
 or equivalently, choosing $\beta$ the smallest,   between $\ln 2$,    \eqref{choice beta rate k0} and \eqref{choice beta rate}, by 
\begin{equation}\label{choice beta rate 2}
0< \beta\ <\ \min\left\{\ln 2; \frac{ \Aks }{4 \bar c_{ k_0+1} } \frac{\|f\|_{L^1_{\ks}} }{1+  \Aks }\right\},
\end{equation} 
for $k_0$ large enough to satisfy condition \eqref{choice beta rate k00} and the corresponding constant  $\bar c_{ k_0+1}$ from  \eqref{pomocna 11-3}
the upper bound for the right hand side in \eqref{pomocna 25.c} yields the absorption    ODI,
\begin{equation}\label{ODI absorbing}
\frac{\mathrm{d}}{\mathrm{d} t} \mathcal{E}^n_{\hg}(\beta t,t) \leq  - \frac{\Aks}{2\beta t} \left( \mathcal{E}^n_{\hg}(\beta t,t) - 2\|f\|_{L^1_{\ks}} \right).
\end{equation}
or, since  $\|f\|_{L^1_1}$ is constant in time, due to conservation of mass and energy, we can  equivalently write  for $X(t):= \mathcal{E}^n_{\hg}(\beta t,t) - 2\|f\|_{L^1_1}$, 
\begin{equation*}
\frac{\mathrm{d}}{\mathrm{d} t}  X(t)  \leq  - \frac{\Aks}{2\beta t} X(t), \quad \text{ that implies} \quad X(t)\le X(0) , \quad\forall t\in (0,\bar{T}_n],
\end{equation*}
by Gronwall inequality  or equivalently,  since $ \mathcal{E}^n_{\hg}(0,0) <  \|f\|_{L^1_{\ks}} $
 for $f\in \Omega$,
\begin{align}\label{ODI absorbing-2}
\mathcal{E}^n_{\hg}(\beta t,t)  - 2 \|f\|_{L^1_{\ks}}   &\le  \mathcal{E}^n_{\hg}(0,0)  - 2 \|f\|_{L^1_{\ks}}  \leq  - 1 \|f\|_{L^1_{\ks}} , \ \text{that implies}\nonumber\\[4pt]
 \mathcal{E}^n_{\hg}(\beta t,t)  &\leq \ \|f\|_{L^1_{\ks}} ,\qquad\forall t\in (0,\bar{T}_n].
\end{align}

By the  continuity of  partial sums $\mathcal{E}^n_{\hg}(\beta t,t)$ these inequalities hold on a slightly larger interval, which contradicts maximality of $\bar{T}_n$, unless $\bar{T}_n=1$. Therefore, we can conclude $\bar{T}_n=1$  uniformly in  $n \in \mathbb{N}$.

Hence, letting $n\rightarrow \infty$, we conclude 
\begin{equation}\label{pomocna 27}
\mathcal{E}_{\hg}(\beta t,t)   :=   \int_{\mathbb{R}^{3+}} f(t,v,I) \, e^{\beta t \left\langle v, I \right\rangle^{\gamma}} \mathrm{d}I \mathrm{d}v  \leq    \|f\|_{L^1_{\ks}(\mathbb{R}^{3+})}  \qquad \forall t \in [0,1], 
\end{equation}
and thus the exponential moment of the order $\gamma$,  and a rate  $0< \beta$, from \eqref{choice beta rate k0} and \eqref{choice beta rate},  
propagates in the time interval $(0,1]$, and stays uniformly bounded for all $t>1$ to obtain a global in time as the argument bootstraps on time intervals $[l, l+1]$, for $l\in \mathbb N^+$.

It is also {   very} significant   to notice in the generation of exponential moments that the  exponential moment rate $\beta$ is proportional to the coercive constant $\Aks$, and  inversely proportional to the moments generation bounds expressed in the constant $ \bar c_{ k_0+1}$ defined in \eqref{pomocna 11-3} with  ${    k_0 \gamma/2}>\ks$, $k_0$  large enough to satisfy   condition \eqref{choice beta rate k00}, The $\ks$-moment is characterized 
from \eqref{k from moment bound} and depends on the assumptions on the three different transition probability rates associated to the polyatomic gas model presented in this manuscript. 
These conditions,  while somehow  different in substance, {   they} are closely related to the coerciviness property of the Boltzmann collisional binary form for polyatomic gases,   as much as described  at the end of the previous Subsection of theorem~\ref{theorem gen prop ML}, part $(a)$.


\end{proof}

\section*{Acknowledgments}
The authors would like to thank Professor Thierry Magin for fruitful discussions on the topic.   Authors also  thank and gratefully acknowledge the hospitality and support from the Oden Institute of Computational Engineering and Sciences and the University of Texas Austin.
Irene M. Gamba was supported by the funding DMS-RNMS-1107291 (Ki-Net),  NSF DMS1715515 and  DOE DE-SC0016283 project \emph{Simulation Center for Runaway Electron Avoidance and Mitigation}.   
Milana Pav\'c-\v Coli\'c acknowledges  support of the Ministry of Education, Science and Technological Development of the Republic of Serbia (Grant No. 451-03-68/2020-14/ 200125),  and the support of the Science Fund of the Republic of Serbia, PROMIS, \#6066089, MaKiPol.

\appendix

\section{Proof of the Lemma \ref{Lemma coll mapping} (Jacobian of the collision transformation)}\label{app Jacobian}

\begin{proof}
	Using ideas from \cite{DesMonSalv}, we decompose the mapping $T$ from \eqref{coll mapping} into a sequence of mappings and calculate Jacobian of each of them. Then the Jacobian of $T$ will be a product of those Jacobians.   More precisely, $T$ can be understood  as a composition of the following transformations 
	\begin{equation*}
	T= T_9 \circ T_8 \circ T_7 \circ T_6 \circ T_5 \circ T_4 \circ T_3 \circ T_2 \circ T_1,
	\end{equation*}
	where composition is understood as $(f\circ g)(x) = f(g(x))$ and each $T_i$ is described below.
	\begin{itemize}
		\item[(1)]  We first pass to the center-of-mass reference frame 
		\begin{equation*}
		T_1:(v, v_*, I, I_*, r, R, \sigma) \mapsto (u, V,I,I_*,r,R, \sigma),
		\end{equation*}
		where $u$ and $V$ are relative velocity and velocity of center of mass from \eqref{cm-rv}. It is clear that Jacobian of this transformation is 1,
		\begin{equation*}
		J_{T_1}=1.
		\end{equation*}
		\item[(2)] For the relative velocity $u$ we pass to its spherical coordinates $\left(\left|u\right|, \frac{u}{\left|u \right|}\right)$, where $u/\left| u \right| \in S^2$ is the angular variable, with the transformation $T_2$,
		\begin{equation*}
		\left({u},{V},I,I_*,r,R,{\sigma}\right) \mapsto
		\left(\left|u\right|,\frac{{u}}{\left| u \right|},{V},I,I_*,r,R,{\sigma}\right),
		\end{equation*}
		whose Jacobian is
		$$J_{T_2}=\left| {u}\right| ^{-2}.$$
		\item[(3)] In order to facilitate further calculation, we consider square of relative speed instead of relative speed itself, 
		$$T_3: \left(\left|u\right|,\frac{u}{\left| u \right|}, V, I,I_*,r,R, \sigma \right) \mapsto \left(\left|u \right|^2 ,\frac{{u}}{\left| {u} \right|}, V,I,I_*,r,R, \sigma\right) $$
		with the Jacobian  
		$$J_{T_3}= 2 \left| {u}\right|.$$
		\item[(4)] Instead of $I_*$ we will use total energy $E$, 	linked with the equation \eqref{micro CL energy mass-rel vel},
		\begin{equation*}
		T_4:  \left(\left|u \right|^2 ,\frac{{u}}{\left| {u} \right|}, V,I,I_*,r,R, \sigma\right) \mapsto \left(\left|u \right|^2 ,\frac{{u}}{\left| {u} \right|}, V,I,E,r,R, \sigma\right)
		\end{equation*}
		whose Jacobian is 1,
		$$ J_{T_4}= 1.$$
		\item[(5)] Moreover, instead of $R$ we want to  have $ER$,
		\begin{equation*}
		T_5: \left(\left|u \right|^2 ,\frac{{u}}{\left| {u} \right|}, V,I,E,r,R, \sigma\right) \mapsto \left(\left|u \right|^2 ,\frac{{u}}{\left| {u} \right|}, V,I,E,r,E R, \sigma\right),
		\end{equation*}
		with Jacobian 
		$$ J_{T_5}= E.$$
		\item[(6)]  Finally, we pass to pre-collisional  quantities with the following mapping
		$$T_6: \left(\left|u \right|^2 ,\frac{{u}}{\left| {u} \right|}, V,I,E,r,ER, \sigma\right) \mapsto \left(\left|u' \right|^2 ,\frac{{u'}}{\left| u' \right|}, V',I',I'_*,r',R', \sigma'\right). $$
		Let us compute Jacobian of this central transformation. First, for $V$ we are using conservation law \eqref{CL V}. Change of the unit vectors  $\frac{u}{\left| u \right|}$ and $\sigma$ can be considered as a rotation. Thus we eliminate these variables and for the rest of variables, we use the following relations
		\begin{multline*}
		\left|u' \right|^2 = \frac{4 RE}{m}, \ I'=r(1-R)E, \  I'_*=(1-r)(1-R)E, \\ r'=\frac{I}{E-\frac{m}{4}\left|u\right|^2}, \ R'=\frac{m \left|u\right|^2}{4 E},
		\end{multline*} 
		and compute the corresponding Jacobian
		\begin{multline*}
		J_{\left(\left|u \right|^2 ,I,E,r,ER\right)  \mapsto \left(\left|u' \right|^2,I',I'_*,r',R'\right)}\\
		=\begin{vmatrix}
		0 & 0 & \frac{4 R}{m}& 0& \frac{4}{m} \\[5pt]
		0 &0&  r(1-R)& (1-R)E& -r\\[5pt]
		0 &0&  (1-r)(1-R)& -(1-R)E& 1-r\\[5pt]
		\frac{m I}{4(E-\frac{m}{4}\left|u\right|^2)^2} &\frac{1}{E-\frac{m}{4}\left|u\right|^2}&  -	\frac{ I}{(E-\frac{m}{4}\left|u\right|^2)^2} & 0 & 0\\[5pt]
		\frac{m}{4 E}& 0 & - \frac{ m \left|u\right|^2}{4 E^2} & 0&0\\[5pt]
		\end{vmatrix} \\[5pt]
		=\frac{(-1)^{4+2}}{E-\frac{m}{4}\left|u\right|^2} \frac{(-1)^{4+1}m}{4E} 
		\begin{vmatrix}
		\frac{4 R}{m}& 0& \frac{4}{m} \\[5pt]
		r(1-R)& (1-R)E& -r\\[5pt]
		(1-r)(1-R)& -(1-R)E& 1-r\\[5pt]
		\end{vmatrix} \\[5pt]
		=\frac{-m(1-R)E}{(E-\frac{m}{4} \left| u \right|^2)4E} 
		\begin{vmatrix}
		\frac{4 R}{m}& 0& \frac{4}{m} \\[5pt]
		r(1-R)& 1& -r\\[5pt]
		(1-r)(1-R)& -1& 1-r\\[5pt]
		\end{vmatrix} \\[5pt]
		=\frac{-m(1-R)}{4(E-\frac{m}{4}\left| u \right|^2)} 
		\begin{vmatrix}
		\frac{4 R}{m}& 0& \frac{4}{m} \\[5pt]
		1-R& 0& -1\\[5pt]
		(1-r)(1-R)& -1& 1-r\\[5pt]
		\end{vmatrix} \\[5pt]
		= \frac{-m(1-R)}{4\left(E-\frac{m}{4}\left| u \right|\right)}\left(\frac{-4 R}{m}+\frac{4(R-1)}{m}\right) \nonumber 
		\end{multline*}
		Finally,
		$$
		J_{T_6}=\frac{1-R}{(E-\frac{m}{4} \left|u\right|^2)}= \frac{1-R}{I+I_*}=\frac{1-R}{(1-R')E }.$$
		\item[(7)] Now we go back, first from squares to squares of relative speed to relative speed itself,
		$$T_7:  \left(\left|u' \right|^2 ,\frac{{u'}}{\left| u' \right|}, V',I',I'_*,r',R', \sigma'\right) \mapsto \left(\left|u' \right| ,\frac{{u'}}{\left| u' \right|}, V',I',I'_*,r',R', \sigma'\right). $$	
		with Jacobian
		$$J_{T_7}= \frac{1}{2 \left|u'\right|}. $$
		\item[(8)] For $u'$ we pass from spherical coordinates to  Cartesian ones,
		$$T_8:  \left(\left|u'\right|,\frac{{u'}}{\left| u' \right|}, V',I',I'_*,r',R', \sigma'\right) \mapsto \left(u', V',I',I'_*,r',R', \sigma'\right). $$	
		with Jacobian
		$$J_{T_8}= \left|u'\right|^2. $$
		\item[(9)] We go back from center-of-mass reference frame,  $$T_9:(u',V',I',I'_*,r',R',\sigma') \mapsto (v', v'_*,I',I'_*,r',R', \sigma') $$ with unit Jacobian
		$$J_{T_9}=1.$$
	\end{itemize}
	Finally, we get the Jacobian of transformation $T$,
	$$J_T = \prod_{i=1}^{9}J_{T_i}= \frac{ (1-R) \left|u'\right|  }{ (1-R') \left|u\right| }=  \frac{ (1-R) R^{1/2} }{ (1-R') R'^{1/2} },
	$$
	where for the last inequality  we have used $\left|u'\right| = \sqrt{\frac{4RE}{m}}$ and $\left|u\right| = \sqrt{\frac{4R'E}{m}}$.
\end{proof}

\section{Explicit calculation of multiplicative factors to the transition function models}\label{Sec: App bounds}
This appendix provides upper and lower estimates for the multiplicative factors $\dgl$, $\dgu$, $\egl$, $\egu$ for the three models of transition function ${\mathcal{B}} = \mathcal{B}(v, v_*, I, I_*, r,  R, \sigma)$ introduced in section~\ref{Sec: tf},  \eqref{model 1}, \eqref{model 2} and \eqref{model 3}, namely\\[5pt]
(Model 1) $\displaystyle \mathcal{B}  = b(\hat{u}\cdot\sigma) \left( \frac{m}{4} \left|u\right|^2 + I + I_* \right)^{\gamma/2}$,\\[5pt]
(Model 2) 	$ \displaystyle	\mathcal{B} =b(\hat{u}\cdot\sigma) \left(  R^{\gamma/2} |u|^\gamma + (1-R)^{\gamma/2} \left(\frac{I+I_*}{m}\right)^{\gamma/2} \right)$,\\[5pt]
(Model 3) $	 \displaystyle \mathcal{B}= b(\hat{u}\cdot\sigma) \left(  R^{\gamma/2} |u|^\gamma  + \left( r (1-R)\frac{I}{m} \right)^{\gamma/2} + \left( (1-r) (1-R)\frac{I_*}{m} \right)^{\gamma/2} \right)$,\\[5pt]
where $u:=v-v_*$ and $\hat{u}=u/\left| u \right|$. 

\subsection{Calculation of the upper bounds} For the three models  \eqref{model 1}, \eqref{model 2} and \eqref{model 3}, we   determine functions  $\dgu$ and $\egu$ that appear in bound from above for the   transition function $\mathcal{B}$, as defined in  \eqref{trans prob rate ass}. 

\subsubsection{Model 1}\label{Sec Model 1 upper}
We write for the model 1, from \eqref{model 1},
\begin{multline*}
\left( \frac{m}{4} \left|v-v_*\right|^2 + I + I_* \right)^{\gamma/2} = m^{\gamma/2}\left( \frac{1}{4} \left|v-v_*\right|^2 + \frac{I}{m}  + \frac{I_*}{m}  \right)^{\gamma/2}
\\
\leq m^{\gamma/2}\left( \left|v-v_*\right|^2 + \frac{I}{m}  + \frac{I_*}{m}  \right)^{\gamma/2}
\leq m^{\gamma/2} \left( \tf \right),
\end{multline*}
for any $\gamma \in (0,2]$. 
Therefore, we can take   $\dgu=1$. and $\egu=m^{\gamma/2}$.

\subsubsection{Model 2}\label{Sec Model 2 upper}
  We first estimate  model 2 in \eqref{model 2}  by 
\begin{multline*}
R^{\gamma/2} |u|^\gamma + (1-R)^{\gamma/2} \left(\frac{I+I_*}{m}\right)^{\gamma/2}
\\
 \leq \max\left\{ R, (1-R) \right\}^{\gamma/2} \left( \tfu \right),
\end{multline*}
and thus one possible choice is
\begin{equation}\label{dgu model 2 max}
\dgu=1, \quad \egu=\max\left\{ R, (1-R) \right\}^{\gamma/2}.
\end{equation}
Another more course estimate can be obtained by using $R\leq 1$, which leads to  the choice
\begin{equation*}
\dgu=\egu= 1.
\end{equation*}

\subsubsection{Model 3}\label{Sec Model 3 upper}
Finally,  model 3 from \eqref{model 3} is estimated as follows 
\begin{multline*}
 R^{\gamma/2} |v-v_*|^\gamma  + \left( r (1-R)\frac{I}{m} \right)^{\gamma/2} + \left( (1-r) (1-R)\frac{I_*}{m} \right)^{\gamma/2}
 \\
 \leq \max\left\{ R, (1-R) \right\}^{\gamma/2} \left(   |u|^\gamma + \left(\frac{I}{m}\right)^{\gamma/2} + \left(\frac{I_*}{m}\right)^{\gamma/2} \right)
\\ \leq 2^{1-\gamma/2}  \max\left\{ R, (1-R) \right\}^{\gamma/2} \left( \tfu \right).
\end{multline*}
since $\max\left\{1, r^{\gamma/2}, (1-r)^{\gamma/2} \right\} =1$. Thus, one possible choice is
\begin{equation}\label{dgl model 3 min}
\dgu=1, \quad \egu=  2^{1-\gamma/2}   \max\left\{ R, (1-R) \right\}^{\gamma/2}.
\end{equation}
Another possibility comes with the use of  $R \leq 1$,  which leads to 
\begin{equation*}
\dgu=1, \quad \egu=  2^{1-\gamma/2}.
\end{equation*}

\subsection{Calculation of the lower bounds} In this Section we provide various choices for functions  $\dgl$ and $\egl$ that appear in the bound from below of the transition function $\mathcal{B}$, from \eqref{trans prob rate ass}, for the three models of transition functions   introduced in \eqref{model 1},  \eqref{model 2} and \eqref{model 3}.

\subsubsection{Model 1}\label{Sec Model 1 lower}
For the model 1 from \eqref{model 1} we can estimate
\begin{equation*}
 \left( \frac{m}{4} \left|u\right|^2 + I + I_* \right)^{\gamma/2} \geq m^{\gamma/2} 2^{-(\gamma/2-1) }\left(  \tfu\right),
\end{equation*}
which implies  $\dgl=1$, and $\egl =  m^{\gamma/2} 2^{-(\gamma/2-1) }$.

\subsubsection{Model 2}\label{Sec Model 2 lower} For the model 2 introduced in \eqref{model 2}, we write
\begin{multline*}
 R^{\gamma/2} |u|^\gamma + (1-R)^{\gamma/2} \left(\frac{I+I_*}{m}\right)^{\gamma/2} \\ \geq \min\{ R^{\gamma/2}, (1-R)^{\gamma/2}   \} \left( \tfu \right)
\\
\geq R^{\gamma/2} (1-R)^{\gamma/2} \left( \tf \right),
\end{multline*}
and therefore  $\dgl=1$ and for $\egl$ we have two possible choices, 
\begin{equation}\label{dgl model 2 min}
\egl =  \min\{ R^{\gamma/2}, (1-R)^{\gamma/2}   \},
\end{equation}
or
\begin{equation}\label{dgl model 2}
\egl = R^{\gamma/2} (1-R)^{\gamma/2}.
\end{equation}

\subsubsection{Model 3}\label{Sec Model 3 lower} For the model 3 from \eqref{model 3} we estimate 
\begin{multline*}
 R^{\gamma/2} |u|^\gamma  + \left( r (1-R)\frac{I}{m} \right)^{\gamma/2} + \left( (1-r) (1-R)\frac{I_*}{m} \right)^{\gamma/2} \\ 
 \geq \min\left\{ R, (1-R)  \right\}^{\gamma/2} \left(  |u|^\gamma  + \left( r \frac{I}{m} \right)^{\gamma/2} + \left( (1-r) \frac{I_*}{m} \right)^{\gamma/2} \right)
\\
\geq \min\left\{ R, (1-R)  \right\}^{\gamma/2} \min\left\{ r, (1-r)  \right\}^{\gamma/2}  \left( \tfu \right) 
\\
\geq  r^{\gamma/2} (1-r)^{\gamma/2} R^{\gamma/2} (1-R)^{\gamma/2} \left( \tfu \right),
\end{multline*}
which allows two different choices
\begin{equation}\label{egl model 3 min} 
\dgl = \min\left\{ r, (1-r)  \right\}^{\gamma/2}, \quad \egl = \min\left\{ R, (1-R)  \right\}^{\gamma/2},
\end{equation}
or
\begin{equation}
\dgl = r^{\gamma/2} (1-r)^{\gamma/2}, \quad \egl = R^{\gamma/2} (1-R)^{\gamma/2}.
\end{equation}

\section{Computation of the constant $\mathcal{C}_{n}^1$ from Lemma \ref{lemma povzner}(The Polyatomic Compact Manifold Averaging Lemma)}\label{App Const}

This Section is devoted to the computation of the constant $\mathcal{C}_{n}^1$ from \eqref{Cinfty integral}, namely,
\begin{equation*}
\mathcal{C}_{n}^1 =\int_{0}^{1}  \int_{0}^1 \left( \max\left\{ \frac{1+ \mu }{2},  r \right\}\right)^{n} \mathrm{d}r \,  \mathrm{d}\mu.
\end{equation*}
We split integration domain of $r$ as $[0,1]= [0, \tfrac{1+\mu}{2}] \cup  [ \tfrac{1+\mu}{2}, 1] $. Then we use
\begin{equation*}
\max\left\{ \frac{1+ \mu }{2},  r \right\} = 
\begin{cases}
\frac{1+ \mu }{2}, \quad &r\in [0, \tfrac{1+\mu}{2}],\\[5pt]
r, \quad &r\in [\tfrac{1+\mu}{2},1]\\
\end{cases}
\end{equation*}
which yields
\begin{align}
\mathcal{C}_{n}^1 &=\int_{0}^{1}  \int_{0}^{\tfrac{1+\mu}{2}} \left( \frac{1+ \mu }{2}\right)^{n} \mathrm{d}r \,  \mathrm{d}\mu + \int_{0}^{1}  \int_{\tfrac{1+\mu}{2}}^1 r^{n} \mathrm{d}r \,  \mathrm{d}\mu \nonumber
\\
&=\int_{0}^{1}   \left( \frac{1+ \mu }{2}\right)^{n+1}  \mathrm{d}\mu + \frac{1}{n+1} \int_{0}^{1}  \left( 1- \left( \frac{1+ \mu }{2}\right)^{n+1}  \right) \mathrm{d}\mu \nonumber
\\
&= \frac{1}{n+1} +  \frac{2 n}{(n+1)(n+2)} \left( 1 - \left(\frac{1}{2}\right)^{n+2}  \right). \label{CinftyApp}
\end{align}
that completes the computation of $\mathcal{C}_{n}^1 $ as given in \eqref{Cinfty}.

\section{Some technical results }

\begin{lemma}[Polynomial inequality I] \label{binomial}
	Assume $p>1$, and let $n_p=\lfloor\frac{p+1}{2}\rfloor$. Then for all $x, y>0$, the following inequality holds
	\begin{equation*}
	\left(x+y\right)^p - x^p - y^p \leq \sum_{n=1}^{n_p} \left( \begin{matrix}
	p \\ n 
	\end{matrix} \right) \left(x^n y^{p-n} + x^{p-n}  y^n \right).
	\end{equation*}	
\end{lemma}

\begin{lemma}[Polynomial inequality II] \label{moment products}
	Let $b+1\leq a\leq \frac{p+1}{2}$. Then for any  $x, y\geq0$, 
	\begin{equation*}
	x^a y^{p-a} + x^{p-a} y^a \leq x^b y^{p-b} + x^{p-b} y^b.
	\end{equation*}	
\end{lemma}

\begin{lemma}[Interpolation inequality]
	Let $k=\alpha k_1 + (1-\alpha) k_2 $, $\alpha\in(0,1)$, $0<k_1\leq k \leq k_2$. Then for any $g \in L_{k}^1$
	\begin{equation}\label{interpolation inequality}
	\left\| g \right\|_{L_{k}^1} \leq \left\| g \right\|_{L_{k_1}^1}^{\alpha}  \left\| g \right\|_{L_{k_2}^1}^{1-\alpha}.
	\end{equation}
\end{lemma}

\section{Existence and Uniqueness Theory for ODE in Banach spaces}\label{Appendix exi and uni}

\begin{theorem}
	\label{Theorem general}
	Let $E:=(E,\left\| \cdot \right\|)$ be a Banach space, $\mathcal{S}$ be a bounded, convex and closed subset of $E$, and $\mathcal{Q}:\mathcal{S}\rightarrow E$ be an operator satisfying the following properties:
	\begin{itemize}
		\item[(a)] H\"{o}lder continuity condition
		\begin{equation*}
		\left\| \mathcal{Q}[u] - \mathcal{Q}[v] \right\| \leq C \left\| u-v \right\|^{\beta}, \ \beta \in (0,1), \ \forall u, v \in \mathcal{S};
		\end{equation*}
		\item[(b)] Sub-tangent condition
		\begin{equation*}
		\lim\limits_{h\rightarrow 0+} \frac{\text{dist}\left(u + h \mathcal{Q}[u], \mathcal{S} \right)}{h} =0, \ \forall u \in \mathcal{S};
		\end{equation*}
		\item[(c)] One-sided Lipschitz condition
		\begin{equation*}  
		\left[ \mathcal{Q}[u] - \mathcal{Q}[v], u - v \right] \leq C \left\| u-v \right\|, \ \forall u, v \in \mathcal{S},
		\end{equation*}
		where 	$\left[\varphi,\phi\right]=\lim_{h\rightarrow 0^-} h^{-1}\left(\left\| \phi + h \varphi \right\| - \left\| \phi \right\| \right)$.
	\end{itemize}
	Then the equation
	\begin{equation*}
	\begin{split}
	\partial_t  u = \mathcal{Q}[u], \ \text{for} \ t\in(0,\infty), \ \text{with initial data}  \
	u(0)= u_0 \ \text{in} \ \mathcal{S},
	\end{split}
	\end{equation*}
	has a unique solution in $C([0,\infty),\mathcal{S})\cap C^1((0,\infty),E)$.
\end{theorem}
The proof of this Theorem on ODE flows on Banach spaces  can be found in \cite{Alonso-IG-BAMS}.
\begin{remark}\label{Lip remark}
	In Section \ref{Section Ex Uni proof}, we will concentrate on $E:=L_{\ho}^1$. Therefore, for one-sided Lipschitz condition, we will use the following inequality,   
	$$
	\left[\varphi,\phi\right]\leq \int_{ {\mathbb{R}^{3+}}}  \varphi(v,I) \, \text{sign}(\phi(v,I))\left\langle v,I \right\rangle^2 \mathrm{d}I \,\mathrm{d}v,
	$$
	as pointed out in \cite{Alonso-IG-BAMS}.
\end{remark}

\medskip
\medskip


\begin{thebibliography}{99}

	
\bibitem{Alexeev}
\newblock B. V. Alexeev, A. Chikhaoui, and I. T. Grushin
\newblock	Application of the generalized Chapman-Enskog method to the transport-coefficient calculation in a reacting gas mixture,
\newblock\emph{Phys. Rev. E}, \textbf{49}(4): 2809-–2825, 1994.

\bibitem{Alonso}
\newblock	R. Alonso,
\newblock Emergence of exponentially weighted $L^p$-norms and Sobolev regularity for the Boltzmann equation,
\newblock \emph{Commun. Part. Diff. Eq.}, \textbf{44}(5): 416--446, 2019.


\bibitem{AlonsoLods18}{
	\newblock	R. Alonso, V. Bagland, Y. Cheng, and B. Lods, 
	\newblock { One-dimensional dissipative Boltzmann equation: measure
		solutions, cooling rate, and self-similar profile}, 
	\newblock \emph{SIAM J. Math. Anal.}, \textbf{50}(1): 1278--1321, 2018.
	
} 

\bibitem{Gamba13}{
	\newblock	R. Alonso, J. A. Ca\~nizo, I. M. Gamba, and C. Mouhot, 
	\newblock { A new approach to the creation and propagation of exponential
		moments in the Boltzmann equation}, 
	\newblock \emph{Comm. Partial Differential Equations}, 38 (1): 155--169, 2013.
}


	
\bibitem{Alonso-IG-BAMS}
\newblock R. J. Alonso and I. M. Gamba,
\newblock  Revisiting the Cauchy problem for the Boltzmann equation for hard
potentials with integrable cross section: from generation of moments to
propagation of $L^{\infty}$ bounds.
\newblock preprint, 2018.

\bibitem{IG-A-M}
\newblock  R. J. Alonso,  I. M. Gamba, and M. Pavi\'c-\v Coli\'c,
\newblock Propagation of weighted Banach space regularity to solutions of the Boltzmann equations for polyatomic gases,
\newblock work in progress, 2020.

\bibitem{AlonsoLods}
R. Alonso and B. Lods,
Free cooling and high-energy tails of granular gases with variable restitution coefficient, Siam J. Math. Anal, 42 (6), 2499-2538, 2010.

\bibitem{IG-Alonso-Tran-pre}
\newblock  R. J. Alonso, I. M. Gamba and M. B.Tran,
\newblock The Cauchy problem and BEC stability for the quantum Boltzmann-Condensation system for bosons at very low temperature, 
\newblock preprint, ArXiv 1609.07467.v3, 2018.



	
\bibitem{Soares-Salvarani}
\newblock B. Anwasia, M. Bisi, F. Salvarani, A. J. Soares,
\newblock On the Maxwell-Stefan diffusion limit for a reactive mixture of polyatomic gases in non-isothermal setting,
\newblock   \emph{Kinet. Relat. Models}, \textbf{13}(1), 63 -- 95, 2020.
	
\bibitem{Rugg-6}
\newblock 	T. Arima, T. Ruggeri, M. Sugiyama, and S. Taniguchi,
\newblock  Non-linear extended thermodynamics of real gases with 6 fields
\newblock\emph{Int. J. Non Linear Mech.}, 72: 6-–15, 2015.
	
\bibitem{LD-Bisi} 
\newblock  C. Baranger, M. Bisi, S. Brull and L.  Desvillettes, 
\newblock On the Chapman-Enskog asymptotics for a mixture of monoatimic and polyatomic rarefied gases,
\newblock \emph{Kinet. Relat. Models}, \textbf{11}(4), 821 --858, 2018.

\bibitem{Bisi-multi-1}
\newblock M. Bisi, G. Martal\`{o}, G. Spiga, 
\newblock Multi-temperature hydrodynamic limit from ki- netic theory in a mixture of rarefied gases, 
\newblock \emph{Acta Appl. Math.} \textbf{122}, 37–-51, 2012. 


\bibitem{Bisi-multi-2}
\newblock M. Bisi, G. Martal\`{o}, G. Spiga, 
\newblock Multi-temperature Euler hydrodynamics for a reacting gas from a kinetic approach to rarefied mixtures with resonant collisions, 
\newblock \emph{Europhys. Lett.} \textbf{95}, 55002, 2011. 

\bibitem{Soares-Bisi}
\newblock M. Bisi, R. Monaco and A. J. Soares,
\newblock A BGK model for reactive mixtures of polyatomic gases with continuous internal energy,
\newblock \emph{J. Phys. A}, \textbf{51}(12), 125501, 2018.

\bibitem{Rugg-Bisi-6} 
\newblock M. Bisi, T. Ruggeri and G. Spiga, 
\newblock Dynamical pressure in a polyatomic gas: Interplay between kinetic theory and Extended Thermodynamics,
\newblock \emph{Kinet. Relat. Models}, \textbf{11},  71-–95, 2018.


\bibitem{Bob97} 
\newblock A. V. Bobylev,
\newblock {Moment inequalities for the Boltzmann equation and applications to spatially homogeneous problems},
\newblock \emph{J. Statist. Phys.}, {88}: 1183--1214, 1997.

\bibitem{BobGamba17}  
\newblock A. V. Bobylev and I. M. Gamba,
\newblock Upper Maxwellian bounds for the Boltzmann equation with pseudo-Maxwell molecules.
\newblock \emph{Kinet. Relat. Models}, \textbf{10},  573--585, 2017.
	
\bibitem{GambaBobPanf04} 
\newblock A. V. Bobylev, I. M. Gamba, and V. A. Panferov,
\newblock Moment inequalities and high-energy tails for Boltzmann equations with inelastic interactions,
\newblock \emph{J. Statist. Phys.}, \textbf{116}, 1651--1682, 2004.
	
	
\bibitem{LD-Bourgat94} 
\newblock J.-F. Bourgat, L. Desvillettes, P. Le Tallec, and B. Perthame.
\newblock {Microreversible collisions for polyatomic gases and Boltzmann's theorem}, 
\newblock \emph{Eur. J. Mech., B/Fluids}, 13(2): 237--254, 1994.

\bibitem{Cha-Cow} 
\newblock S. Chapman and T.G. Cowling
\newblock The Mathematical Theory of Non-Uniform Gases, 3rd edn. 
\newblock Cambridge University Press, Cambridge, 1990.

\bibitem{Della}{
	\newblock	S. Dellacherie, 
	\newblock { On the Wang Chang-Uhlenbeck equations}, 
	\newblock \emph{Discrete Cont Dyn-B}, \textbf{3}(2): 229--253, 2003.
	
} 


\bibitem{Des93} 
\newblock L. Desvillettes,
\newblock Some applications of the method of moments for the homogeneous Boltzmann and Kac equations,
\newblock \emph{Arch. Rational Mech. Anal.}, \textbf{123}, 387--404, 1993.

\bibitem{LD-Toulouse} 
\newblock L. Desvillettes,
\newblock Sur un mod\`{e}le de type Borgnakke–Larsen conduisant \`{a} des lois d’energie non-lin\'{e}aires en temp\'{e}rature pour les gaz parfaits
polyatomiques,
\newblock \emph{ Ann. Fac. Sci. Toulouse Math.}, \textbf{6}(2),  257--262, 1997.


\bibitem{DesMonSalv}
\newblock L. Desvillettes, R. Monaco and F. Salvarani
\newblock A kinetic model allowing to obtain the energy law of polytropic gases in the presence of chemical reactions,
\newblock \emph{Eur. J. Mech. B Fluids}, \textbf{24}, 219--236,2005.



\bibitem{MPC-Dj-S}
\newblock V. \DJ{}or\dj{}i\'c,  M. Pavi\'c-\v Coli\'c, and N. Spasojevi\'c,
\newblock {Kinetic and macroscopic modelling of a polytropic gas}, 
\newblock ArXiv:2004.12225, 2020.


\bibitem{IG-P-C} 
\newblock I. M. Gamba, and M. Pavi\'c-\v Coli\'c,
\newblock {On existence and uniqueness to homogeneous Boltzmann flows of monatomic gas mixtures}, 
\newblock \emph{ Arch. Ration. Mech. Anal.}, \textbf{235}, 723--781, 2020.


\bibitem{GambaPanfVil09}  
\newblock I. M. Gamba, V. Panferov, and C. Villani,
\newblock Upper Maxwellian bounds for the spatially homogeneous Boltzmann equation,
\newblock \emph{ Arch. Ration. Mech. Anal.}, \textbf{194} (2009), 253--282.



\bibitem{IG-Smith-Tran}  
\newblock I. M. Gamba, L. Smith and M.B. Tran 
\newblock On the wave turbulence theory for stratified flows in the ocean,
\newblock Mathematical Models and Methods in Applied Sciences, \textbf{30}(1), 105--137, 2020.


\bibitem{Gio} 
\newblock V. Giovangigli,
\newblock Multicomponent Flow Modeling,
\newblock MESST Series, Birkhauser Boston, 1999.

\bibitem{Aoki} 
\newblock S. Kosuge and K. Aoki
\newblock Shock-wave structure for a polyatomic gas with large bulk viscosity
\newblock \emph{Phys. Rev. Fluids}, \textbf{3} 023401, 2018.


\bibitem{LuMouhot12} 
\newblock X. Lu and C. Mouhot,
\newblock On measure solutions of the Boltzmann equation, part I: moment production and stability estimates,
\newblock \emph{J. Differential Equations}, \textbf{252}, 3305--3363, 2012.

\bibitem{Martin} R. H. Martin,
\newblock {Nonlinear operators and differential equations in Banach spaces.}
\newblock \emph{Pure and Applied Mathematics. Wiley-Interscience}, 1976.


\bibitem{Mouhot06} 
\newblock C. Mouhot,
\newblock Rate of convergence to equilibrium for the spatially homogeneous Boltzmann equation with hard potentials,
\newblock \emph{ Comm. Math. Phys.}, \textbf{261}, 629--672, 2006.




\bibitem{MPC} 
\newblock M. Pavi\'c-\v Coli\'c
\newblock Multi-velocity and multi-temperature model of the mixture of polyatomic gases issuing from kinetic theory
\newblock \emph{Physics Letters A}, \textbf{383}(24), 2829--2835, 2019.

\bibitem{MPC-Madj-Sim}
\newblock M. Pavi\'c-\v Coli\'c, D. Madjarevic  and S. Simi\'c,
\newblock Polyatomic gases with dynamic pressure: Kinetic non-linear closure and the shock structure,
\newblock   \emph{International Journal of Non-Linear Mechanics}      \textbf{92} 160–175, 2017.

\bibitem{MPC-Rugg-Sim} 
\newblock M. Pavi\'c, T. Ruggeri and S. Simi\'c,
\newblock Maximum entropy principle for rarefied polyatomic gases, 
\newblock \emph{Physica A}, \textbf{392}, 1302--1317, 2013.


\bibitem{MPC-Sim} 
\newblock M. Pavi\'c and S. Simi\'c,
\newblock Moment equations for polyatomic gases, 
\newblock \emph{Acta Appl. Math.}, \textbf{132}(1), 469--482, 2014.

\bibitem{PT}
\newblock M. Pavi\'c-\v Coli\'c, and M. Taskovi\'c,
\newblock Propagation of stretched exponential moments for the Kac equation and Boltzmann equation with Maxwell molecules,
\newblock  \emph{Kinet. Relat. Models}, \textbf{11}(3),  597--613, 2018.



\bibitem{Rugger-MEP} 
\newblock T. Ruggeri, 
\newblock Non-linear maximum entropy principle for a polyatomic gas subject to the dynamic pressure,
\newblock \emph{Bull. Inst. Math., Acad. Sin. (New Ser.)}, \textbf{11}(1), 1--22, 2016.

\bibitem{Rugg-Poly} 
\newblock T. Ruggeri and M. Sugiyama, 
\newblock Rational Extended Thermodynamics beyond the Monatomic Gas, 
\newblock Springer, New York, 2015.


\bibitem{MPC-Madj-Sim-Parma}
\newblock  S. Simi\'c, M. Pavi\'c-\v Coli\'c and D. Madjarevic, 
\newblock Non-equilibrium mixtures of gases: Modelling and computation,
\newblock   \emph{Rivista di Matematica della Universita di Parma}      \textbf{6}(1) 135–214, 2015.



\bibitem{GambaTask18}{
	\newblock 	M. Taskovi\'c, R. J. Alonso,  I. M. Gamba, and N. Pavlovi\'c,
	\newblock {On Mittag-Leffler moments for the Boltzmann equation for hard potentials without cutoff},
	\newblock \emph{SIAM J. Math. Anal.}, {50(1)}: 834--869, 2018.
}


\bibitem{W-C} 
\newblock C.S. Wang Chang, G.E. Uhlenbeck and J.  de Boer,
\newblock The heat conductivity and viscosity of polyatomic gases,
\newblock In: Studies in Statistical Mechanics, vol. II, North-Holland, Amsterdam, 243–268, 1964.

\bibitem{wennberg97}
\newblock B.~Wennberg,
\newblock Entropy dissipation and moment production for the	{B}oltzmann equation,
\newblock \emph{Jour. Statist. Phys.}, \textbf{86}(5-6), 1053--1066, 1997.


\end{thebibliography}
\end{document}